\documentclass[journal,letter,onecolumn]{IEEEtran} %
\usepackage[utf8]{inputenc} 
\usepackage[T1]{fontenc}
\usepackage{url}
\usepackage{ifthen}
\usepackage[cmex10]{amsmath} %
\usepackage{amssymb,amsfonts,amsmath,amsthm,amscd,dsfont,mathrsfs,bbm}
\usepackage{amsfonts}
\usepackage[mathcal]{euscript} 
\usepackage{amssymb}
\usepackage{graphics,subfig,epsfig}
\usepackage{verbatim}
\usepackage{cite,calc} 
\usepackage{algorithm,algorithmicx,algpseudocode}
\usepackage{subdepth}  %
\usepackage{bm} %
\usepackage{xcolor}
\usepackage{hyperref} %
\usepackage{multirow}
\allowdisplaybreaks

\DeclareMathAlphabet{\mathbbb}{U}{bbold}{m}{n}  %

\theoremstyle{definition}
\newtheorem{definition}{Definition}%
\theoremstyle{plain}
\newtheorem{theorem}{Theorem}%
\theoremstyle{plain}
\theoremstyle{plain}
\newtheorem{lemma}{Lemma}%

\newtheorem{proposition}{Proposition}
\newtheorem{corollary}{Corollary}

\newtheorem{remark}{Remark}

\newcommand{\ds}{\displaystyle}
\newcommand{\secref}[1]{Section~\ref{#1}}
\newcommand{\appref}[1]{Appendix~\ref{#1}}

\newcommand{\figref}[1]{Fig.~\ref{#1}}
\newcommand{\lemref}[1]{Lemma~\ref{#1}}
\newcommand{\thmref}[1]{Theorem~\ref{#1}}
\newcommand{\propref}[1]{Proposition~\ref{#1}}

\newcommand{\defref}[1]{Definition~\ref{#1}}
\newcommand{\algoref}[1]{Algorithm~\ref{#1}}

\DeclareMathOperator*{\argmin}{arg\,min}

\DeclareMathOperator*{\tr}{tr}
\DeclareMathOperator*{\sgn}{sgn}

\DeclareMathOperator*{\diag}{diag}

\newcommand{\eps}{\epsilon}

\newcommand{\fb}{\bar{f}}
\newcommand{\fh}{\hat{f}}

\newcommand{\gb}{\bar{g}}
\newcommand{\gh}{\hat{g}}

\newcommand{\X}{{\mathcal{X}}}

\newcommand{\Y}{{\mathcal{Y}}}

\newcommand{\cN}{\mathcal{N}}

\newcommand{\cS}{\mathcal{S}}
\newcommand{\cI}{\mathcal{I}}

\newcommand{\cX}{\X}

\newcommand{\cY}{\Y}

\newcommand{\ev}{{\mathbb{E}}}

\newcommand{\E}[1]{\ev\left[#1\right]}

\renewcommand{\P}{{\mathbb{P}}}

\renewcommand{\vec}[1]{\underline{#1}}

\newcommand{\norm}[1]{\|#1\|}

\newcommand{\ip}[2]{\langle#1,#2\rangle}

\newcommand{\ub}{\underline{\beta}}

\newcommand{\psib}{\bar{\psi}{}}
\newcommand{\psih}{\hat{\psi}{}}

\newcommand{\phib}{\bar{\phi}{}}
\newcommand{\phih}{\hat{\phi}{}}

\newcommand{\frob}[1]{\|#1\|_\mathrm{F}}
\newcommand{\bfrob}[1]{\bigl\|#1\bigr\|_\mathrm{F}}
\newcommand{\bbfrob}[1]{\left\|#1\right\|_\mathrm{F}}

\newcommand{\spectral}[1]{\|#1\|_\mathrm{s}}
\newcommand{\bspectral}[1]{\bigl\|#1\bigr\|_\mathrm{s}}
\newcommand{\bbspectral}[1]{\left\|#1\right\|_\mathrm{s}}
\newcommand{\trop}[1]{\tr\left\{#1\right\}}
\newcommand{\defeq}{\triangleq}

\newcommand{\nbhd}{\cN}

\newcommand{\Ph}{\hat{P}}

\newcommand{\dtm}{\mathbf{B}}
\newcommand{\dtmh}{\hat{\dtm}}

\newcommand{\dtmt}{\tilde{\dtm}}

\newcommand{\T}{\mathrm{T}}
\newcommand{\Pt}{\tilde{P}}

\newcommand{\bone}{\mathbf{1}}

\newcommand{\Phib}{{\bold{\Phi}}}
\newcommand{\Phih}{\hat{\Phib}}
\newcommand{\Psib}{{\bold{\Psi}}}
\newcommand{\Psih}{\hat{\Psib}}
\newcommand{\Ab}{\mathbf{A}}
\newcommand{\Vb}{\mathbf{V}}
\newcommand{\vb}{\bm{v}}
\newcommand{\Lambdab}{\mathbf{\Lambda}}

\newcommand{\varthetab}{{\bm{\vartheta}}} 
\newcommand{\varthetabb}{\bar{\varthetab}} 
\newcommand{\phibh}{\hat{\phib}}
\newcommand{\Mbh}{\hat{\Mb}}

\newcommand{\cIb}{\mathcal{I}^{\mathsf{c}}}
\newcommand{\Ib}{\mathbf{I}}
\newcommand{\Cb}{\mathbf{C}}
\newcommand{\Mb}{\mathbf{M}}
\newcommand{\Pb}{\mathbf{P}}
\newcommand{\Qb}{\mathbf{Q}}
\newcommand{\qb}{\bm{q}}
\newcommand{\Xib}{\mathbf{\Xi}}
\renewcommand{\phib}{\bm{\phi}}
\newcommand{\imate}{\Gamma} %
\newcommand{\imat}{\bm{\imate}} %
\newcommand{\dtmdmate}{\Xi} %
\newcommand{\dtmdmat}{\bm{\dtmdmate}} %

\newcommand{\Gb}{\mathbf{G}}
\newcommand{\Jb}{\mathbf{J}}
\newcommand{\Lb}{\mathbf{L}}
\newcommand{\Lbh}{\bar{\Lb}}
\newcommand{\thetab}{\bm{\theta}}
\renewcommand{\psib}{\bm{\psi}}
\newcommand{\eqspace}{}%
\newcommand{\zerob}{\bm{0}}
\newcommand{\vbh}{\hat{\vb}}
\newcommand{\Vbh}{\hat{\Vb}}
\newcommand{\Ob}{\mathbf{O}}
\newcommand{\Wb}{\mathbf{W}}
\DeclareMathOperator{\vecop}{vec}
\renewcommand{\vec}{\vecop}
\renewcommand{\ub}{\bm{u}}

\newcommand{\imatem}{\zeta} %
\newcommand{\imatm}{\bm{\imatem}} %
\newcommand{\imatej}{\imate} %
\newcommand{\imatj}{\mathbf{\imatej}} %
\newcommand{\imates}{\Upsilon} %
\newcommand{\imats}{\mathbf{\imates}} %
\newcommand{\Ub}{\mathbf{U}}
\newcommand{\F}{\mathrm{F}}
\newcommand{\Er}{\mathrm{E}}
\newcommand{\Erh}{\hat{\Er}}
\newcommand{\Erb}{\bar{\Er}}

\newcommand{\alphah}{\bar{\alpha}}
\newcommand{\betah}{\bar{\beta}}

\newcommand{\Gbh}{\bar{\Gb}}
\newcommand{\Jbh}{\bar{\Jb}}
\newcommand{\convec}{\bm{\varsigma}} %
\newcommand{\nub}{\bm{\nu}} %
\newcommand{\Phibb}{\bar{\Phib}}
\newcommand{\Psibb}{\bar{\Psib}}
\newcommand{\psibb}{\bar{\psib}}
\newcommand{\phibb}{\bar{\phib}}
\newcommand{\psibar}{\bar{\psi}}
\newcommand{\phibar}{\bar{\phi}}
\newcommand{\nbhdbar}{\bar{\nbhd}}

\newcommand{\varphib}{\bm{\varphi}} %
\newcommand{\varphibb}{\bar{\varphib}} %
\newcommand{\ubb}{\bar{\ub}} %

\renewcommand{\Pt}{\bar{P}}
\newcommand{\Phibt}{\bar{\Phib}}
\newcommand{\dtmtt}{\bar{\dtm}}
\newcommand{\cSt}{\bar{\cS}}
\newcommand{\dtmdmatet}{\bar{\Xi}} %
\newcommand{\dtmdmatt}{\bar{\bm{\dtmdmate}}} %

\newcommand{\dtmeh}{\hat{B}}
\newcommand{\dtmet}{\tilde{B}}
\newcommand{\dtmett}{\bar{B}}

\newcommand{\Bt}{\tilde{B}}
\newcommand{\Sigmab}{\bm{\Sigma}}
\newcommand{\reply}[1]{{#1}}

\definecolor{electricpurple}{rgb}{0.75, 0.0, 1.0}

\begin{document}
\title{On the Sample Complexity of HGR Maximal Correlation Functions\reply{ for Large Datasets} }

\author{
Shao-Lun Huang,~\IEEEmembership{Member,~IEEE}
and~Xiangxiang Xu,~\IEEEmembership{Graduate Student Member,~IEEE}%
\thanks{This paper was presented in part at the Inform.\ Theory Workshop (ITW-2019),
Visby, Sweden, Aug. 2019 and at Allerton Conf. Commun., Contr., Computing (Allerton-2019), Monticello, IL, Sep. 2019.}
\thanks{S.-L. Huang is with the Data Science and Information Technology Research
Center, Tsinghua--Berkeley Shenzhen Institute, Shenzhen, China (e-mail:
shaolun.huang@sz.tsinghua.edu.cn).}%
\thanks{X. Xu is with the Department of Electronic Engineering, Tsinghua University, Beijing, China (e-mail: xuxx14@mails.tsinghua.edu.cn).}
}

\maketitle

\begin{abstract}
  The Hirschfeld--Gebelein--R\'{e}nyi (HGR) maximal correlation and the corresponding functions have been shown useful in many machine learning scenarios. In this paper, we study the sample complexity of estimating the HGR maximal correlation functions by the alternating conditional expectation (ACE) algorithm using training samples from large datasets. Specifically, we develop a mathematical framework to characterize the learning errors between the maximal correlation functions computed from the true distribution, and the functions estimated from the ACE algorithm. For both supervised and semi-supervised learning scenarios, we establish the analytical expressions for the error exponents of the learning errors. Furthermore, we demonstrate that for large datasets, the upper bounds for the sample complexity of learning the HGR maximal correlation functions by the ACE algorithm can be expressed using the established error exponents.  
 Moreover, with our theoretical results, we investigate the sampling strategy for different types of samples in semi-supervised learning with a total sampling budget constraint, and an optimal sampling strategy is developed to maximize the error exponent of the learning error. Finally, the numerical simulations are presented to support our theoretical results. 

\end{abstract}


\IEEEpeerreviewmaketitle

\section{Introduction} \label{sec:1}

Learning informative and generalizable representations of data is a crucial issue in machine learning \cite{bengio2013representation}. To measure the correlation and select informative features, the Hirschfeld--Gebelein--R\'{e}nyi (HGR) maximal correlation~\cite{H35, H41, RenyiCorrelation}  is a normalized measure of the dependence between two random variables and has been widely applied as an information metric to study inference and learning problems~\cite{bell1962mutual, ahlswede1976spreading, lopez2013randomized}. \reply{Specifically, given a pair of jointly distributed discrete random variables $X,Y$ over finite alphabets $\cX, \cY$, their HGR maximal correlation $\rho(X; Y)$ is defined as
\begin{align} \label{eq:HGR:org}
\rho(X;Y) \triangleq \max_{f\colon \cX \mapsto \mathbb{R}, \ g\colon \cY \mapsto \mathbb{R}} \E{f(X)g(Y)},
\end{align}
where the maximum is taken over all functions $f, g$ with zero mean and unit variance. Therefore, the HGR maximal correlation characterizes the correlation between the most correlated function mappings of $X$ and $Y$, and the optimal functions $f, g$ that achieve the maximal correlation essentially extract the most correlated aspects between $X$ and $Y$. Recently, the HGR maximal correlation has been further generalized to consider the correlation in the $k$-dimensional functional spaces by defining} \cite{makur2015efficient}
\begin{align} \label{eq:HGR}
\rho_k(X;Y) \triangleq \max_{\substack{ f\colon \cX \mapsto \mathbb{R}^k, \ g\colon \cY \mapsto \mathbb{R}^k  \\ {\E{f(X)} = \E{g(Y)} = \bold{0}} \\ {\E{f(X)f^{\T}(X)} = \E{g(Y)g^{\T}(Y)} = \mathbf{I}}}} \E{f^{\T}(X)g(Y)},
\end{align}
of which the special case with $k = 1$ corresponds to the original problem \eqref{eq:HGR:org}. In particular, the optimal functions $f^*,g^*$ maximizing~\eqref{eq:HGR}, referred to as the \emph{maximal correlation functions}, have been shown to take important roles in statistics \cite{ahlswede1976spreading}, information theory~\cite{makur2015efficient, huang2017information}, machine learning \cite{razaviyayn2015discrete, wang2019efficient, xu2020maximal}, and especially in interpreting deep neural networks \cite{HXZW2019}. For example, in machine learning scenarios, the variable $Y$ can be viewed as the label, and $X$ is the data variable that is used to infer or predict about attributes of $Y$. Then, $f^*$ can be illustrated as the optimal feature to predict $Y$, with $g^*$ being the corresponding weights \cite{xu2020maximal}. Therefore, efficiently and effectively computing maximal correlation functions from data is important in information theory and machine learning.

In this paper, we study the sample complexity of estimating the maximal correlation functions from a sequence of $n$ training samples $(x_1, y_1), \ldots, (x_n, y_n)$, i.i.d. generated from the (unknown) true distribution $P_{XY}$, by the widely adopted alternating conditional expectation (ACE) algorithm~\cite{ACE}. Mathematically, the ACE algorithm computes the maximal correlation functions over the empirical joint distribution $\Ph_{XY}$ of the training samples. Therefore, there exists a learning error between the true maximal correlation functions and the computed functions due to the i.i.d. sampling process. To quantify this learning error, for functions computed from the ACE algorithm, we apply the \emph{H-score} introduced by~\cite{HXZW2019} as the performance metric. It has been shown that the H-score of a function indicates the performance of employing that function as the input feature to the softmax regression~\cite{HXZW2019}, and hence is a meaningful information metric in machine learning applications. Then, we study the large deviation property of the H-score for the functions computed from the ACE algorithm, in which we characterize the error exponent of the learning error in the asymptotic regime, i.e., when $n$ tends to infinity. In particular, we establish the analytical solutions for this error exponent expressed by the true distribution $P_{XY}$. Furthermore, for large datasets, we demonstrate an upper bound for the sample complexity of learning the maximal correlation functions, which can be expressed using the established error exponent.
Our results %
also provide insights in designing the dimensionality $k$ for the selected features to effectively extract correlation structures among data variables in machine learning problems.

In addition, we investigate the sample complexity of the maximal correlation functions in the semi-supervised learning \cite{zhu2005semi} scenario, in which not only the labeled samples, but also a sequence of $m$ unlabeled training samples $x_{n+1}, \ldots , x_{n+m}$ is observed. In this case, the empirical joint distribution can in general be learned more accurately, since the marginal distribution of $X$ is trained better due to the unlabeled samples. Thus, the sample complexity is expected to be improved. To quantify this performance gain, we first propose a generalized ACE algorithm to deal with the unlabeled training samples, and then study the sample complexity of estimating the maximal correlation functions from the generalized ACE algorithm. As in the supervised case, we develop the closed form expressions for the error exponent of the learning error, and demonstrate the performance gain from the unlabeled samples. In addition, the theoretical results are applied to study the optimal sampling strategy, when the labeled and unlabeled samples have different acquiring costs. In particular, we formulate an optimization problem to maximize the error exponent of estimating the maximal correlation functions from different types of samples, subjected to a total budget constraint on acquiring these samples. The solution of this optimization problem then illustrates the optimal design of selecting different types of samples in semi-supervised learning problems. Finally, our theoretical results are supported by numerical simulations.

The rest of this paper is organized as follows. In \secref{sec:pf}, we formulate the sample complexity problem of the maximal correlation functions and define the corresponding error exponent of the test error, and the mathematical framework for computing this error exponent is presented in \secref{sec:mpa}. In \secref{sec:sup}, we establish the analytical expression for the error exponent in the supervised learning scenario, in which the number of required samples for computing the maximal correlation functions on large datasets is provided. Then, similar analyses of the error exponent and the number of samples required for the semi-supervised learning are presented in \secref{sec:semi}, in which we develop the optimal sampling strategy for semi-supervised learning with a sampling budget constraint. Finally, the numerical simulations are presented in \secref{sec:ns} to support our theoretical results.

\section{Problem Formulation} \label{sec:pf}
We commence by briefly introducing the singular value decomposition (SVD) structure of the HGR maximal correlation problem and the ACE algorithm for computing the maximal correlation functions. Given the joint distribution\footnote{We assume that the true marginal distributions $P_X(x) > 0$, and $P_{Y}(y) > 0$, for all $x,y$, since otherwise we can remove the symbols with probability 0 from the alphabets.} $P_{XY}$, the HGR maximal correlation \eqref{eq:HGR:org} is known to be the second largest singular value of the matrix $\dtm \in \mathbb{R}^{|\cY|\times |\cX|}$, also referred to as the \emph{divergence transition matrix} (DTM) \cite{huang2012linear}, whose entries are \cite{witsenhausen1975sequences}
\begin{align*}
  B(y, x) = \frac{P_{XY}(x, y)}{\sqrt{P_X(x)P_Y(y)}},
\end{align*}
and the maximal correlation functions $f^*, g^*$ are chosen such that the vectors\footnote{In our development, we may simply take the alphabets $\cX = \{1 , 2 , \ldots , |\cX| \}$ and $\cY = \{1 , 2 , \ldots , |\cY| \}$, which corresponds to some given alphabet orders of random variables $X$ and $Y$.}
\begin{align*}
  \left[f^*(1)\sqrt{P_{X}(1)}, \dots, f^*(|\cX|)\sqrt{P_{X}(|\cX|)}\right]^\T \text{ and } \left[g^*(1)\sqrt{P_{Y}(1)}, \dots, g^*(|\cY|)\sqrt{P_{Y}(|\cY|)}\right]^\T
\end{align*}
are the right and left singular vector of $\dtm$ associated with the second largest singular value, respectively. It can be shown that he largest singular value of $\dtm$ is $1$, with the corresponding right and left singular vectors being $  \left[\sqrt{P_{X}(1)}, \dots, \sqrt{P_{X}(|\cX|)}\right]^\T$ and $\left[\sqrt{P_{Y}(1)}, \dots, \sqrt{P_{Y}(|\cY|)}\right]^\T$, and thus it would be more convenient to subtract the top singular mode and introduce the matrix $\dtmt \in \mathbb{R}^{|\cY|\times|\cX|}$ with entries
\begin{align} \label{eq:dtmt}
  \Bt(y, x) = B(y, x) - \sqrt{P_X(x)P_Y(y)} = \frac{P_{XY}(x, y) - P_X(x)P_Y(y)}{\sqrt{P_X(x)P_Y(y)}},
\end{align} 
referred to as the \emph{canonical dependence matrix} (CDM) \cite{huang2019universal}. Then, the HGR maximal correlation and the 1-dimensional maximal correlation functions can be represented by the largest singular value and the corresponding singular vectors of the $\dtmt$. It can be shown that the generalized HGR maximal correlation \eqref{eq:HGR} has retained this SVD structure \cite{huang2019universal}. Specifically, the maximal correlation $\rho_k(X; Y)$ is the sum of the largest $k$ singular values (i.e., the Ky Fan $k$-norm) of $\dtmt$, and the maximal correlation functions $f^*, g^*$ correspond to its top $k$ right and left singular vectors, respectively.

\begin{algorithm}[!t]
  \caption{\reply{The ACE Algorithm for computing $k$-dimensional maximal correlation functions}}
\reply{
  \begin{algorithmic}[1]
    \State \textbf{Input}: $n$ i.i.d. samples $(x_1, y_1), \dots, (x_n, y_n)$.
    \State \textbf{Initialize}: random choose zero-mean functions $\fh\colon \cX \mapsto \mathbb{R}^k$ and $\gh \colon \cY \mapsto \mathbb{R}^k$.
    \Repeat
    \State $\ds\fh(x) \gets    \E{\Lambdab_{\gh}^{-1}\gh(Y) \,\middle|\, X=x}\ \text{for all } x \in \cX$\label{alg:ace:line:begin}
    \vspace{.1em}
    \State $\ds\gh(y) \gets   \E{ \Lambdab_{\fh}^{-1}\fh(X) \,\middle|\, Y=y}\ \text{for all } y \in \cY$
    \vspace{.1em}
    \Until $\ds 2\E{\fh^\T(X)\gh(Y)} - \trop{\Lambdab_{\fh}\Lambdab_{\gh}}$
    stops increasing \label{alg:ace:line:end}
    \vspace{.1em}
    \State Whiten:
    \State \qquad$\fh(x) \gets \Lambdab_{\fh}^{-1/2}\fh(x)\ \text{ for all } x \in \cX$ \label{alg:ace:line:norm:f}
    \State $\gh(y) \gets  \E{ \fh(X) \,\middle|\, Y=y}\ \text{ for all } y \in \cY$ \label{alg:ace:line:g}
    \vspace{.1em}
    \State Whiten:
    \State \qquad$\gh(y) \gets \Lambdab_{\gh}^{-1/2}\gh(y)\ \text{ for all } y \in \cX$ \label{alg:ace:line:norm:g}
    \State \textbf{Output:} $\fh$, $\gh$ and the maximal correlation $\E{\fh^\T(X)\gh(Y)}$.
  \end{algorithmic}
  \label{alg:ace}
}
\end{algorithm}

\reply{In practical learning applications, since $X$ and $Y$ can have large alphabets or even be continuous, the matrix $\dtmt$ cannot be easily estimated from data samples for computing the maximal correlation functions. 
  Instead, we can use the  ACE algorithm \cite{ACE} to iteratively compute the maximal correlation functions, which is mathematically equivalent to the power method on $\dtmt$.} In particular, given a sequence of $n$ training samples $(x_1, y_1), \ldots, (x_n, y_n)$, i.i.d. generated from the joint distribution $P_{XY}$, the ACE algorithm estimates the $k$-dimensional maximal correlation functions of~\eqref{eq:HGR} can be summarized as \algoref{alg:ace}\footnote{There are also other variants of ACE algorithm for computing $k$-dimensional maximal correlation functions using different numerical techniques for normalization, see, e.g., \cite[Algorithm 3]{makur2015efficient} and \cite[Algorithm 1]{huang2019universal}.}, where the expectations are taken over the empirical distributions $\Ph_{XY}$, or conditional distributions $\Ph_{Y|X}$ and $\Ph_{X|Y}$ from the training samples. In addition, $\Lambdab_{\fh}$ and $\Lambdab_{\gh}$ are the empirical covariance matrices defined as
\begin{align*}
\Lambdab_{\fh} = \frac{1}{n}\sum_{i = 1}^{n} \fh(x_i) \fh^{\T}(x_i), \quad \Lambdab_{\gh} = \frac{1}{n}\sum_{i = 1}^{n} \gh(y_i) \gh^{\T}(y_i).
\end{align*}

Now, let us define the $| \cX |$ and $| \cY |$ dimensional vectors $\hat{\phib}_i$ and  $\hat{\psib}_i$, respectively, for $i = 1, \ldots , k$ as
\begin{align} \label{eq:correspondence}
\phih_i (x) = \sqrt{\Ph_X(x)} \fh_i(x), \quad \psih_i (y) = \sqrt{\Ph_Y(y)} \gh_i (y),
\end{align}
where $\Ph_X$ and $\Ph_Y$ are the empirical marginal distributions, and $\fh_i$ and $\gh_i$ are the $i$-th dimension of $\fh$ and $\gh$, i.e.,
\begin{align*}
\fh(x) = \left[ \fh_1(x), \cdots, \fh_k(x) \right]^{\T}, \quad \gh(y) = \left[ \gh_1(y), \cdots, \gh_k(y) \right]^{\T},
\end{align*}
for all $x,y$. In addition, we define $\dtmh \in  \mathbb{R}^{|\cY| \times |\cX|}$ as
\begin{align} \label{eq:dtmh}
\dtmeh (y,x) = \frac{\Ph_{XY}(x,y)}{\sqrt{\Ph_{X}(x)\Ph_Y(y)}} - \sqrt{\Ph_{X}(x)\Ph_Y(y)},
\end{align}
which is the CDM matrix for training samples. Then, the key steps of \algoref{alg:ace} that alternatively compute conditional expectations (cf. line \ref{alg:ace:line:begin}--\ref{alg:ace:line:end}) can be equivalently expressed as alternating iterations \cite[Eq. (10) and (12)]{HXZW2019}
\begin{equation}
\begin{aligned} 
\text{i)}\quad \Phih_k &\leftarrow \dtmh^{\T} \Psih_k \left( \Psih_k^{\T} \Psih_k \right)^{-1}  \\
\text{ii)}\quad \Psih_k &\leftarrow \dtmh \Phih_k \left( \Phih_k^{\T} \Phih_k \right)^{-1}
\end{aligned}
\label{eq:ace_matrix_a}
\end{equation}
until $\bbfrob{\dtmh}^2- \bbfrob{\dtmh - \Psih_k\Phih_k^\T}^2 $ stops increasing, where 
\begin{align} \label{eq:Phih:Psih}
\Phih_k = \left[ \hat{\phib}_1, \dots, \hat{\phib}_k \right], \quad  \Psih_k = \left[ \hat{\psib}_1, \dots, \hat{\psib}_k \right].
\end{align}
Note that \eqref{eq:ace_matrix_a} in fact coincides with the alternating least squares algorithm \cite{koren2009matrix} for %
solving the low-rank approximation problem
\begin{align}
  \min_{\Psih_k, \Phih_k}\,\bbfrob{\dtmh - \Psih_k\Phih_k^\T}^2.
  \label{eq:low-rank}
\end{align}
Therefore, from the Eckart--Young--Mirsky theorem \cite{eckart1936approximation}, \algoref{alg:ace} essentially computes the singular vectors of $\dtmh$ with respect to the top $k$ singular values, with more detailed illustrations provided in \appref{sec:app:ace} for completeness.
In the following, we simply use $\hat{\phib}_i$ to denote the $i$-th right singular vector of $\dtmh$, and $\Phih_k$ is the $|\cX| \times k$ matrix formed by the top $k$ right singular vectors.

As discussed above, the maximal correlation functions of~\eqref{eq:HGR} correspond to the top $k$ singular vectors [cf.~\eqref{eq:correspondence}] of the matrix $\dtmt$ as defined in \eqref{eq:dtmt}. Therefore, if the empirical distribution $\Ph_{XY}$ coincides with the true distribution $P_{XY}$, then the matrix $\dtmh$ satisfies $\dtmh = \dtmt$, and the ACE algorithm outputs the maximal correlation functions of~\eqref{eq:HGR} precisely. However, since the training samples are i.i.d. sampled from $P_{XY}$, the empirical distribution might deviate from the true distribution, which leads to deviations between the true singular vectors and the ones computed from the ACE algorithm. In order to quantify this deviation, we define the $|\cX| \times k$ matrix $\Phib_k$ as %
\begin{align*} %
\Phib_k \triangleq \left[ \phib_1, \dots, \phib_k \right] , %
\end{align*}
where $\phib_i$ is the $i$-th right singular vector of $\dtmt$.
Then, we apply the difference of the Frobenius norm-squares $\bbfrob{\dtmt \Phib_k}^2 - \bbfrob{\dtmt \Phih_k}^2$ as the measurement to quantify how $\Phih_k$ deviates from $\Phib_k$. \reply{It is worth emphasizing that this measurement represents the test error of learned singular vectors and can be more effective than directly computing the difference between $\Phib_k$ and $\Phih_k$. For example, consider the simplest case $k = 1$, and let $\sigma_i$ denote the $i$-th singular value of $\dtmt$. Without loss of generality, we assume that the estimated $\phibh_1$ satisfies $\ip{\phibh_1}{\phib_1} \geq 0$, since if not we can use $-\phibh_1$ as the estimated singular vector. Then, it is shown in \appref{sec:app:metric} that when $\sigma_1 > \sigma_2$, we have
\begin{align}
  \norm{\phib_1 - \phibh_1}^2 \leq \frac{2}{\sigma_1^2 - \sigma_2^2}\left( \norm{\dtmt \phib_1}^2 - \norm{\dtmt \phibh_1}^2\right),
  \label{eq:metric}
\end{align}
where the Frobenius norm becomes $\ell_2$-norm since $k = 1$. Therefore, a small error in $\norm{\dtmt \phib_1}^2 - \norm{\dtmt \phibh_1}^2$ implies a small error measured in $\norm{\phib_1 - \phibh_1}$. However, when $\sigma_1 = \sigma_2$, both $\phib_1$ and $\phib_2$ (and thus any linear combination of $\phib_1$ and $\phib_2$) correspond to the maximal correlation function. While the measurement $\norm{\dtmt \phib_1}^2 - \norm{\dtmt \phibh_1}^2$ is able to indicate zero learning error for both optimal choices $\phibh_1 = \phib_1$ and $\phibh_1 = \phib_2$, the measurement $\norm{\phib_1 - \phibh_1}$ would indicate a large error for the optimal estimation $\phibh_1 = \phib_2$.}

In addition, this measurement can be interpreted as the performance of learning the matrix $\dtmt$ using $\Phih_k$ by low-rank approximation, since
\begin{align*}
  \bbfrob{\dtmt \Phih_k}^2 = \bbfrob{\dtmt}^2 - \min_{\Psih_k}\,\bbfrob{\dtmt - \Psih_k\Phih_k^\T}^2.
\end{align*}

Moreover, it is shown in~\cite{HXZW2019} that $\bbfrob{\dtmt \Phih_k}^2$ is related to the softmax regression loss, and is called H-score therein, which provides the operational meaning for our selected performance measurement.

Note that $\bbfrob{\dtmt \Phib_k}^2 - \bbfrob{\dtmt \Phih_k}^2 \geq 0$, for all $|\cX| \times k$ matrices $\Phih_k$ satisfying $\Phih_k^{\T} \Phih_k = \Ib_k$, where $\Ib_k$ is the $k \times k$ identity matrix. \reply{In this paper, our goal is to characterize the sample complexity in learning maximal correlation functions for large datasets, i.e., the minimum number of samples required such that with high probability, the learning error $\bbfrob{\dtmt \Phib_k}^2 - \bbfrob{\dtmt \Phih_k}^2$ is small \cite{shalev2014understanding}. %
} To this end, we first consider the related error exponent $\Er_k$ defined as\footnote{Throughout, all logarithms are base $e$, i.e., natural.}

\begin{align} \label{eq:asym_sample_complexity}
\Er_k \defeq -\lim_{\eps \rightarrow 0^{+}} \frac{1}{\eps} \lim_{n \rightarrow \infty} \frac{1}{n} \log \mathbb{P}_n \left\{ \bbfrob{ \dtmt \Phib_k }^2 - \bbfrob{ \dtmt \Phih_k }^2 > \eps  \right\},
\end{align} 
where the probability is measured over the i.i.d. sampling process from $P_{XY}$. In particular, the first limit in~\eqref{eq:asym_sample_complexity} indicates the asymptotic regime of $n$ we are interested in, since in large datasets the number of i.i.d. samples $n$ can be sufficiently large. Then, the second limit of $\eps$ is naturally from that in this asymptotic regime, the empirical distribution converges to the true distribution with high probability, and thus the learning error
is small with high probability; see, e.g.,  \cite[Proposition 4.6]{makur2019information} or \cite[Proposition 47]{huang2019universal} for a more rigorous characterization. 

In the remaining sections, we will show that these two limits facilitate the derivation of the analytical solution for the error exponent~\eqref{eq:asym_sample_complexity}, and further use this exponent to provide the upper bounds for the sample complexity on large datasets, where the test error $\eps$ is small. 
In order to establish the analytical solution of \eqref{eq:asym_sample_complexity}, in the next section, we develop a mathematical framework for computing the learning error $\bbfrob{\dtmt \Phib_k}^2 - \bbfrob{\dtmt \Phih_k}^2$
for the empirical distributions $\Ph_{XY}$ close to the true distribution $P_{XY}$.

\section{The Matrix Perturbation Analyses} \label{sec:mpa}

\newcommand{\pv}{\tau} %

Suppose that $\Ab \in \mathbb{R}^{d \times d}$ is a symmetric matrix with the eigenvectors $\vb_1, \dots, \vb_d$ and the eigenvalues $\lambda_1 \geq \dots \geq \lambda_d$. In addition, we denote 
\begin{align} \label{eq:vk}
\Vb_k \triangleq \left[ \vb_1, \dots, \vb_k \right] \in \mathbb{R}^{d \times k}
\end{align}
as the matrix formed by the top $k$ eigenvectors of $\Ab$. %
Then, it follows that
\begin{align} \label{eq:tr1}
\trop{\Vb_k^{ \T} \Ab \Vb_k} = \sum_{i=1}^k \lambda_i,
\end{align}
where $\tr\{ \cdot \}$ denotes the trace of the matrix.
Now, suppose that $\Ab(\pv)$ is a family of symmetric matrices parametrized by $\pv$ with $\Ab(0) = \Ab$, and is an analytic function of $\pv$ with the Taylor series expansion
\begin{align*}
\Ab(\pv) = \Ab + \pv \Ab' + o(\pv),
\end{align*}
where $\Ab' = \Ab'(0)$ is %
the first-order derivative of $\Ab(\pv)$ with respect to $\pv$ at $\pv = 0$. In addition, let $\Vb_k(\pv) \in \mathbb{R}^{d \times k}$ be the matrix formed by the top $k$ eigenvectors of $\Ab(\pv)$ defined similarly to~\eqref{eq:vk}. Then, when $\lambda_k > \lambda_{k+1}$, the following lemma characterizes the second-order Taylor series expansion of $\tr \left\{ \Vb_k^{\T}(\pv) \Ab \Vb_k(\pv) \right\}$ with respect to $\pv$.

\begin{lemma} \label{lem:eig:k}
Suppose that $\lambda_k > \lambda_{k+1}$, then
\begin{align*}
&\tr \left\{ \Vb_k^{\T}(\pv) \Ab \Vb_k(\pv) \right\} %
  = \tr \left\{ \Vb_k^{ \T} \Ab \Vb_k \right\} - \pv^2\sum_{i=1}^k\sum_{j = k + 1}^d\frac{\bigl(\vb_i^{\T}\Ab'\vb_j\bigr)^2}{\lambda_i - \lambda_j} + o(\pv^2),
\end{align*}
where $\tr\{ \cdot \}$ denotes the trace of the matrix.
\end{lemma}

\begin{proof}
See Appendix~\ref{sec:appa}.
\end{proof}

Moreover, for the case $\lambda_k = \lambda_{k+1}$, we apply the notation $[d] \defeq \{1, \dots, d\}$, and define the indices set $\cI_k \defeq \{i \in [d]\colon \lambda_i = \lambda_k \}$, and the complement set $\cIb_k \defeq [d] \setminus \cI_k = \{i \in [d]\colon \lambda_i \neq \lambda_k \}$. Furthermore, we define the matrix $\Vb_{\cI_k} \defeq [\vb_{i}, i\in \cI_k] \in \mathbb{R}^{d \times |\cI_k|} $ as the submatrix of $\Vb$ composed of the columns of $\Vb$ with indices in $\cI_k$. Then, the following lemma establishes the second-order Taylor series expansion of $\tr \left\{ \Vb_k^{\T}(\pv) \Ab \Vb_k(\pv) \right\}$ for the case $\lambda_k = \lambda_{k+1}$.

\begin{lemma} \label{lem:2}
  Suppose that $\lambda_k = \lambda_{k + 1}$, then
  \begin{align*}
    \tr \left\{ \Vb_k^{\T}(\pv) \Ab \Vb_k(\pv) \right\} %
    = \tr \left\{ \Vb_k^{ \T} \Ab \Vb_k \right\} - \pv^2\sum_{i=1}^{l - 1}\sum_{j = l}^d\frac{\left(\vb_i^{\T}\Ab'\vb_j\right)^2}{\lambda_i - \lambda_j} %
    - \pv^2\sum_{i = l}^k\sum_{j \in \cIb_k}\frac{\left(\vbh_{i}^{\T}\Ab'\vb_j\right)^2}{\lambda_k - \lambda_j} + o(\pv^2),
\end{align*}
where $l$ is the minimal element of $\cI_k$, and
\begin{align*}
  \vbh_{i} \defeq \Vb_{\cI_k}\ub_{i - l + 1}, \quad l \leq i \leq k,
\end{align*}
where $\ub_{1}, \dots, \ub_{k - l + 1} \in \mathbb{R}^{ |\cI_k|}$ are the top $k - l + 1$ eigenvectors of $\Vb_{\cI_k}^{\T}\Ab'\Vb_{\cI_k}$.
\end{lemma}

\begin{proof}
See Appendix~\ref{sec:appb}.
\end{proof}

Note that since the Frobenius norm $\bbfrob{\dtmt \Phih_k}^2 = \tr \left\{ \Phih_k^{\T} \dtmt^{\T} \dtmt \Phih_k \right\}$, the results we developed in this section essentially characterize the difference between $\left\| \dtmt \Phih_k \right\|^2_{\mathrm{F}}$ and $\left\| \dtmt \Phib_k \right\|^2_{\mathrm{F}}$ with respect to the perturbations on $\dtmt$ due to the difference between $P_{XY}$ and $\Ph_{XY}$. These results will be useful when we derive the error exponent~\eqref{eq:asym_sample_complexity} in the rest of this paper.

\section{The Supervised Learning} \label{sec:sup}
Given $n$ training samples $(x_1, y_1), \ldots, (x_n, y_n)$, i.i.d. generated from the joint distribution $P_{XY}$, in this section we develop the error exponent~\eqref{eq:asym_sample_complexity} and an upper bound for sample complexity for large datasets, in both cases $\sigma_k > \sigma_{k+1}$ and $\sigma_k = \sigma_{k+1}$, where $\sigma_k$ and $\sigma_{k+1}$ are the $k$-th and $(k+1)$-th largest singular values of $\dtmt$. 

\subsection{The Sample Complexity for the Case $\sigma_k > \sigma_{k+1}$} \label{sec:k>k+1}
To delineate our results, we need the following definitions. First, we define the quantity $\alpha_k$ for the given $P_{XY}$, which will be useful in characterizing the error exponent \eqref{eq:asym_sample_complexity}.

\begin{definition}%
  \label{def:Gb:alpha}
  Given a joint distribution $P_{XY}$ and $k \in \mathbb{N}^+$, 
  we define the matrix $\Gb_k$ as
\begin{align}  \label{eq:gb:k}
  \Gb_k \defeq \Lb^{\T} \left(\sum_{i = 1}^k\sum_{j = k + 1}^d \frac{\thetab_{ij}\thetab_{ij}^{\T}}{\sigma_i^2 - \sigma_j^2}\right) \Lb,
\end{align}
where $d \triangleq |\cX |$, and $\sigma_i$ denotes the $i$-th singular value\footnote{We define $\sigma_i = 0$ for $i > |\cY|$, if $|\cX| > |\cY|$.} of $\dtmt$. In addition, $\Lb$ is an $(|\cX| \cdot |\cY|) \times (|\cX| \cdot |\cY|)$ matrix, whose entry at the $[(x - 1)|\cY| + y]$-th row and $[(x' - 1)|\cY| + y']$-th column is defined as
\begin{align} %
  &\sqrt{\frac{P_{XY}(x', y')}{P_X(x)P_Y(y)}} \biggl[\delta_{xx'}\delta_{yy'} - \frac{1}{2}\left(\frac{\delta_{xx'}}{P_X(x)} + \frac{\delta_{yy'}}{P_Y(y)}\right)%
      \cdot\left[P_{XY}(x,y)+P_{X}(x)P_Y(y)\right]\biggr],\label{eq:L:def}
\end{align}
where $\delta_{ij}$ is the Kronecker delta, and
\begin{align} \label{eq:theta_sup}
  \thetab_{ij} \defeq \phib_j \otimes \bigl(\dtmt\phib_i\bigr) +  \phib_i \otimes \bigl(\dtmt\phib_j\bigr),\quad 1 \leq i \leq j \leq d,
\end{align}
where ``$\otimes$'' denotes the Kronecker product, and $\phib_i$ %
represents the $i$-th right singular vector %
of $\dtmt$. Then, $\alpha_k$ is defined as the spectral norm of $\Gb_k$.
\end{definition}

Then, the error exponent $\Er_k$ as defined in \eqref{eq:asym_sample_complexity} can be established as follows, whose proof will later be provided.

  \begin{theorem} \label{thm:exponent}
    If $\sigma_k > \sigma_{k + 1}$, then the error exponent $\Er_k$ as defined in \eqref{eq:asym_sample_complexity} is
    \begin{align} 
      \Er_k = -\lim_{\eps \rightarrow 0^{+}} \frac{1}{\eps} \lim_{n \rightarrow \infty} \frac{1}{n} \log \mathbb{P}_n \left\{ \bbfrob{ \dtmt \Phib_k }^2 - \bbfrob{ \dtmt \Phih_k }^2 > \eps  \right\} = \frac{1}{2\alpha_k}.
      \label{eq:exponent}
    \end{align}
  \end{theorem}

  \reply{%
    Then, the following result illustrates that, for large datasets where the learning error is typically small, an upper bound of sample complexity can be established using the error exponent $\Er_k$.}

  \reply{
    
  \begin{theorem} \label{thm:samples}
    For given $P_{XY}$, there exist an absolute positive constant $\eps_0 > 0$ that depends only on $P_{XY}$, such that for all $\eps \in (0, \eps_0)$ and $\delta \in (0, 1)$, we have
    \begin{align}
      \mathbb{P}_n \left\{ \bbfrob{ \dtmt \Phib_k }^2 - \bbfrob{ \dtmt \Phih_k }^2 > \eps \right\} < \delta
      \label{eq:pac}
    \end{align}
    for all $n > N^{(4\alpha_k)}(\eps, \delta)$,
    where we have defined%
    \begin{align}
      N^{(t)}(\eps, \delta) \defeq \frac{t|\cX||\cY|}{\eps} \log \frac{6t|\cX||\cY|}{\eps} + \frac{t}{\eps}\log\frac1\delta,
    \end{align}
    and where $\alpha_k$ is as defined in~\defref{def:Gb:alpha}.
  \end{theorem}
}

\begin{proof}
  See \appref{sec:app:col:non-asym}.
\end{proof}

\reply{
  From \thmref{thm:samples}, to guarantee that the error in learning maximal correlation functions is within some small $\eps$ with probability at least $1 - \delta$, it suffices to use $n = O \left( \frac{\alpha_k}{\eps} \log \frac{\alpha_k}{\eps\delta} \right) = O \left( \frac{1}{\eps\Er_k} \log \frac{1}{\delta\eps\Er_k} \right)$ samples.

  \begin{remark}
    For comparison, an upper bound of sample complexity $n = O(\frac{1}{\eps^2}\log \frac{1}{\delta})$ was provided in \cite[Proposition 4.6]{makur2019information} and \cite[Proposition 47]{huang2019universal}. In particular, this upper bound is obtained via analyzing the concentration properties of $\dtmh$, with the assumption that the true marginal distributions $P_{X}$ and $P_Y$ have been known in advance. While our analysis does not rely on such assumptions, the resulting upper bound of sample complexity for large datasets can be significantly tighter.
  \end{remark}
}

When we are interested in learning the entire correlation structure between $X,Y$, i.e., $k = d-1$, the following proposition establishes a simple closed form solution of the error exponent.%

\begin{proposition}
  \label{eg:1}
  If $d = |\cX| \leq |\cY|$, $\sigma_{d - 1} > 0$, and $k = d - 1$%
  , then we have $\alpha_k = \frac{\sigma_1^2}{4}$ and
  \begin{align*}
    \Er_k = \frac{2}{\sigma_1^2}.%
  \end{align*}
\end{proposition}
\begin{proof}
  See \appref{sec:appc}.
\end{proof}

\reply{
We then introduce the proof of \thmref{thm:exponent}, which will make use of the perturbation analyses established in \secref{sec:mpa}. To begin, we first define the following sets of empirical distributions.}%

\begin{definition} For all $\eps > 0$, the set $\cS_1(\eps)$ is defined as
  \label{def:nbhd:S1}
  \begin{align}
    \cS_1(\eps) \triangleq \left\{ \Ph_{XY}\colon  \bbfrob{\dtmt \Phib_k}^2 - \bbfrob{\dtmt \Phih_k}^2 > \eps  \right\},
    \label{eq:s1_sup}
  \end{align}
where $\Phih_k$ corresponds to the top $k$ right singular vectors of $\dtmh$, as defined in~\eqref{eq:dtmh} and~\eqref{eq:Phih:Psih}. Moreover, the set $\nbhd(\eps)$ is defined as
  \begin{align}
    \nbhd(\eps) \defeq \left\{\Ph_{XY}\colon D(\Ph_{XY}\|P_{XY}) \leq \frac{\eps}{\alpha_k}\right\}.
    \label{eq:nbhd}
  \end{align}
\end{definition}

Furthermore, to characterize the empirical distributions $\Ph_{XY} \in \nbhd(\eps)$, we denote the difference between $\Ph_{XY}$ and the true distribution $P_{XY}$ as
      \begin{align} \label{eq:imate:def}
        \imate(y,x) \defeq
        \begin{cases}
          \displaystyle \frac{\Ph_{XY} (x,y) - P_{XY} (x,y)}{ \sqrt{\eps P_{XY}(x,y)}},&\!\!\text{if}~P_{XY}(x,y) > 0,\\
          0,&\!\!\text{if}~P_{XY}(x,y) = 0,
        \end{cases}
      \end{align}
      which induces a one-to-one correspondence\footnote{Note that since $D(\Ph_{XY}\|P_{XY})$ is finite, we have $\Ph_{XY}(x, y) = 0$ for each $(x, y)$ with $P_{XY}(x, y) = 0$. Therefore, we can obtain the distribution $\Ph_{XY}$ from $\imate$, via
        \begin{align*}
          \Ph_{XY}(x, y) = P_{XY}(x, y) + \sqrt{\eps P_{XY}(x, y)}\imate(y, x), \text{~for all~} (x, y) \in \cX \times \cY.
        \end{align*}
      } $\Ph_{XY} \leftrightarrow \imate$. We also define the $|\cY| \times |\cX|$ matrices $\imat$ and $\Xib$, with entries at the $y$-th row and $x$-th column being $\imate(y, x)$ and %
  \begin{align}
    \dtmdmate(y, x) &\defeq \frac{\sqrt{P_{XY}(x,y)}}{\sqrt{P_X(x)P_Y(y)}} \imate(y, x) -  \frac{P_{XY}(x, y) + P_X(x)P_Y(y)}{2\sqrt{P_X(x)P_Y(y)}}\notag\\
                    &\qquad\cdot\left[\frac{1}{P_X(x)}  \sum_{y' \in \cY} \sqrt{P_{XY}(x, y')} \imate(y', x) +\frac{1}{P_Y(y)} \sum_{x' \in \cX}\sqrt{P_{XY}(x', y)}\imate(y, x') \right], \label{eq:theta}
  \end{align}
  respectively. Then, using \lemref{lem:eig:k}, the matrix $\dtmh$ estimated from data samples with the empirical distribution in $\nbhd(\eps)$ can be represented in a perturbation form illustrated as follows.

  \begin{lemma}
    \label{lem:perturb:B}
    For given $P_{XY}$,  there exists a constant $C > 0$, such that for all $\eps > 0$ and $\Ph_{XY} \in \nbhd(\eps)$, we have $\bbfrob{\dtmdmat} \leq C$ and
    \begin{align}
      \label{eq:perturb:B}
      \dtmh  = \dtmt + \sqrt{\eps}\, \dtmdmat + o\left(\sqrt{\eps}\right).
    \end{align}    
  \end{lemma}
  \begin{proof}
    See \appref{sec:app:dtmh}.
  \end{proof}

  Moreover, the following lemma characterizes the I-projection of $P_{XY}$ onto the set $\cS_1(\eps) \cap \nbhd(\eps)$, which will be useful for characterizing the error exponent.

  \begin{lemma}
  \label{lem:S1:nbhd}
  For $\cS_1(\eps)$ and $\nbhd(\eps)$ as defined in \defref{def:nbhd:S1}, we have\footnote{Given a distribution $P$, we adopt the notation \cite[p.~431]{csiszar2004information}
    \begin{align*}
    D(\cS\|P) \defeq \inf_{Q \in \cS}D(Q\|P),
    \end{align*}
 where %
    $\cS$ is a set of distributions.}
    \begin{align}
      \lim_{\eps \rightarrow 0^+} \eps^{-1}\,  D\left(\cS_1(\eps) \cap \nbhd(\eps)\,\middle\|\, P_{XY}\right) = \frac{1}{2\alpha_k}.
      \label{eq:S1:nbhd}
    \end{align}    
  \end{lemma}

  \begin{proof}
    See \appref{sec:app:S1:nbhd}.
  \end{proof}

  Based on the above lemmas, \thmref{thm:exponent} can be established as follows.
  
  \begin{proof}[Proof of \thmref{thm:exponent}]
    First, it follows from Sanov's theorem \cite{csiszar2004information} that
    \begin{align} 
      \mathbb{P}_n \left\{ \bbfrob{ \dtmt \Phib_k }^2 - \bbfrob{ \dtmt \Phih_k }^2 > \eps  \right\}%
       \doteq \exp \left\{ -n  D\left(\cS_1(\eps)\middle\| P_{XY}\right) \right\}, \label{eq:scnn}
    \end{align} 
    where ``$\doteq$'' is the conventional dot-equal notation.\footnote{In particular, we use $a_n \doteq \exp(nb)$ to denote
      \begin{align*}
        \lim_{n \to \infty} \frac{1}{n}\log a_n = b.
      \end{align*}
    } Therefore, the error exponent \eqref{eq:exponent} can be expressed as
    \begin{align}
      -\lim_{\eps \rightarrow 0^+} \frac{1}{\eps} \lim_{n \rightarrow \infty} \frac{1}{n} \log \mathbb{P}_n \left\{ \bbfrob{ \dtmt \Phib_k }^2 - \bbfrob{ \dtmt \Phih_k }^2 > \eps  \right\} = \lim_{\eps \rightarrow 0^+} \eps^{-1} D\left(\cS_1(\eps)\middle\| P_{XY}\right). %
      \label{eq:lim:sanov}
    \end{align}

    From Lemma \ref{lem:S1:nbhd}, there exists an $\eps_0 > 0$ such that, for all $\eps \in (0, \eps_0)$,
    \begin{align*}
       D\left(\cS_1(\eps) \cap \nbhd(\eps)\,\middle\|\, P_{XY}\right)
      < \frac{\eps}{\alpha_k}.
    \end{align*}
    In addition, note that from~\eqref{eq:nbhd}, for all $\Ph_{XY} \in \cS_1(\eps) \setminus \nbhd(\eps)$ we have $D\bigl(\Ph_{XY}\big\| P_{XY}\bigr) > \frac{\eps}{\alpha_k}$. Hence, for all $\eps \in (0, \eps_0)$ we have
    \begin{align*}
      D\left( \cS_1(\eps)\,\middle\|\, P_{XY}\right)
         = D\left( \cS_1(\eps) \cap \nbhd(\eps)\,\middle\|\, P_{XY}\right),
    \end{align*}
    which implies that
    \begin{align}
      \lim_{\eps \rightarrow 0^+} \eps^{-1}  D\left(\cS_1(\eps) \,\middle\|\, P_{XY}\right) = \lim_{\eps \rightarrow 0^+} \eps^{-1} D\left( \cS_1(\eps) \cap \nbhd(\eps)\,\middle\|\, P_{XY}\right) = \frac{1}{2\alpha_k}.
      \label{eq:lim:S1}
    \end{align}
    Combining \eqref{eq:lim:sanov} and \eqref{eq:lim:S1}, we obtain \eqref{eq:exponent}.
  \end{proof}

\subsection{The Sample Complexity for the Case $\sigma_k = \sigma_{k+1}$}
\label{sec:exponent:k:k+1}

The idea of deriving the sample complexity in this case is similar to the case $\sigma_k > \sigma_{k+1}$. To delineate the result, we first define
\begin{align} \label{eq:Jk:lmin:Phib}
  \cI_k \defeq \left\{i \in [d]\colon \sigma_i = \sigma_k\right\},
  \quad
  l \defeq \min \cI_k,
  \quad\text{and}\quad
  \Phib_{\cI_k} \defeq \left[\phib_i\colon i \in \cI_k\right] \in \mathbb{R}^{|\cX| \times |\cI_k|}. 
\end{align}

Similar to $\Gb_k$ and $\alpha_k$ defined in \secref{sec:k>k+1}, for the case $\sigma_k = \sigma_{k + 1}$ we define the matrix $\Jb_k$ and $\beta_k$ to characterize the error exponent $\Er_k$.

\begin{definition}
  \label{def:Jb:beta}
    Given $\imatj \in \mathbb{R}^{|\cY| \times |\cX|}$, the matrix $\Jb_k(\imatj) \in \mathbb{R}^{(|\cX||\cY|) \times (|\cX||\cY|)}$ is defined as
    \begin{align}\label{eq:Jbk}
      \Jb_k(\imatj) \defeq \Gb_{l - 1} + \Lb^{\T}\left(\sum_{i = l}^k\sum_{j \in \cIb_k} \frac{\varthetab_{ij}\varthetab^{\T}_{ij}}{\sigma_i^2 - \sigma_j^2}\right)\Lb, 
    \end{align}
    where $\Gb_{l - 1}$ and $\Lb$ are as defined in \eqref{eq:gb:k}--\eqref{eq:L:def}, where $\cIb_k \defeq [d]\setminus \cI_k$, and $\varthetab_{ij}$ are defined as, for all $i,j$,
      $\varthetab_{ij} \defeq \phib_j \otimes \bigl(\dtmt\varphib_i\bigr) + \varphib_i \otimes \bigl(\dtmt\phib_j\bigr)$, 
    where $\varphib_i$ are defined as
    \begin{align} \label{eq:varphi:def}
      \varphib_i \defeq \Phib_{\cI_k}\ub_{i - l + 1},\quad l \leq i \leq k,
    \end{align}
    where $\ub_1, \dots, \ub_{k - l + 1} \in \mathbb{R}^{|\cI_k|}$ are the top $k - l + 1$ eigenvectors of the matrix $\Phib_{\cI_k}^{\T}\left(\dtmt^{\T}\dtmdmat + \dtmdmat^{\T}\dtmt\right)\Phib_{\cI_k}$, and $\dtmdmat$ is as defined in \eqref{eq:theta}. In addition, $\beta_k$ is defined as the optimal value of the optimization problem\footnote{Here, we apply the vectorization operation $\vec(\cdot)$ to stack all columns of a matrix into a vector. Specifically, for $\Wb = [w_{ij}] \in \mathbb{R}^{p\times q}$, $\vec(\Wb)$ is a $pq$-dimensional column vector with the $[p(j-1) + i]$-th entry being $w_{ij}$.}
    \begin{subequations}
    \begin{alignat}{2}
      &\max_{\imatj} &\quad &\vec^{\T}\left(\imatj\right)\Jb_k(\imatj)\vec\left(\imatj\right)    \\
      &~\,\mathrm{s.t.}& & \|\imatj\|_{\F}^2 \leq 1.
    \end{alignat}
    \label{eq:opt:k:k+1:Jbk}
  \end{subequations}
  \end{definition}
  Then, the following theorem characterizes the error exponent for the general case where $\sigma_k$ and $\sigma_{k + 1}$ can be equal, and the corresponding upper bound of sample complexity can be established similar to \thmref{thm:samples}.
\begin{theorem}\label{thm:exponent:k:k+1}
  If $\sigma_k = \sigma_{k + 1}$, the error exponent $\Er_k$ as defined in \eqref{eq:asym_sample_complexity} is
  \begin{align} 
    \Er_k = -\lim_{\eps \rightarrow 0^+} \frac{1}{\eps} \lim_{n \rightarrow \infty} \frac{1}{n} \log \mathbb{P}_n \left\{ \bbfrob{ \dtmt \Phib_k }^2 - \bbfrob{ \dtmt \Phih_k }^2 > \eps  \right\} = \frac{1}{2\beta_k}.
    \label{eq:exponent:k:k+1}
  \end{align}
\end{theorem}
\begin{proof}
    See \appref{sec:appd}.
\end{proof}

Note that $\Jb_k(\imat)$ in~\eqref{eq:Jbk} is dependent on $\imatj$, since $\varthetab_{ij}$ is dependent on $\dtmdmat$. Therefore, unlike Theorem~\ref{thm:exponent}, the optimal value of~\eqref{eq:opt:k:k+1} is not simply the largest singular value of some given matrix, and the optimization problem~\eqref{eq:opt:k:k+1} is in general neither convex nor concave. However, note that if $\Jb_k$ is fixed, the optimization problem~\eqref{eq:opt:k:k+1} is reduced to solving the largest singular value of $\Jb_k$. Therefore, we can first fix $\imatj$ to compute (or update) $\Jb_k$, and then solve the largest singular vector of $\Jb_k$ to update $\imatj$, and so forth. This iterative procedure, summarized in Algorithm~\ref{alg:alpha}, solves the local optimum of the optimization problem~\eqref{eq:opt:k:k+1}. In particular, to update $\imatj$ (cf. line \ref{alg:alpha:proj:begin}--\ref{alg:alpha:proj:end} of Algorithm~\ref{alg:alpha}), we project $\vec\left(\imatj\right)$ onto the eigenspace of $\Jb_k$ associated with its largest singular value, where a learning rate $\eta$ is used to enhance the robustness of the update.

\begin{algorithm}[t]
  \caption{An algorithm for computing the error exponent of \thmref{thm:exponent:k:k+1}}%
  \label{alg:alpha}
  \begin{algorithmic}[1]
    \State {\bfseries Input:} $k$ and a learning rate $\eta$.
    \State Compute $\cI_k$, $l$, and $\Phib_{\cI_k}$.
    \State Randomly choose a $\imat \in \mathbb{R}^{|\cY|\times |\cX|}$ with $\|\imat\|_{\F} = 1$.
    \Repeat
    \State Compute $\dtmdmat$ from \eqref{eq:theta}.
    \State $\Wb \gets \Phib_{\cI_k}^{\T}\bigl(\dtmt^{\T}\dtmdmat + \dtmdmat^{\T}\dtmt\bigr)\Phib_{\cI_k}$
    \For {$i = 1, \dots, k - l + 1$}
    \State $\ub_i \gets$ the $i$-th eigenvector of $\Wb$
    \State $\varphib_{i+l-1} \gets \Phib_{\cI_k}\ub_{i}$
    \EndFor
    \State Compute $\Jb_k(\imat)$ from \eqref{eq:Jbk}.
    \State $\beta_k \gets \vec^{\T}\left(\imat\right){\Jb_k(\imat)}\vec\left(\imat\right)$
    \State Compute the eigenvectors $\qb_1, \dots, \qb_s$ of $\Jb_k(\imat)$ associated with its largest eigenvalue.\label{alg:alpha:proj:begin}
    \State $\Qb \gets [\qb_1, \dots, \qb_s]$ 
    \State $\vec\left(\imat\right) \gets \vec\left(\imat\right) + \eta\Qb\Qb^{\T}\vec\left(\imat\right)$%
    \State $\displaystyle\vec\left(\imat\right) \gets \frac{\vec\left(\imat\right)}{\left\|\vec\left(\imat\right)\right\|}$   \label{alg:alpha:proj:end}
    \Until{$\beta_k$ converges.}
    \State $\Er_k \gets (2\beta_k)^{-1}$
    \State {\bfseries Output:} $\Er_k$.
  \end{algorithmic}
\end{algorithm}

While there is in general no closed form solution for~\eqref{eq:opt:k:k+1}, for some special joint distributions the closed form solutions exist.

\begin{corollary}
  \label{eg:2}
Suppose $d = |\cX| \leq |\cY|$, and the joint distribution $P_{XY}(x, y)$ takes the form
  \begin{align*}
    P_{XY}(x, y) =
    \begin{cases}
      p_1, &\text{if}~x = y,\\
      p_2,&\text{if}~x \neq y,
    \end{cases}
  \end{align*}
  where the probabilities $p_1$ and $p_2$ satisfy $p_1 \neq p_2$ and $d[p_1 + (|\cY| - 1)p_2] = 1$. Then for any dimension $k \in [d - 1]$, we have $\beta_k = \frac{\sigma_1^2}{4}$ and thus
  \begin{align*}
    \Er_k = \frac{2}{\sigma_1^2}, %
  \end{align*}
  where
  \begin{align}\label{eq:sigma:eg2}
    \sigma_1 = \dots = \sigma_{d - 1} = \frac{|p_1 - p_2|\sqrt{d}}{\sqrt{p_1 + (d - 1)p_2}}
  \end{align}
  are the none-zero singular values of the corresponding $\dtmt$.
\end{corollary}
\begin{proof}
  See \appref{sec:appe}.
\end{proof}

\subsection{Remarks on the General Trend of Error Exponent}

In machine learning problems, it is also interesting to investigate $\bfrob{\dtmt \Phih_k}^2 / \bfrob{\dtmt \Phib_k}^2$, which tells how effective $\Phih_k$ is, compared to $\Phib_k$. In particular, this is studied by the asymptotic problem
\begin{subequations}
  \begin{align} 
    \Erh_k  \defeq -\lim_{\eps \rightarrow 0^+} \frac{1}{\eps} \lim_{n \rightarrow \infty} \frac{1}{n} \log \mathbb{P}_n \left\{ \frac{ \bbfrob{ \dtmt \Phib_k }^2 - \bbfrob{ \dtmt \Phih_k }^2 }{\bbfrob{ \dtmt \Phib_k }^2}> \eps  \right\},
  \end{align} 
  in which the normalized error exponent $\Erh_k$ is given by
  \begin{align} 
    \Erh_k = \bbfrob{ \dtmt \Phib_k }^2 \cdot \Er_k = \left(\sum_{i=1}^k \sigma_i^2\right)\cdot \Er_k,%
  \end{align}
  \label{eq:error_exp_norm}
\end{subequations}
  where $\sigma_1, \ldots , \sigma_k$ are the top $k$ singular values of $\dtmt$. 

  The error exponent~\eqref{eq:error_exp_norm} combining with \thmref{thm:samples} offers insights on designing the dimensionality $k$ to effectively extract the correlation structure among different data variables from a given set of training samples. In particular, since both $X$ and $Y$ are discrete, the true distribution $P_{XY}$ can be approximated by the empirical distribution $\Ph_{XY}$ of training samples, and then the normalized error exponent $\Erh_k$ can be obtained via computing the corresponding $\alpha_k$ or $\beta_k$ from $\Ph_{XY}$. %

  However, in real algorithm designs, it is more useful to provide a general trend for error exponent over different $k$. For certain symmetric joint distributions, it can be verified that the normalized error exponent $\Er_k$ is linear to $k$. For example, consider the joint distribution $P_{XY}$ as constructed in Corollary~\ref{eg:2}, then we have $\Erh_k = \left(\sum_{i = 1}^k \sigma_i^2\right) \Er_k = k \sigma_1^2 \Er_k = 2k$.
  However, one can easily construct examples that the error exponent is not monotonic with respect to $k$, and the behavior is generally complicated.

  To gain more insights, we uniformly sample\footnote{In particular, we generate independent random numbers uniformly from $[0,1]$ for each entry of $P_{XY}$, and then normalize the sum of the entries to $1$.} the joint distribution $P_{XY}$ from the distribution space of $X,Y$, and consider the empirical average of the error exponents~\eqref{eq:error_exp_norm} over the sampled joint distributions. \figref{fig:alpha1} plots the empirical average of the error exponents over $10^5$ sampled joint distributions with $| \cX | = 12$ and $|\cY| = 10$, from which we observe that the error exponent grows linear to $k$, for small $k$, and becomes super-linear when $k$ is large. Although we do not have a rigorous proof in this paper, this general trend of the normalized error exponent, combining with \thmref{thm:samples}, may provide a practical design guidance for selecting the dimensionality of the maximal correlation functions in real problems.

\begin{figure}[t]
  \centering
  \includegraphics[width = .5\textwidth]{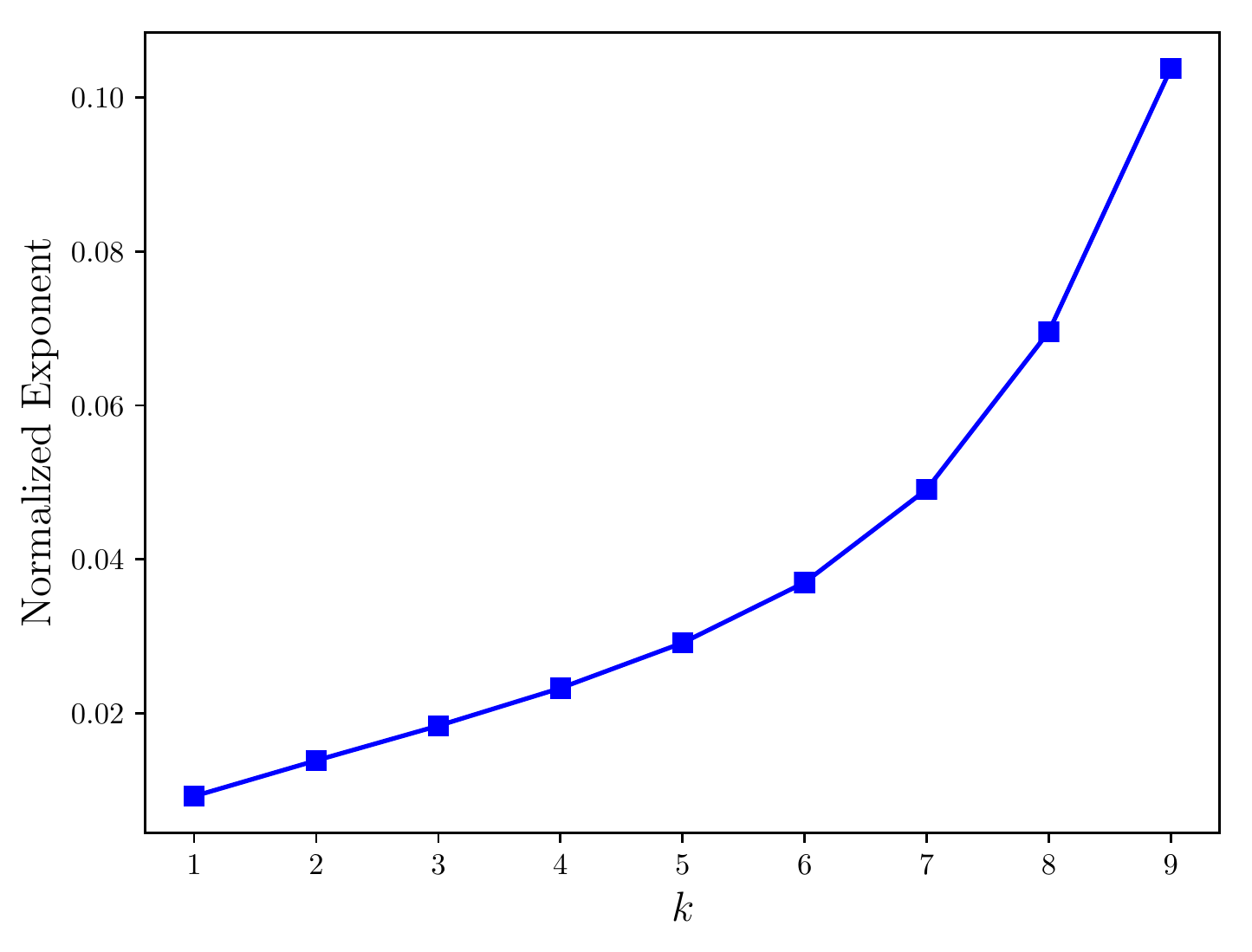}
  \caption{The empirical average of the normalized exponent $\Erh_k$, with respect to the dimensionality $k$ of the HGR maximal correlation functions. The average is taken over $10^5$ randomly generated joint distributions $P_{XY}$ with cardinalities $| \cX | = 12$ and $|\cY| = 10$.
}
  \label{fig:alpha1}
\end{figure}

\section{The Semi-supervised Learning} \label{sec:semi}

In the semi-supervised learning setup, in addition to the $n$ labeled samples $(x_1, y_1), \dots, (x_n, y_n)$, we also have $m = nr$ unlabeled samples $x_{n+1} , \ldots , x_{n+m}$, where $r$ is the ratio between the labeled and unlabeled samples, and the unlabeled samples are assumed to be i.i.d. sampled from the marginal distribution $P_X$ and independent of the labeled samples. In order to estimate the maximal correlation functions from both the labeled and unlabeled samples, we denote the empirical distribution of the unlabeled samples $x_{n+1} , \ldots , x_{n+m}$ as $Q_X$, and again apply $\Ph_{XY}$ and $\Ph_{Y|X}$ to denote the empirical joint and conditional distributions of the labeled samples. Moreover, we denote the empirical distribution for $x_1, \ldots , x_{n+m}$ as $\Pt_X$, which is the empirical marginal distribution of $X$ over all samples, and can be expressed as
\begin{align}
  \Pt_X(x) &= \frac{n}{m + n}\Ph_X(x) + \frac{m}{m + n}Q_X(x)%
    = \frac{1}{r + 1}\Ph_X(x) + \frac{r}{r + 1}Q_X(x).\label{eq:Px:semi}
\end{align}
\reply{
Then, \algoref{alg:ace} can be generalized to estimating the maximal correlation functions from both labeled and unlabeled samples. For this purpose, we define
\begin{align} \label{eq:PtXY}
  \Pt_{XY}(x,y) \defeq \Ph_{Y|X}(y|x) \Pt_X(x)
\end{align}
as the empirical joint distribution by including both the labeled and unlabeled samples, with the corresponding marginal distributions being $\Pt_X$ and 
\begin{align}\label{eq:Pt:y}
  \Pt_{Y}(y) \defeq \sum_{x \in \cX}\Ph_{Y|X}(y|x) \Pt_X(x),
\end{align}
respectively. Similarly, we obtain the conditional distributions
\begin{align*}
  \Pt_{Y|X}(y|x) = \frac{\Pt_{XY}(x, y)}{\Pt_X(x)} = \Ph_{Y|X}(y|x)
\end{align*}
and 
\begin{align*}
  \Pt_{X|Y}(x|y) = \frac{\Pt_{XY}(x, y)}{\Pt_Y(y)} = \Ph_{Y|X}(y|x) \frac{\Pt_X(x)}{\Pt_Y(y)}.
\end{align*}
Then, we can generalize \algoref{alg:ace} to semi-supervised learning by replacing the all expectation operations taken over empirical distributions $\Ph_{X}, \Ph_{Y}, \Ph_{XY}, \Ph_{X|Y}, \Ph_{Y|X}$ with the expectations taken over empirical distributions $\Pt_{X}, \Pt_{Y}, \Pt_{XY}, \Pt_{X|Y}, \Pt_{Y|X}$, respectively. In particular, let $\fb\colon \cX \mapsto \mathbb{R}^k$ and $\gb \colon \cY \mapsto \mathbb{R}^k$ denote the initially chosen functions of $X$ and $Y$, which are zero-mean over the distribution $\Pt_X$ and $\Pt_{Y}$, respectively. Then, the alternating conditional expectation operations (cf. line \ref{alg:alpha:proj:begin}--\ref{alg:alpha:proj:end} of \algoref{alg:ace}) can be represented as:
\begin{equation} \label{eq:gACE}
\begin{aligned} 
\text{i)}\quad \fb(x) &\gets \Lambdab_{\gb}^{-1}  \sum_{y \in \Y} \gb(y) \Ph_{Y|X}(y|x) \\
\text{ii)}\quad \gb(y) &\gets \Lambdab_{\fb}^{-1} \sum_{x \in \X} \fb(x) \Pt_{Y|X}(y|x) = \Lambdab_{\fb}^{-1} \sum_{x \in \X} \fb(x) \Ph_{Y|X}(y|x) \frac{\Pt_X(x)}{\Pt_Y(y)}
\end{aligned}
\end{equation}
where %
 $\Lambdab_{\fb}$ and $\Lambdab_{\gb}$ are the covariance matrices of $\fb$ and $\gb$, defined as
\begin{align*}
  \Lambdab_{\fb} \defeq \sum_{x\in\cX} \Pt_X(y)\fb(x) \fb^{\T}(x) = \frac{1}{n+m}\sum_{i = 1}^{n+m} \fb(x_i) {\fb}^{\T}(x_i)\quad\text{and}\quad
  \Lambdab_{\gb} \defeq \sum_{y\in\cY} \Pt_Y(y)\gb(y) \gb^{\T}(y).%
\end{align*}
}

The idea of this generalization is that the unlabeled samples do not help the estimation of the conditional distribution $P_{Y|X}$, but improve the estimation of the marginal distribution $P_X$. Therefore, the first step in the algorithm remains the same, while in the second step, the improved empirical marginal distribution $\Pt_X$ is applied to improve the estimation of the conditional distribution $P_{X|Y}$. In practice, we may assume that the marginal distribution is much easier to estimate than the joint distribution~\cite{makur2015efficient}, and hence the above generalized ACE algorithm can still be implemented for computing the maximal correlation functions from training samples. In this section, our goal is to characterize the corresponding error exponent for the generalized ACE algorithm~\eqref{eq:gACE} in the semi-supervised learning scenario. %

To this end, let us define  the $|\cY| \times |\cX|$ matrix $\dtmtt$, with entries 
\begin{align}
  \label{eq:dtm:bar}
  \dtmett(y, x) = \frac{\Pt_{XY}(x, y)}{\sqrt{\Pt_X(x)\Pt_Y(y)}} - \sqrt{\Pt_X(x)\Pt_Y(y)},
\end{align}
where $\Pt_Y$ is the marginal distribution of $\Pt_{XY}$ as defined in \eqref{eq:Pt:y}. Then, it is shown in \appref{sec:appf} that the algorithm~\eqref{eq:gACE} essentially computes the top $k$ singular vectors of $\dtmtt$. %
In addition, we denote $\Phibt_k$ as the $|\cX| \times k$ dimensional matrix with the $i$-th column being the $i$-th right singular vector of $\dtmtt$. Then for large datasets, the sample complexity of estimating the maximal correlation functions by the algorithm~\eqref{eq:gACE} can be characterized by investigating the error exponent
\begin{align} \label{eq:asym_sample_complexity:semi}
  \Erb_k(r) \defeq -\lim_{\eps \rightarrow 0^+} \frac{1}{\eps} \lim_{n \rightarrow \infty} \frac{1}{n} \log \mathbb{P}_{n, m} \left\{ \bbfrob{ \dtmt \Phib_k }^2 - \bbfrob{ \dtmt \Phibt_k }^2 > \eps  \right\},
\end{align}
where the probability is measured over $n$ i.i.d. samples from $P_{XY}$ and the $m = nr$ i.i.d samples from $P_X$. In the following, we develop the error exponent~\eqref{eq:asym_sample_complexity:semi} for the semi-supervised learning, for both cases $\sigma_k > \sigma_{k+1}$ and $\sigma_k = \sigma_{k+1}$, where $\sigma_k$ and $\sigma_{k+1}$ are the $k$-th and $(k+1)$-th largest singular values of $\dtmt$.

\subsection{The Sample Complexity for the Case $\sigma_k > \sigma_{k+1}$}
\label{sec:k>k+1:semi}

Similar to \secref{sec:k>k+1}, we first define the matrix $\Gbh_k(r)$ and the quantity $\alphah(r)$, which will be useful in characterizing the exponent $\Erb_k$.%

\begin{definition}
  \label{def:Gbh:alphah}
  For given $r \geq 0$, the matrix $\Gbh_k(r)$ is defined as
  \begin{align}\label{eq:gbh:def}
    \Gbh_k(r) \defeq \Lbh^{\T}(r) \Gb_k \Lbh(r),
  \end{align}
  where $\Gb_k$ is as defined in \eqref{eq:gb:k}, and $\Lbh(r)$ is an $(|\cX|\cdot|\cY|)\times[|\cX|(|\cY| + 1)]$ matrix with its entry at the $[(x-1)|\cY| + y]$-th row and $[(x'-1)|\cY| + y']$-th column defined as
  \begin{subequations}
    \begin{align}
      \delta_{xx'}\delta_{yy'} - \frac{r}{1 + r}\cdot\sqrt{P_{Y|X}(y|x)P_{Y|X}(y'|x')}\cdot \delta_{xx'},
      \label{eq:Lbh:def:1}
    \end{align}
    and the entry at the $[(x-1)|\cY| + y]$-th row and $[|\cX|\cdot|\cY| + x']$-th column defined as
    \begin{align}
      \frac{\sqrt{r}}{1 + r} \cdot\sqrt{P_{Y|X}(y|x)}\cdot \delta_{xx'},
      \label{eq:Lbh:def:2}
    \end{align}
    \label{eq:Lbh:def}
  \end{subequations}
  for all $x, x' \in \{1, \dots, |\cX|\}$ and $y, y' \in \{1, \dots, |\cY|\}$, where $\delta_{ij}$ denotes the Kronecker delta. Then, $\alphah_k(r)$ %
  is defined as the spectral norm of the matrix $\Gbh_k(r)$.
\end{definition}

\reply{
Then, the following theorem establishes the analytical form of the error exponent $\Erb_k$, whose proof will be later provided.}

\begin{theorem} \label{thm:exponent:semi}
  If $\sigma_k > \sigma_{k + 1}$, the error exponent
  $\Erb_k(r)$ as defined in \eqref{eq:asym_sample_complexity:semi}
  is
  \begin{align} %
    &\Erb_k(r) = - \lim_{\eps \rightarrow 0^+} \frac{1}{\eps} \lim_{n \rightarrow \infty} \frac{1}{n} \log \mathbb{P}_{n,m} \left\{ \bbfrob{ \dtmt \Phib_k }^2 - \bbfrob{ \dtmt \Phibt_k }^2 > \eps  \right\}%
      = \frac{1}{2\alphah_k(r)}.
    \label{eq:exponent:semi}
  \end{align} 
\end{theorem}

Similar to the case of supervised learning, for semi-supervised learning we have the following result establishing the upper bound for the sample complexity of learning maximal correlation functions on large datasets.
\begin{theorem}
  \label{thm:sample:semi}
  \reply{
    For given $P_{XY}$ and $r > 0$, there exists an absolute positive constant and $\bar{\eps}_0$ that depends only on $P_{XY}$ and $r$, such that for all $\eps \in (0, \bar{\eps}_0)$ and $\delta \in (0, 1)$, we have
    \begin{align}
      \mathbb{P}_{n, m} \left\{ \bbfrob{ \dtmt \Phib_k }^2 - \bbfrob{ \dtmt \Phibt_k }^2 > \eps  \right\}  < \delta
      \label{eq:pac:semi}
    \end{align}
    for all $n > \bar{N}^{(4\alphah_k(r))}(\eps, \delta, r)$,
  where we have defined%
  \begin{align}
    \bar{N}^{(t)}(\eps, \delta, r) \defeq  \frac{t|\cX|(1+|\cY|)}{\eps} \log \frac{9t(1+r)|\cX||\cY|}{\eps} + \frac{t}{\eps}\log\frac{1}{\delta},
  \end{align}
  and where  $\alphah_k(r)$ is as defined in \defref{def:Gbh:alphah}.
  }
  \end{theorem}
  \begin{proof}
    See \appref{sec:app:thm:non-asym:semi}.
  \end{proof}

  Furthermore, the performance gain of estimating the maximal correlation functions with the aids of the unlabeled samples can be characterized by the following proposition.

\begin{proposition}
  \label{prop:alphah}
  For all $r \geq 0$, the $\alphah_k(r)$ as defined in \defref{def:Gbh:alphah} %
  is a non-increasing and convex function of $r$, and satisfies
    \begin{align}
      \frac{1}{1 + r} \alphah_k(0) \leq \alphah_k(r) \leq \frac{1}{1 + r} \alphah_k(0) + \frac{r}{1 + r} \alphah_k(\infty),
      \label{eq:alphah:bound}
    \end{align}
    where $\alphah_k(\infty)$ is defined as\footnote{The limit exists since $\alphah_k(r)$ is non-increasing and has a lower bound $0$.} $\alphah_k(\infty) \defeq \lim_{r\to +\infty} \alphah_k(r),$ and we have 
    $\alphah_k(0) = \alpha_k$ with $\alpha_k$ as defined in \defref{def:Gb:alpha}.%

\end{proposition}
\begin{proof}
  See \appref{sec:apph}.
\end{proof}

From Proposition \ref{prop:alphah}, the error exponent $\Erb_k(r) = [2\alphah_k(r)]^{-1}$ of semi-supervised learning is a non-decreasing function of $r$. Thus, with more unlabeled data samples used to train maximal correlation functions, we can obtain better performance. Moreover, it follows immediately from the first inequality of \eqref{eq:alphah:bound} that
\begin{align}
  \Erb_k(r) \leq (1 + r)\Er_k,
  \label{eq:Erb:bound}
\end{align}
where $(1 + r)\Er_k$ 
can be interpreted as the error exponent in the case where we replace all $nr$ unlabeled data samples with labeled data samples and obtain $n(1 + r)$ labeled samples of $(X, Y)$. Therefore, the upper bound \eqref{eq:Erb:bound} simply implies that the labeled data is generally more useful in estimating the maximal correlation functions. However, this upper bound is achievable for certain cases, where the unlabeled data can be as useful as the labeled data, as illustrated in the following proposition (cf. \propref{eg:1}).

\begin{proposition}
  \label{eg:1:semi}  
  If $d = |\cX| \leq |\cY|$, $\sigma_{d - 1} > 0$, and $k = d - 1$%
  , then we have $\alphah_k(r) = \frac{\alpha_k}{1 + r} = \frac{\sigma_1^2}{4(1 + r)}$, and thus
  \begin{align*}
    \Erb_k(r) = (1 + r)\Er_k =  \frac{2(1+r)}{\sigma_1^2}.
  \end{align*}
\end{proposition}

\begin{proof}
  See \appref{sec:app:eg1:semi}.
\end{proof}

In \propref{eg:1} and \propref{eg:1:semi}, we are interested in learning the entire correlation structure between $X$ and $Y$, i.e., $k = d - 1$. In such cases, learning top $d-1$ singular vectors $\Phib_k = [\phib_1, \dots, \phib_{d-1}]$ is equivalent to learning the last singular vector
\begin{align*}
  \phib_d = \left[\sqrt{P_X(1)}, \dots, \sqrt{P_X(d)}\right]^{\T},
\end{align*}
which depends only on the marginal distribution $P_X$. Hence, the unlabeled data samples of $X$ is as useful as the labeled data samples of $(X, Y)$, and thus we can achieve the upper bound of \eqref{eq:Erb:bound}.

\reply{
We then introduce the proof of \thmref{thm:exponent}, which will again make use of the perturbation analyses established in \secref{sec:mpa}. To start, we define the sets of the joint distributions $\Pt_{XY}$ as follows.}
\begin{definition} For all $\eps > 0$, the set $\cSt_1(\eps)$ is defined as
  \label{def:nbhd:S1:semi}
  \begin{align}
    \cSt_1(\eps) \triangleq \left\{ \Pt_{XY}\colon \bbfrob{\dtmt \Phib_k}^2 - \bbfrob{\dtmt \Phibt_k }^2 > \eps  \right\},
    \label{eq:s1_sup:semi}
  \end{align}
  where $\Phibt_k$ corresponds to the top $k$ right singular vectors of $\dtmtt$ as defined in~\eqref{eq:dtm:bar}. Moreover, the set $\nbhdbar(\eps)$ is defined as
  \begin{align}    
    \nbhdbar(\eps) \defeq \left\{\Pt_{XY}\colon
    D(\Ph_{XY}\| P_{XY})+ r D(Q_{X}\| P_{X}) \leq \frac{\eps}{\alphah_k(r)}\right\},
    \label{eq:nbhd:bar}
  \end{align}
  where for given $\Ph_{XY}$ and $Q_X$, the joint distribution $\Pt_{XY}$ is as defined in \eqref{eq:PtXY}.
\end{definition}

Furthermore, for each $\Pt_{XY} \in \nbhdbar(\eps)$ with the corresponding empirical distributions $\Ph_{XY}$ and $Q_X$, we introduce the one-to-one correspondences %
$\Ph_{XY} \leftrightarrow \imate$ and $Q_{X} \leftrightarrow \imatem$, where $\imate(y, x)$ is as defined in \eqref{eq:imate:def}, and where, similarly, we have defined %
\begin{align}\label{eq:Q}
  \imatem(x) \defeq \frac{Q_{X}(x) - P_{X}(x)}{ \sqrt{\eps P_{X}(x)}}.
\end{align}

Moreover, we define $\imatm$ as the $|\cX|$-dimensional vector with the $x$-th entry being $\imatem(x)$, and define the $|\cY| \times |\cX|$ matrix $\dtmdmatt$ with the entries $\dtmdmatet(y,x)$ being %
\begin{align}
  \dtmdmatet(y, x) &\defeq \frac{\sqrt{P_{XY}(x,y)}}{\sqrt{P_X(x)P_Y(y)}} \imates(y, x)%
    -  \frac{P_{XY}(x, y) + P_X(x)P_Y(y)}{2\sqrt{P_X(x)P_Y(y)}}\notag\\
  &\qquad \cdot\left[\frac{1}{P_X(x)}  \sum_{y' \in \cY} \sqrt{P_{XY}(x, y')} \imates(y', x)%
    +\frac{1}{P_Y(y)} \sum_{x' \in \cX}\sqrt{P_{XY}(x', y)}\imates(y, x') \right], \label{eq:theta:semi}
\end{align}
where we have defined
\begin{align}
  \label{eq:imates}%
      \imates(y, x)%
  &\defeq \imatej(y, x) + \frac{r}{1 + r}\sqrt{P_{Y|X}(y|x)}%
    \cdot\left[\imatem(x) - \sum_{y'\in \cY}\sqrt{P_{Y|X}(y'|x)}\imatej(y', x)\right]. %
\end{align}

Then, similar to \lemref{lem:perturb:B}, the matrix $\dtmtt$ estimated from data samples can also be represented in a perturbation form, as the following lemma expresses.%

  \begin{lemma}
    \label{lem:perturb:B:semi}
    For given $P_{XY}$ and $r \geq 0$, there exists a constant $\bar{C} > 0$, such that for all $\eps > 0$ and $\Pt_{XY} \in \nbhdbar(\eps)$, we have $\bbfrob{\dtmdmatt} \leq \bar{C}$ and
    \begin{align}
      \label{eq:perturb:B:semi}
      \dtmtt  = \dtmt + \sqrt{\eps}\, \dtmdmatt + o\left(\sqrt{\eps}\right).
    \end{align}    
  \end{lemma}
  \begin{proof}
    See \appref{sec:app:dtmtt}.
  \end{proof}

  In addition, the following lemma characterizing the error exponent will be useful in our analysis, and can be obtained using Sanov's theorem. 
  \begin{lemma}  
    \label{lem:exponent:kl:semi}
    Given $\eps > 0$, we have
    \begin{align} 
          \mathbb{P}_{n,m} \left\{ \bbfrob{ \dtmt \Phib_k }^2 - \bbfrob{ \dtmt \Phibt_k }^2 > \eps  \right\}%
        &\doteq \exp \biggl\{ -n \cdot \inf_{\Pt_{XY} \in \cSt_1(\eps)} \Bigl[D(\Ph_{XY}\| P_{XY})+ r D(Q_{X}\| P_{X}) \Bigr] \biggr\},
          \label{eq:exponent:kl:semi}
    \end{align}
    where $\cSt_1(\eps)$ is as defined in \eqref{eq:s1_sup:semi}.
  \end{lemma}
  \begin{proof}
    See \appref{sec:app:sanov:semi}.
  \end{proof}

  From \eqref{eq:exponent:kl:semi}, for given $\eps > 0$, the error exponent is determined by the infimum of a weighted sum of K-L divergences. Furthermore, if we restrict our attention to the distributions $\Pt_{XY} \in \nbhdbar(\eps)$, the following result provides a characterization of this infimum in the small $\eps$ regime, and will also be useful in our analysis.

  \begin{lemma}
    \label{lem:S1t:nbhd}
    For $\cSt_1(\eps)$ and $\nbhdbar(\eps)$ as defined in \defref{def:nbhd:S1:semi}, we have  
    \begin{align}
      -\lim_{\eps \rightarrow 0^+} \frac{1}{\eps} \inf_{\Pt_{XY} \in \cSt_1(\eps) \cap \nbhdbar(\eps)} \left[D\bigl(\Ph_{XY}\big\| P_{XY}\bigr) + r D(Q_X \| P_X)\right] = \frac{1}{2\alphah_k(r)}.
      \label{eq:S1t:nbhd}
    \end{align}    

  \end{lemma}
  \begin{proof}
    See \appref{sec:app:S1t:nbhd}.
  \end{proof}

Using \lemref{lem:exponent:kl:semi} and \lemref{lem:S1t:nbhd}, \thmref{thm:exponent:semi} can be established as follows.

\begin{proof}[Proof of \thmref{thm:exponent:semi}]
  From Lemma \ref{lem:exponent:kl:semi}, the error exponent $\Erb_k(r)$ can be expressed as
  \begin{align}
      \Erb_k(r) = -\lim_{\eps \rightarrow 0^+} \frac{1}{\eps} \lim_{n \rightarrow \infty} \frac{1}{n} \log \mathbb{P}_{n, m} \left\{ \bbfrob{ \dtmt \Phib_k }^2 - \bbfrob{ \dtmt \Phibt_k }^2 > \eps  \right\} = \lim_{\eps \rightarrow 0^+} \frac{1}{\eps} \inf_{\Pt_{XY} \in \cSt_1(\eps) \cap \nbhdbar(\eps)} \left[D\bigl(\Ph_{XY}\big\| P_{XY}\bigr) + r D(Q_X \| P_X)\right].
      \label{eq:lim:sanov:semi}
    \end{align}

  Moreover, from Lemma \ref{lem:S1t:nbhd}, there exists an $\eps_0 > 0$ such that for all $\eps \in (0, \eps_0)$, we have
    \begin{align*}
      \inf_{\Pt_{XY} \in \cSt_1(\eps) \cap \nbhdbar(\eps)} \left[D\bigl(\Ph_{XY}\big\| P_{XY}\bigr) + r D(Q_X \| P_X)\right] < \frac{\eps}{\alphah_k(r)}.
    \end{align*}
     In addition, note that for all $\Pt_{XY} \in \cSt_1(\eps) \setminus \nbhdbar(\eps)$ we have $\left[D\bigl(\Ph_{XY}\big\| P_{XY}\bigr) + r D(Q_X \| P_X)\right] > \frac{\eps}{\alphah_k(r)}$. Hence, for all $\eps \in (0, \eps_0)$ we have
     \begin{align*}
       \inf_{\Pt_{XY} \in \cSt_1(\eps)} \left[D\bigl(\Ph_{XY}\big\| P_{XY}\bigr) + r D(Q_X \| P_X)\right]=
       \inf_{\Pt_{XY} \in \cSt_1(\eps) \cap \nbhdbar(\eps)} \left[D\bigl(\Ph_{XY}\big\| P_{XY}\bigr) + r D(Q_X \| P_X)\right],
    \end{align*}
    which implies that
    \begin{align}
      \lim_{\eps \rightarrow 0^+} \frac{1}{\eps} \inf_{\Pt_{XY} \in \cSt_1(\eps)} \left[D\bigl(\Ph_{XY}\big\| P_{XY}\bigr) + r D(Q_X \| P_X)\right] = \lim_{\eps \rightarrow 0^+} \frac{1}{\eps} \inf_{\Pt_{XY} \in \cSt_1(\eps) \cap \nbhdbar(\eps)} \left[D\bigl(\Ph_{XY}\big\| P_{XY}\bigr) + r D(Q_X \| P_X)\right] = \frac{1}{2\alphah_k(r)}.
      \label{eq:lim:S1t}
    \end{align}
    Combining \eqref{eq:lim:sanov:semi} and \eqref{eq:lim:S1t}, we obtain \eqref{eq:exponent:semi}.
\end{proof}

\subsection{The Sample Complexity for the Case $\sigma_k = \sigma_{k+1}$}
\label{sec:exponent:semi:k:k+1}

With ${\cI_k}$, $l$, and $\Phib_{\cI_k}$ as defined in~\eqref{eq:Jk:lmin:Phib}, we further introduce the quantity $\betah(r)$ as follows.
\begin{definition}
  Given $r \geq 0$, $\imat \in \mathbb{R}^{|\cY| \times |\cX|}$, and $\imatm \in \mathbb{R}^{|\cX|}$, the matrix $\Jbh_k(r, \imat, \imatm)$ is defined as
  \begin{align}\label{eq:Jbh}
    \Jbh_k(r, \imat, \imatm) \defeq \Gbh_{l-1}(r) + \Lbh^{\T}(r)\Lb^{\T}\left(\sum_{i = l}^k\sum_{j \in \cIb_k} \frac{\varthetabb_{ij}\varthetabb^{\T}_{ij}}{\sigma_i^2 - \sigma_j^2}\right)\Lb\Lbh(r), 
  \end{align}
  where $\Gbh_{l - 1}$ and $\Lbh(r)$ are as defined in \eqref{eq:gbh:def}--\eqref{eq:Lbh:def}, and $\Lb$ is as defined in \eqref{eq:L:def}
  and $\varthetabb_{ij}$ are defined as, for all $i, j$, $\varthetabb_{ij} \defeq \phib_j \otimes \bigl(\dtmt\varphibb_i\bigr) + \varphibb_i \otimes \bigl(\dtmt\phib_j\bigr)$, where $\varphibb_i$ are defined as %
  \begin{align} \label{eq:varphi:def:semi}
    \varphibb_i \defeq \Phib_{\cI_k}\ubb_{i - l + 1},\quad l \leq i \leq k,
  \end{align}
  and where $\ubb_1, \dots, \ubb_{k - l + 1} \in \mathbb{R}^{|\cI_k|}$ are the top $k - l + 1$ eigenvectors of the matrix $\Phib_{\cI_k}^{\T}\left(\dtmt^{\T}\dtmdmatt + \dtmdmatt^{\T}\dtmt\right)\Phib_{\cI_k}$. Then, $\betah_k(r)$ is defined as the optimal value of the optimization problem
  \begin{subequations}
    \begin{alignat}{2}
      &\max_{\convec} &\quad &\convec^{\T}\Jbh_k(r, \imat, \imatm)\convec\\
      &~\,\mathrm{s.t.}& & \|\convec\|^2 \leq 1,
    \end{alignat}
    \label{eq:opt:semi:k:k+1:simple}
  \end{subequations}
    where $\convec \in \mathbb{R}^{|\cX|(|\cY| + 1)}$ is defined as
    \begin{align}
       \convec \defeq
      \begin{bmatrix}
        \vec(\imatj)\\
        \sqrt{r}\imatm
      \end{bmatrix}
      .
      \label{eq:def:convec}
    \end{align}
\end{definition}

Then we have the following result characterizing the error exponent \eqref{eq:asym_sample_complexity:semi}, and the corresponding upper bound of sample complexity for large datasets can be established similar to \thmref{thm:sample:semi}.

\begin{theorem}
  \label{thm:exponent:semi:k:k+1}
    If $\sigma_k = \sigma_{k + 1}$,
    the error exponent
    $\Erb_k(r)$ as defined in \eqref{eq:asym_sample_complexity:semi}
    is
    \begin{align} 
        \Erb_k(r) = -\lim_{\eps \rightarrow 0^+} \frac{1}{\eps} \lim_{n \rightarrow \infty} \frac{1}{n} \log \mathbb{P}_{n, m} \left\{ \bbfrob{ \dtmt \Phib_k }^2 - \bbfrob{ \dtmt \Phibt_k }^2 > \eps  \right\}%
        = \frac{1}{2\betah_k(r)}.
      \label{eq:exponent:semi:k:k+1}
    \end{align} 

\end{theorem}
\begin{proof}
  See \appref{sec:app:thm:semi:k:k+1}.
\end{proof}

Note that $\Jbh_k$ in~\eqref{eq:Jbh} is dependent on $\convec$, since $\varthetabb_{ij}$ is dependent on $\dtmdmatt$. Therefore, unlike Theorem~\ref{thm:exponent:semi}, the optimal value of~\eqref{eq:opt:semi:k:k+1:simple} is not simply the largest singular value of some given matrix. However, if we fix $\Jbh_k$, the optimization problem \eqref{eq:opt:semi:k:k+1:simple} is reduced to solving the largest singular value of $\Jbh_k$. As a result, similar to the approach introduced in \secref{sec:exponent:k:k+1}, we can alternatively solve the optimal $\convec$ and $\Jbh_k$, as summarized in Algorithm~\ref{alg:alphah}.

\begin{algorithm}[t]
  \caption{An algorithm for computing the error exponent of \thmref{thm:exponent:semi:k:k+1}}%
  \label{alg:alphah}
  \begin{algorithmic}[1]
    \State {\bfseries Input:} $k$, $r$, and a learning rate $\eta$.
    \State Compute $\cI_k$, $l$, and $\Phib_{\cI_k}$.
    \State Randomly choose a $\convec \in \mathbb{R}^{|\cX|( |\cY| + 1)}$ with $\|\convec\|^2 = 1$.
    \Repeat
    \State Compute $\imatj$ and $\imatm$ from \eqref{eq:def:convec}.
    \State Compute $\dtmdmatt$ from %
    \eqref{eq:theta:semi}.
    \State $\Wb \gets \Phib_{\cI_k}^{\T}\left(\dtmt^{\T}\dtmdmatt + \dtmdmatt^{\T}\dtmt\right)\Phib_{\cI_k}$
    \For {$i = 1, \dots, k - l + 1$}
    \State $\ubb_i \gets$ the $i$-th eigenvector of $\Wb$
    \State $\varphibb_{i+l-1} \gets \Phib_{\cI_k}\ubb_{i}$
    \EndFor
    \State Compute $\Jbh_k(r, \imat, \imatm)$ from \eqref{eq:Jbh}.
    \State $\betah_k(r) \gets \convec^{\T}{\Jbh_k(r, \imat, \imatm)}\convec$
    \State Compute the eigenvectors $\qb_1, \dots, \qb_s$ of $\Jbh_k(r, \imat, \imatm)$ associated with its largest eigenvalue
    \State $\Qb \gets \left[{\qb}_1, \dots, {\qb}_s\right]$ 
    \State $\convec \gets \convec + \eta\Qb\Qb^{\T}\convec$%
    \State $\displaystyle\convec \gets \convec/\|\convec\|$
    \Until{$\betah_k(r)$ converges.}
    \State $\Erb_k(r) \gets \left[2\betah_k(r)\right]^{-1}$
    \State {\bfseries Output:} $\Erb_k(r)$.
  \end{algorithmic}
\end{algorithm}

Similar to Corollary \ref{eg:2}, we can compute the sample complexity in closed form for some joint distributions.

\begin{corollary}
  \label{eg:2:semi}
  For the joint distribution $P_{XY}$ as constructed in Corollary \ref{eg:2}, all non-zero singular values of the corresponding $\dtmt$ are $\sigma_1 = \sigma_2 = \dots = \sigma_{d - 1}$ as given by \eqref{eq:sigma:eg2}. Then, for all $k \in [d - 1]$, we have $\betah_k(r) = \frac{\sigma_1^2}{4(1 + r)}$, and thus the corresponding error exponent is
  \begin{align*}
    \Erb_k(r) = \frac{2(1 + r)}{\sigma_1^2}.%
  \end{align*}
\end{corollary}

\begin{proof}
  See \appref{sec:app:eg2:semi}.
\end{proof}

\subsection{The Optimal Number of Samples with the Cost Constraint}

In semi-supervised learning, while the labeled samples are more useful than the unlabeled samples in learning problems, it is often much more expensive to acquire the labeled samples than the unlabeled ones. Therefore, it is important to understand the fundamental tradeoff between the sampling cost and the performance in learning tasks. In the following, we investigate such tradeoff for the sample complexity of learning the maximal correlation functions.

Suppose that the costs of acquiring the labeled and unlabeled samples are $\mathsf{C}_{\ell}$ and $\mathsf{C}_{u}$ per sample, respectively, and the total budget for sampling is $\mathsf{C}$. Then, the number of labeled samples $n_{\ell}$ and the number of unlabeled samples $n_u$ we can get are constrained by $n_{\ell} \mathsf{C}_{\ell} + n_u \mathsf{C}_{u} \leq \mathsf{C}$. Without loss of generality, we consider the case $\sigma_k > \sigma_{k + 1}$, and it follows from Theorem~\ref{thm:exponent:semi} that 
the error exponent for estimating the $k$-dimensional maximal correlation functions with these samples is $\eps n_{\ell} / [2\alphah_k(r)]$, %
where $r = n_u / n_{\ell}$. Hence, the optimal error exponent that can be achieved by the sampling budget $\mathsf{C}$ is given by
\begin{align*}
\max_{n_{\ell} \mathsf{C}_{\ell} + n_u \mathsf{C}_{u} \leq \mathsf{C}}\, \frac{\eps n_{\ell}}{2\alphah_k(r)} = \max_{r \geq 0}\, \frac{\eps \mathsf{C}}{2 (\mathsf{C}_{\ell} + r\mathsf{C}_u) \alphah_k(r) },
\end{align*}
which immediately implies the following proposition.
\begin{proposition}
Given the sampling budget constraint $\mathsf{C}$, the optimal number of labeled samples $n_{\ell}$ and unlabeled samples $n_u$ to optimize the sample complexity of estimating the $k$-dimensional maximal correlation functions are
\begin{align*}
n_{\ell} = \frac{\mathsf{C}}{\mathsf{C}_{\ell} + r^* \mathsf{C}_u}, \quad n_u = \frac{r^*\mathsf{C}}{\mathsf{C}_{\ell} + r^* \mathsf{C}_u},
\end{align*}
where 
\begin{align} \label{eq:semi_design}
r^* = \mathop{\argmin}_{r \geq 0} \ (\mathsf{C}_{\ell} + r\mathsf{C}_u) \alphah_k(r).
\end{align}
\end{proposition}
Note that the optimal ratio $r^*$ is independent of $\mathsf{C}$, which indicates the relative importance of the unlabeled samples compared to the labeled samples, by taking the sampling costs into account. While the optimization problem~\eqref{eq:semi_design} has no analytical solution and is neither convex nor concave, we can solve the local optimum by the numerical differentiation approach~\cite{heath2018scientific}. In particular,
the local optimum of $r$ can be computed via the updating rule 
\begin{align*}
  r \gets r - \frac{\eta}{h} \left[(\mathsf{C}_{\ell} + (r+h)\mathsf{C}_u) \alphah_k(r + h) - (\mathsf{C}_{\ell} + r\mathsf{C}_u) \alphah_k(r)\right],
\end{align*}
where $h > 0$ is the step size for computing the numerical differentiation, and $\eta > 0$ is the learning rate for gradient descent.

\section{The Numerical Simulations} \label{sec:ns}

In this section, we validate our theoretical results by some numerical simulations. In our experiments, we choose $| \cX | = |\cY | = 4$, and the joint distribution $P_{XY}$ as
\begin{align} \label{eq:pxyexp}
  P_{XY}(x, y) =
  \begin{cases}
    \frac{1}{8},&\text{if}~x = y,\\
    \frac{1}{24},&\text{if}~x \neq y.
  \end{cases}
\end{align}
In the following, we compare the empirical error exponents~\eqref{eq:asym_sample_complexity}
and \eqref{eq:asym_sample_complexity:semi} for estimating $k=2$ dimensional maximal correlation functions with the theoretical results. Note that since the joint distribution $P_{XY}$ of~\eqref{eq:pxyexp} is a special case of Corollary \ref{eg:2} and \ref{eg:2:semi}, we can apply the results from the corollaries as our theoretical benchmarks.

\subsection{Supervised Learning}
\label{sec:sim:supervised}

In this experiment, we sample the learning error $\bbfrob{\dtmt \Phib_k}^2 - \bbfrob{\dtmt \Phih_k}^2$ as follows. For each sample of the learning error, we first generate $n = 10^6$ pairs of $(x_i , y_i)$, i.i.d. from $P_{XY}$, and then compute the $\dtmh$ from the empirical distribution $\Ph_{XY}$ of these $n$ pairs. Then, the singular vectors of $\dtmh$ are computed to get a sample of $\bbfrob{\dtmt \Phib_k}^2 - \bbfrob{\dtmt \Phih_k}^2$. We repeat this sampling process for the learning error for $10^5$ times, and consider the empirical probability  %
\begin{align*}
p_n(\eps) = \mathbb{P} \left\{\bbfrob{\dtmt \Phib_k}^2 - \bbfrob{\dtmt \Phih_k}^2 > \eps\right\}
\end{align*}
over the $10^5$ samples. Then, the empirical error exponent can be computed as
\begin{align*}
  -\frac{1}{n \eps} \log p_n(\eps).
\end{align*}
The comparison between the empirical error exponent and the error exponent computed from Corollary \ref{eg:2} %
is plotted in~\figref{fig:exponent}, in which we can see the coincidence between these two error exponents.

\begin{figure}[t]
  \centering
  \includegraphics[width = .5\textwidth]{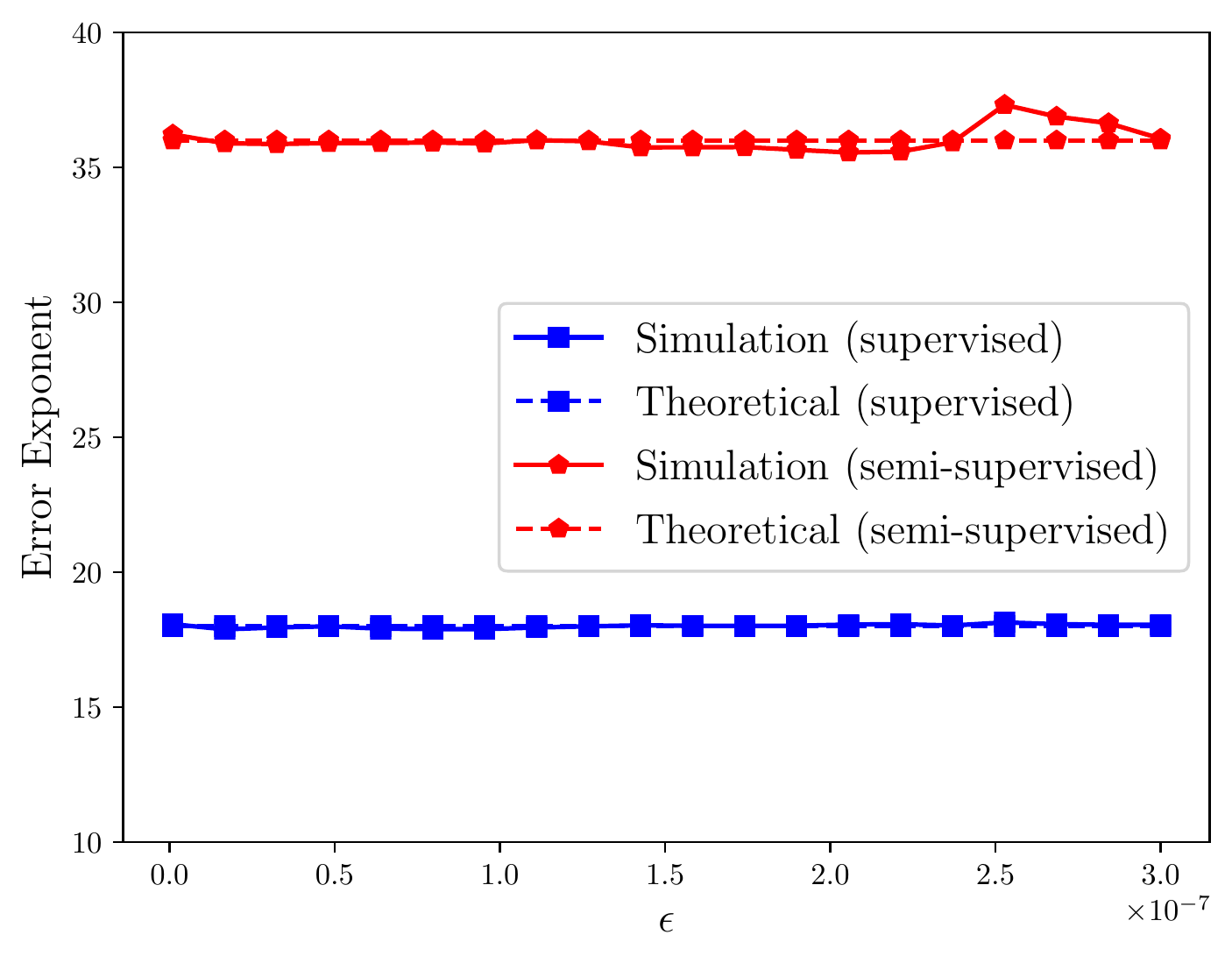}
  \caption{The comparison of theoretical and experimental error exponents in both supervised and semi-supervised scenarios}
  \label{fig:exponent}
\end{figure}

\subsection{Semi-supervised Learning}
\label{sec:sim:semi}

In the experiment for the semi-supervised learning, for each sample of learning error, we take $r = 1$, and generate $n = 10^6$ pairs of $(x_i , y_i)$, i.i.d. from $P_{XY}$, and $m = nr = 10^6$ of $x_j$, i.i.d. from $P_{X}$, and then compute the $\dtmtt$ according to the empirical distribution $\Pt_{XY}$ from~\eqref{eq:PtXY}. This sampling process for the learning error is repeated for $10^5$ times, and the empirical probability of the learning error exceeding $\eps$ is defined as
\begin{align*}
\bar{p}_{n, m}(\eps) = \mathbb{P} \left\{\bbfrob{\dtmt \Phib_k}^2 - \bbfrob{\dtmt \Phibt_k}^2 > \eps\right\}.
\end{align*}
Then, the empirical error exponent can be computed as
\begin{align*}
  -\frac{1}{n \eps} \log \bar{p}_{n, m}(\eps).
\end{align*}
The comparison between the empirical error exponent and the error exponent computed from Corollary \ref{eg:2:semi} %
is plotted in~\figref{fig:exponent}, in which we can see the coincidence between these two error exponents.

\appendices
\reply{
\section{Alternating Conditional Expectation Algorithm (\algoref{alg:ace})}
\label{sec:app:ace}
For convenience, we assume that the empirical distribution $\Ph_{XY} = P_{XY}$ and thus $\dtmh = \dtmt$. Let $\sigma_i$ denote the $i$-th singular value of $\dtmt$, we will show that \algoref{alg:ace} converges to the maximal correlation functions $f^*$ and $g^*$ which achieve the maximal correlation $\rho_k(X; Y) = \sum_{i = 1}^{k} \sigma_i$.

To begin, let $\Phib_k$ and $\Psib_k$ be the matrix composed of top $k$ right singular vectors and left singular vectors of $\dtmt$, respectively, as defined in \secref{sec:pf}.
Then, from the analyses in \secref{sec:pf}, we know that after the alternating conditional expectation processes (cf. line \ref{alg:ace:line:begin}--\ref{alg:ace:line:end} of \algoref{alg:ace}), $\Phih_k$ and $\Phib_k$ have the same column space. Then, after the whitening of $\fh$ in line \ref{alg:ace:line:norm:f}, we have
\begin{align}
  \Phih_k = \Phib_k \Qb,
  \label{eq:phih}
\end{align}
where $\Qb \in \mathbb{R}^{k\times k}$ is an orthogonal matrix. Then, with line \ref{alg:ace:line:g}--\ref{alg:ace:line:norm:g}, we obtain
\begin{align}
  \Phih_k &= \dtmt \Phih_k ((\dtmt \Phih_k)^{\T}\dtmt \Phih_k)^{-1/2}\\
  &= \Psib_k \Sigmab_k \Qb (\Qb^\T \Sigmab_k \Psib_k^\T\Psib_k \Sigmab_k \Qb)^{-1/2}\label{eq:comp:phi:1}\\
  &= \Psib_k \Sigmab_k \Qb (\Qb^\T \Sigmab_k^2 \Qb)^{-1/2}\label{eq:comp:phi:2}\\
  &= \Psib_k \Sigmab_k \Qb (\Qb^\T \Sigmab_k^{-1} \Qb)\label{eq:comp:phi:3}\\
  &= \Psib_k \Qb \label{eq:comp:phi:4}
\end{align}
where \eqref{eq:comp:phi:1} follows from the fact that $\dtmt \Phih_k = \dtmt\Phib_k \Qb = \Psib_k \Sigmab_k \Qb$, where $\Sigmab_k = \diag\{\sigma_1, \dots, \sigma_k\}$. In addition, \eqref{eq:comp:phi:2} follows from the fact that $\Phib_k^\T\Phib_k = \Ib_k$, and \eqref{eq:comp:phi:4} follows from that $\Qb_k\Qb_k^\T = \Ib_k$ since $\Qb$ is orthogonal, with $\Ib_k$ representing the identity matrix of order $k$.

From \eqref{eq:phih} and \eqref{eq:comp:phi:4}, we know that
\begin{align*}
  \E{\fh(X)\fh^\T(X)} = \Phih_k^\T \Phih_k = \Ib_k \quad\text{and}\quad   \E{\gh(Y)\gh^\T(Y)} = \Psih_k^\T \Psih_k = \Ib_k.
\end{align*}
Finally, we have
\begin{align*}
  \E{\fh(X)\gh(Y)} = \trop{\Psih^\T\dtmt\Phih} = \trop{\Qb^\T\Psib_k^\T\dtmt\Phib_k\Qb} = \trop{\Psib_k^\T\dtmt\Phib_k} = \trop{\Sigmab_k} = \sum_{i = 1}^k \sigma_i = \rho_k(X; Y).
\end{align*}
As a consequence,  $\fh$ and $\gh$ are the maximal correlation functions that achieve the HGR maximal correlation $\rho_k(X; Y)$.
}
\reply{
\section{Proof of Eq. \eqref{eq:metric}} \label{sec:app:metric}
    Firstly, note that
    \begin{align*}
      \norm{\dtmt \phibh_1}^2
      =  \sum_{i = 1}^d  \sigma_i^2 \ip{\phibh_1}{\phib_i}^2
      &\leq \sigma_1^2\ip{\phibh_1}{\phib_1}^2 + \sigma_2^2\sum_{i = 2}^d   \ip{\phibh_1}{\phib_i}^2\\
      &= (\sigma_1^2 - \sigma_2^2)\ip{\phibh_1}{\phib_1}^2 + \sigma_2^2\sum_{i = 1}^d   \ip{\phibh_1}{\phib_i}^2\\
      &\leq (\sigma_1^2 - \sigma_2^2)\ip{\phibh_1}{\phib_1}^2 + \sigma_2^2,
    \end{align*}
    where the second equality follows from the fact that
    \begin{align*}
      \sum_{i = 1}^d   \ip{\phibh_1}{\phib_i}^2 \leq \norm{\phibh_1}^2 = 1
    \end{align*}
    since $\phib_1, \dots, \phib_d$ form an orthonormal set. Therefore, we have
    \begin{align*}
      \norm{\dtmt \phib_1}^2 - \norm{\dtmt \phibh_1}^2
      &\geq (\sigma_1^2 - \sigma_2^2) \left( 1 - \ip{\phibh_1}{\phib_1}^2\right)\\
      &= (\sigma_1^2 - \sigma_2^2) \left( 1 - \ip{\phibh_1}{\phib_1}\right) \left( 1 + \ip{\phibh_1}{\phib_1}\right)\\
      &\geq (\sigma_1^2 - \sigma_2^2)\left( 1 - \ip{\phibh_1}{\phib_1}\right)\\
      &= \frac{\sigma_1^2 - \sigma_2^2}{2}\norm{\phib_1 - \phibh_1}^2,
    \end{align*}
    where the second inequality follows from the assumption that $\ip{\phibh_1}{\phib_1} \geq 0$. As a consequence, we obtain \eqref{eq:metric} as desired.
}
\section{Proof of Lemma~\ref{lem:eig:k}} \label{sec:appa}
  Suppose $\lambda_1, \dots, \lambda_k$ take $q$ distinct values, and the indices $i_0, \dots, i_{q}$ are defined such that $0 = i_0 < \dots < i_q = k$ and 
  \begin{align*}
    \lambda_{i_0+1} = \lambda_{i_1} > \lambda_{i_1+1} = \lambda_{i_2} > \dots > \lambda_{i_{q-1}+1} = \lambda_{i_{q}} > \lambda_{i_{q} + 1}.
  \end{align*}
  Therefore, we have
  \begin{align}
    \tr \left\{\Vb_k^\T(\pv)\Ab\Vb_k(\pv)\right\}
    &= \sum_{j = 1}^k \vb_j^\T(\pv)\Ab\vb_j(\pv)%
      = \sum_{s = 1}^q\sum_{i = i_{s-1}+1}^{i_{s}} \vb_i^\T(\pv)\Ab\vb_i(\pv), \label{eq:tr:vk}      
  \end{align}
  where $\vb_j(\pv)$ denotes the $j$-th column of $\Vb_k$. We first consider the summation for $s = 1$,
  \begin{align*}
    \sum_{i = i_0 + 1}^{i_1} \vb_i^\T(\pv)\Ab\vb_i(\pv) = \tr\left\{\Vb_{i_1}^{\T}(\pv)\Ab\Vb_{i_1}(\pv)\right\},
  \end{align*}
  where $\Vb_{i_1}(\pv) \in \mathbb{R}^{d \times i_1}$ is composed of the first $i_1$ columns of $\Vb_{k}(\pv)$. First, note that since $\Ab(\pv)$ is analytic, there exists a symmetric matrix $\Ab''$ such that
  \begin{align*}
    \Ab(\pv) = \Ab + \pv \Ab' + \pv^2 \Ab'' + o\bigl(\pv^2\bigr).
  \end{align*}
  In addition, the analyticity of $\Ab(\pv)$ implies that the eigenspace $\Vb_{k}(\pv)$ is also analytic \cite{kato1976perturbation}, and thus has
  the expansion
  \begin{align*}
    \Vb_{i_1}(\pv) = \Vbh_{i_1} + \pv\Vb_{i_1}' + \pv^2\Vb_{i_1}'' + o\bigl(\pv^2\bigr),    
  \end{align*}
  where $\Vbh_{i_1}$, $\Vb_{i_1}'$, and $\Vb_{i_1}''$ are matrices in $\mathbb{R}^{d \times i_1}$. Moreover, the columns of $\Vbh_{i_1}$ form an orthonormal basis of the eigenspace of $\Ab$ associated with $\lambda_1$, and thus $\Ab\Vbh_{i_1} = \lambda_1 \Vbh_{i_1}$. Then, from $\Vb_{i_1}^{\T}(\pv)\Vb_{i_1}(\pv) = \Ib_{i_1}$, we obtain
  \begin{align*}
    \Ib_{i_1}
    &= \Vb_{i_1}^{\T}(\pv)\Vb_{i_1}(\pv)%
      = \Vbh_{i_1}^{\T}\Vbh_{i_1} + \pv\bigl(\Vbh_{i_1}^{\T}\Vb_{i_1}' + \Vb_{i_1}'^{\T}\Vbh_{i_1}\bigr)%
      + \pv^2\bigl(\Vbh_{i_1}^{\T}\Vb_{i_1}'' + \Vb_{i_1}''^{\T}\Vbh_{i_1} + \Vb_{i_1}'^{\T}\Vb'_{i_1}\bigr) + o(\pv^2),
  \end{align*}
  which in turn implies
  \begin{subequations}
    \begin{gather}
      \Vbh_{i_1}^{\T}\Vb_{i_1}' + \Vb_{i_1}'^{\T}\Vbh_{i_1} = \Ob_{i_1},\label{eq:ortho:1}\\
      \Vbh_{i_1}^{\T}\Vb_{i_1}'' + \Vb_{i_1}''^{\T}\Vbh_{i_1} + \Vb_{i_1}'^{\T}\Vb'_{i_1} = \Ob_{i_1},
      \label{eq:ortho:2}
    \end{gather}
    \label{eq:ortho}
  \end{subequations}
  where $\Ib_{i_1}$ and $\Ob_{i_1}$ are the identity matrix and the zero matrix in $\mathbb{R}^{i_1 \times i_1}$, respectively. Therefore, we have
  \begin{align}
      \Vb_{i_1}^{\T}(\pv)\Ab\Vb_{i_1}(\pv)%
    &= \left(\Vbh_{i_1} + \pv\Vb_{i_1}' + \pv^2\Vb_{i_1}''\right)^{\T}\Ab\left(\Vbh_{i_1} + \pv\Vb_{i_1}' + \pv^2\Vb_{i_1}''\right)%
      + o(\pv^2)\notag\\
    &= \Vbh_{i_1}^{\T}\Ab\Vbh_{i_1} + \pv \Vbh_{i_1}^\T\Ab\Vb'_{i_1} + \Vb_{i_1}'^\T\Ab\Vbh_{i_1}%
      + \pv^2 \Vbh_{i_1}^{\T}\Ab\Vb_{i_1}'' + \Vb_{i_1}'^{\T}\Ab\Vb_{i_1}' + \Vb_{i_1}''^{\T}\Ab\Vbh_{i_1} + o(\pv^2)\notag\\
    &= \lambda_1 \cdot \Ib_{i_1} + \lambda_1\cdot\pv \Vbh_{i_1}^\T\Vb'_{i_1} + \Vb_{i_1}'^\T\Vbh_{i_1}%
      + \pv^2 \lambda_1\Vbh_{i_1}^{\T}\Vb_{i_1}'' + \Vb_{i_1}'^{\T}\Ab\Vb_{i_1}' + \lambda_1\Vb_{i_1}''^{\T}\Vbh_{i_1} + o(\pv^2)\notag\\
    &= \lambda_1 \cdot \Ib_{i_1} - \pv^2 \Vb_{i_1}'^{\T}(\lambda_1\Ib_{d} - \Ab)\Vb_{i_1}' + o(\pv^2),  \label{eq:tr:i1}
  \end{align}
  where the penultimate equality follows from the fact that $\Ab\Vbh_{i_1} = \lambda_1\Vbh_{i_1}$, the last equality follows from \eqref{eq:ortho}, and $\Ib_d$ is the identity matrix in $\mathbb{R}^{d\times d}$.

  In addition, we define the matrix
  \begin{align*}
  \Lambdab_{i_1}(\pv) \defeq \diag\{\lambda_1(\pv), \dots, \lambda_{i_1}(\pv)\},
  \end{align*}
  where $\lambda_1(\pv), \dots, \lambda_{i_1}(\pv)$ are the largest $i_1$ eigenvalues of $\Ab(\pv)$. Then, it follows from the analyticity of $\Ab(\pv)$ that $\Lambdab_{i_1}(\pv)$ is analytic and can be written as
  \begin{align*}
    \Lambdab_{i_1}(\pv) = \lambda_1 \Ib_{i_1} + \pv \Lambdab_{i_1}' + \pv^2\Lambdab_{i_1}'' + o(\pv^2),
  \end{align*}
  where $\Lambdab_{i_1}'$ and $\Lambdab_{i_1}''$ are both diagonal matrices. Now, from $\Ab(\pv) \Vb_{i_1}(\pv) = \Vb_{i_1}(\pv) \Lambdab_{i_1}(\pv)$ we obtain
  \begin{align*}
    &\left(\Ab + \pv \Ab' + \pv^2 \Ab''\right)\left(\Vbh_{i_1} + \pv\Vb_{i_1}' + \pv^2\Vb_{i_1}''\right)%
      = \left(\Vbh_{i_1} + \pv\Vb_{i_1}' + \pv^2\Vb_{i_1}''\right)\left(\lambda_1 \Ib_{i_1} + \pv \Lambdab_{i_1}' + \pv^2\Lambdab_{i_1}''\right) + o(\pv^2).
  \end{align*}
  Comparing the $\pv$-order terms for both sides, we have
  \begin{align*}
    \Ab'\Vbh_{i_1} + \Ab\Vb_{i_1}' = \lambda_1\Vb_{i_1}' + \Vbh_{i_1}\Lambdab_{i_1}',
  \end{align*}
  and thus
  \begin{align}
    (\lambda_1\Ib_d- \Ab)\Vb_{i_1}' =  \Ab'\Vbh_{i_1} - \Vbh_{i_1}\Lambdab_{i_1}'.
    \label{eq:eig:v}
  \end{align}
  Left multiplying \eqref{eq:eig:v} by $\Vbh_{i_1}^\T$, we obtain
  \begin{align}
    \Lambdab_{i_1}' = \Vbh_{i_1}^\T\Ab'\Vbh_{i_1},    
    \label{eq:lambda:1}
  \end{align}
  where we have again exploited the fact that $\Ab\Vbh_{i_1} = \lambda_1\Vbh_{i_1}$.

  Now, we can rewrite $\left[\Vb_{i_1}'^{\T}(\lambda_1\Ib_{d} - \Ab)\Vb_{i_1}'\right]$ of \eqref{eq:tr:i1} as
  \begin{subequations}
  \begin{align}
      \Vb_{i_1}'^{\T}(\lambda_1\Ib_{d} - \Ab)\Vb_{i_1}'\notag%
    &= \left[(\lambda_1\Ib_{d} - \Ab)\Vb_{i_1}'\right]^{\T}\Vb_{i_1}'\\
    &= \left(\Ab'\Vbh_{i_1} - \Vbh_{i_1}\Lambdab_{i_1}'\right)^{\T}\Vb_{i_1}'\label{eq:eig:tr:1}\\
    &= \left( \Vbh_{i_1}^{\T}\Ab' - \Lambdab_{i_1}'\Vbh_{i_1}^{\T} \right) \Vb_{i_1}' \label{eq:eig:tr:2}\\
    &= \Vbh_{i_1}^\T \Ab'\left[\left(\Ib_d - \Vbh_{i_1}\Vbh_{i_1}^\T\right) + \Vbh_{i_1}\Vbh_{i_1}^\T\right]\Vb'_{i_1}  - \Lambdab_{i_1}'\Vbh_{i_1}^{\T}\Vb'_{i_1} \notag \\
    &= \Vbh_{i_1}^\T \Ab'\left(\Ib_d - \Vbh_{i_1}\Vbh_{i_1}^\T\right)\Vb'_{i_1} + \left(\Vbh_{i_1}^\T \Ab'\Vbh_{i_1}- \Lambdab_{i_1}'\right)\Vbh_{i_1}^{\T}\Vb'_{i_1}\notag\\
    &= \Vbh_{i_1}^\T \Ab'\left(\Ib_d - \Vbh_{i_1}\Vbh_{i_1}^\T\right)\Vb'_{i_1},\label{eq:eig:tr:5}
  \end{align}
  \label{eq:eig:tr}
\end{subequations}
where \eqref{eq:eig:tr:1} follows from \eqref{eq:eig:v}, and \eqref{eq:eig:tr:5} follows from \eqref{eq:lambda:1}. Furthermore, it follows from the eigen-decomposition of $\Ab$ that
\begin{align*}
  \Ab = \lambda_1\sum_{j = 1}^{i_1} \vbh_j\vbh_j^{\T} + \sum_{j = i_1 + 1}^d \lambda_j\vb_j \vb_j^{\T},
\end{align*}
where $\vbh_j$ is the $j$-th column of $\Vbh_{i_1}$~($1 \leq j \leq i_1$). Similarly, we have
\begin{align*}
  \Ib_d = \sum_{j = 1}^{i_1} \vbh_j\vbh_j^{\T} + \sum_{j = i_1 + 1}^d \vb_j \vb_j^{\T}.
\end{align*}
Hence, we obtain
\begin{align*}
  \lambda_1\Ib_d - \Ab  = \sum_{j = i_1 + 1}^d (\lambda_1 - \lambda_j)\vb_j \vb_j^{\T},
\end{align*}
and its Moore-Penrose inverse
\begin{align}
  (\lambda_1\Ib_d - \Ab)^{\dagger} = \sum_{j = i_1 + 1}^d \frac{\vb_j \vb_j^{\T}}{\lambda_1 - \lambda_j}.
  \label{eq:mp:inv}
\end{align}
Therefore, we have
\begin{align*}
  \Ib_d - \Vbh_{i_1}\Vbh_{i_1}^{\T} = \sum_{j = i_1 + 1}^d \vb_j \vb_j^{\T} = (\lambda_1\Ib_d - \Ab)^{\dagger}(\lambda_1\Ib_d - \Ab),
\end{align*}
and hence
\begin{subequations}
  \begin{align}
      \left(\Ib_d - \Vbh_{i_1}\Vbh_{i_1}^\T\right)\Vb'_{i_1}%
    &= (\lambda_1\Ib_d - \Ab)^{\dagger}(\lambda_1\Ib_d - \Ab)\Vb'_{i_1}\\
    &= (\lambda_1\Ib_d - \Ab)^{\dagger}\Ab'\Vbh_{i_1} - (\lambda_1\Ib_d - \Ab)^{\dagger} \Vbh_{i_1}\Lambdab_{i_1}'     \label{eq:eig:mp:3}\\
    &= (\lambda_1\Ib_d - \Ab)^{\dagger}\Ab'\Vbh_{i_1},     \label{eq:eig:mp:4}
  \end{align}
    \label{eq:eig:mp}
  \end{subequations}
where to obtain \eqref{eq:eig:mp:3} we have used \eqref{eq:eig:v}, and to obtain \eqref{eq:eig:mp:4} we have used the fact that $(\lambda_1\Ib_d - \Ab)^{\dagger} \Vbh_{i_1}$ is a zero matrix [cf. \eqref{eq:mp:inv}], since $\vbh_i$ is orthogonal to $\vb_{j}$ for all $i \leq i_1 < j$.

Then, from \eqref{eq:tr:i1}, \eqref{eq:eig:tr} and \eqref{eq:eig:mp}, we obtain
\begin{align}
  &\eqspace\Vb_{i_1}^{\T}(\pv)\Ab\Vb_{i_1}(\pv)%
      = \lambda_1 \cdot \Ib_{i_1} - \pv^2 \Vbh_{i_1}^\T \Ab' ( \lambda_1\Ib_d - \Ab)^{\dagger}\Ab'\Vbh_{i_1} + o(\pv^2),
      \label{eq:V.A.T}
\end{align}
which implies that
\begin{subequations}
\begin{align}
    \tr\left\{\Vb_{i_1}^{\T}(\pv)\Ab\Vb_{i_1}(\pv)\right\}%
  &= \lambda_1 \cdot i_1 - \pv^2 \tr\left\{\Vbh_{i_1}^\T \Ab' ( \lambda_1\Ib_d - \Ab)^{\dagger}\Ab'\Vbh_{i_1}\right\} + o(\pv^2)\notag\\
  &= \lambda_1 \cdot i_1 - \pv^2 \sum_{i = 1}^{i_1} \vbh_i^\T \Ab' (\lambda_1\Ib_d - \Ab)^{\dagger} \Ab'\vbh_i + o(\pv^2)\notag\\
  &= \lambda_1 \cdot i_1 - \pv^2 \sum_{i = 1}^{i_1}\sum_{j= i_1 + 1}^d \frac{\left(\vbh_i^{\T}\Ab'\vb_{j}\right)^2}{\lambda_{1} - \lambda_{j}} + o(\pv^2)\label{eq:tr:eps:1}\\
  &= \lambda_1 \cdot i_1 - \pv^2 \sum_{j= i_1 + 1}^d \frac{\left\|\Vbh_{i_1}^{\T}\Ab'\vb_{j}\right\|^2}{\lambda_{1} - \lambda_{j}} + o(\pv^2)\\
  &= \lambda_1 \cdot i_1 - \pv^2 \sum_{j= i_1 + 1}^d \frac{\left\|\Vb_{i_1}^{\T}\Ab'\vb_{j}\right\|^2}{\lambda_{1} - \lambda_{j}} + o(\pv^2)\label{eq:tr:eps:3}\\
  &= \lambda_1 \cdot i_1 - \pv^2 \sum_{i = 1}^{i_1}\sum_{j= i_1 + 1}^d \frac{\left(\vb_i^{\T}\Ab'\vb_{j}\right)^2}{\lambda_{1} - \lambda_{j}} + o(\pv^2),
\end{align}
\label{eq:tr:eps}
\end{subequations}
where \eqref{eq:tr:eps:1} follows from \eqref{eq:mp:inv}, and $\Vb_{i_1}$ of \eqref{eq:tr:eps:3} is defined as $\Vb_{i_1} \defeq [\vb_1, \dots, \vb_{i_1}] \in \mathbb{R}^{d \times i_1}$. To obtain \eqref{eq:tr:eps:3}, we have used the fact that both the columns of $\Vb_{i_1}$ and $\Vbh_{i_1}$ form an orthonormal basis of the eigenspace of $\Ab$ associated with the eigenvalue $\lambda_1$.

Moreover, similar to the above derivations, for any $s$ we have
\begin{align} %
  &\sum_{i = i_{s-1}+1}^{i_{s}} \vb_i^\T(\pv)\Ab\vb_i(\pv)
  = \lambda_{i_s} (i_s - i_{s - 1})%
    - \pv^2 \sum_{i = i_{s-1}+1}^{i_{s}}\sum_{j: \lambda_j \neq \lambda_{i}} \frac{\left(\vb_i^{\T}\Ab'\vb_{j}\right)^2}{\lambda_{i} - \lambda_{j}} + o(\pv^2).\label{eq:sum_l}
\end{align}
Then from \eqref{eq:tr:vk} and~\eqref{eq:sum_l}, we have
\begin{align*}
    \tr \bigl\{\Vb_k^\T(\pv)\Ab\Vb_k(\pv)\bigr\} =  \sum_{s = 1}^q\sum_{i = i_{s-1}+1}^{i_{s}} \vb_i^\T(\pv)\Ab\vb_i(\pv)%
  &=  \sum_{s = 1}^q \lambda_{i_s} (i_s - i_{s - 1}) - \pv^2\sum_{i=1}^k\sum_{j = k + 1}^d\frac{\left(\vb_i^{\T}\Ab'\vb_j\right)^2}{\lambda_i - \lambda_j} + o(\pv^2)\\
  &= \tr \left\{ \Vb_k^{ \T} \Ab \Vb_k \right\}  - \pv^2\sum_{i=1}^k\sum_{j = k + 1}^d\frac{\left(\vb_i^{\T}\Ab'\vb_j\right)^2}{\lambda_i - \lambda_j} + o(\pv^2),
\end{align*}
which finishes the proof of the lemma.
\section{Proof of Lemma~\ref{lem:2}} \label{sec:appb}

Suppose $\lambda_1, \dots, \lambda_k$ take $q$ distinct values, then we define indices $i_0, \dots, i_q$ such that $0 = i_0 < \dots < i_{q-1} < k < k + 1 \leq i_q$ and 
\begin{align*}
  \lambda_{i_{s-1} + 1} = \lambda_{i_{s}} > \lambda_{i_{s} + 1}, \quad 0 \leq s \leq q.
\end{align*}

We first consider the case $q = 1$, which implies $k < i_1$. From \eqref{eq:V.A.T} we have
\begin{align*}
  &\eqspace\Vb_{k}^{\T}(\pv)\Ab\Vb_{k}(\pv)%
    = \lambda_1 \cdot \Ib_{k} - \pv^2 \Vbh_{k}^\T \Ab' ( \lambda_1\Ib_d - \Ab)^{\dagger}\Ab'\Vbh_{k} + o(\pv^2),
\end{align*}
which implies [cf. \eqref{eq:tr:eps}]
\begin{align*}
  \tr \left\{ \Vb_k^{ \T} \Ab \Vb_k \right\} &= k\lambda_1 - \pv^2 \sum_{i = 1}^{k}\sum_{j \in \cIb_k} \frac{\left(\vbh_i^{\T}\Ab'\vb_{j}\right)^2}{\lambda_{1} - \lambda_{j}} + o(\pv^2)%
\end{align*}

To obtain $\vbh_i~(1 \leq i \leq k)$, note that
\begin{align*}
    \Vb_{k}^{\T}(\pv)\Ab(\pv)\Vb_{k}(\pv)%
  &= \left(\Vbh_{k} + \pv\Vb_{k}' + \pv^2\Vb_{k}''\right)^{\T}\left(\Ab + \pv\Ab' + \pv^2\Ab''\right)%
      \left(\Vbh_{k} + \pv\Vb_{k}' + \pv^2\Vb_{k}''\right) + o(\pv^2)\\
  &= \Vbh_{k}^{\T}\Ab\Vbh_{k} + \pv\left(\Vbh_{k}\Ab\Vb_{k}' + \Vb_{k}'\Ab\Vbh_{k} + \Vbh_{k}^{\T}\Ab\Vbh_{k}\right) + o(\pv)\notag\\
  &= \lambda_1 \Ib_k + \pv\left[\lambda_1\left(\Vbh_{k}^{\T}\Vb_{k}' + \Vb_{k}'^{\T}\Vbh_{k}\right) + \Vbh_{k}^{\T}\Ab'\Vbh_{k}\right] + o(\pv)\\
  &= \lambda_1 \Ib_k + \pv \Vbh_{k}^{\T}\Ab'\Vbh_{k} + o(\pv),
\end{align*}
where to obtain the third equality we have used the fact that $\Ab\Vbh_{k} = \lambda_1\Vbh_{k}$, and to obtain the last equality used that $\left(\Vbh_{k}^{\T}\Vb_{k}' + \Vb_{k}'^{\T}\Vbh_{k}\right)$ is a zero matrix as a consequence of \eqref{eq:ortho:1}.

Since the columns of $\Vbh_{k}$ are $k$ orthonormal vectors in the eigenspace of $\Ab$ associated with the eigenvalue $\lambda_1$, we can write $\Vbh_{k}$ as $\Vbh_{k} = \Vb_{i_1}\Ub$, where ${\Ub = [\ub_1, \dots, \ub_k]\in \mathbb{R}^{i_1 \times k}}$ satisfies 
$\Ub^{\T}\Ub = \Ib_k$. Moreover, from the definition of eigenvectors, $\Vb_{k}(\pv)$ is the $d \times k$ matrix with orthonormal columns that maximizes
\begin{align*}
    \tr\left\{\Vb_{k}^{\T}(\pv)\Ab(\pv)\Vb_{k}(\pv)\right\}%
  &= k\lambda_1 + \pv\tr\left\{\Vbh_{k}^{\T}\Ab'\Vbh_{k}\right\} + o(\pv)%
    = k\lambda_1 + \pv\tr\left\{\Ub^{\T}\Vb_{i_1}^{\T}\Ab'\Vb_{i_1}\Ub\right\} + o(\pv).
\end{align*}
Therefore, $\Ub$ is the optimal solution of 
\begin{alignat*}{2}
  &\max_{\Ub'} &\quad &\tr\left\{\Ub'^{\T}\Vb_{i_1}^{\T}\Ab'\Vb_{i_1}\Ub'\right\}\\
  &~\,\text{s.t.}& &\Ub' \in\, \mathbb{R}^{i_1\times k}, \quad\Ub'^{\T}\Ub' = \Ib_k,
\end{alignat*}
which implies that $\ub_1, \dots, \ub_k$ are the top $k$ eigenvectors 
of the matrix $\Vb_{i_1}^{\T}\Ab'\Vb_{i_1}$. %

As a result, we have
  \begin{align}
    \tr \left\{ \Vb_k^{ \T} \Ab \Vb_k \right\}
    &= k\lambda_1 - \pv^2 \sum_{i = 1}^{k}\sum_{j \in \cIb_k} \frac{\left(\vbh_i^{\T}\Ab'\vb_{j}\right)^2}{\lambda_{1} - \lambda_{j}}+ o(\pv^2)%
      \label{eq:tr:k:k+1:simple}
  \end{align}
  where $\vbh_i = \Vb_{i_1}\ub_i~(1 \leq i \leq k)$, and $\ub_1, \dots, \ub_k$ are the top $k$ eigenvectors of the matrix $\Vb_{i_1}^{\T}\Ab'\Vb_{i_1}$.

  Similarly, for $q > 1$, we have
  \begin{align*}
    &\sum_{i = l}^k \vb_i^{\T}(\pv)\Ab\vb_i(\pv)%
      = (k - l + 1)\lambda_k 
    - \pv^2\sum_{i = l}^k\sum_{j \in \cIb_k} \frac{\left(\vbh_i^{\T}\Ab'\vb_{j}\right)^2}{\lambda_{k} - \lambda_{j}} + o(\pv^2),
  \end{align*}
  where $l$ is the minimal element of $\cI_k$, and
  \begin{align*}
    \vbh_{i} \defeq \Vb_{\cI_k}\ub_{i - l + 1}, \quad l \leq i \leq k,
  \end{align*}
  with $\ub_{1}, \dots, \ub_{k - l + 1}$ being the top $k - l + 1$ eigenvectors of $\Vb_{\cI_k}^{\T}\Ab'\Vb_{\cI_k}$.
  Therefore, we obtain
  \begin{subequations}
    \begin{align}
        \tr \left\{ \Vb_k^{\T}(\pv) \Ab \Vb_k(\pv) \right\}%
      &=\tr \left\{ \Vb_{l - 1}^{\T}(\pv) \Ab \Vb_{l - 1}(\pv) \right\}
        + \sum_{i = l}^k \vb_i^{\T}(\pv)\Ab\vb_i(\pv)\label{eq:tr:k:k+1:1}\\
      &=\tr\left\{ \Vb_{l - 1}^{\T} \Ab \Vb_{l - 1}\right\} - \pv^2\sum_{i=1}^{l - 1}\sum_{j = l}^d\frac{\left(\vb_i^{\T}\Ab'\vb_j\right)^2}{\lambda_i - \lambda_j}%
        + \sum_{i = l}^k \vb_i^{\T}(\pv)\Ab\vb_i(\pv) + o(\pv^2)\label{eq:tr:k:k+1:2}\\
      &=\tr\left\{\Vb_{l - 1}^{\T} \Ab \Vb_{l - 1}\right\} - \pv^2\sum_{i=1}^{l - 1}\sum_{j = l}^d\frac{\left(\vb_i^{\T}\Ab'\vb_j\right)^2}{\lambda_i - \lambda_j}%
        + (k - l + 1)\lambda_r - \pv^2\sum_{i = l}^k\sum_{j \in \cIb_k} \frac{\left(\vbh_i^{\T}\Ab'\vb_{j}\right)^2}{\lambda_{k} - \lambda_{j}} + o(\pv^2)\notag\\
      &=\tr\left\{ \Vb_{k}^{\T} \Ab \Vb_{k}\right\} - \pv^2\sum_{i=1}^{l - 1}\sum_{j = l}^d\frac{\left(\vb_i^{\T}\Ab'\vb_j\right)^2}{\lambda_i - \lambda_j}%
        - \pv^2\sum_{i = l}^k\sum_{j \in \cIb_k} \frac{\left(\vbh_i^{\T}\Ab'\vb_{j}\right)^2}{\lambda_{k} - \lambda_{j}} + o(\pv^2),\label{eq:tr:k:k+1:4}
    \end{align}
    \label{eq:tr:k:k+1}
  \end{subequations}
  where \eqref{eq:tr:k:k+1:2} follows from Lemma \ref{lem:eig:k}.

    \section{Proof of \lemref{lem:perturb:B}}
    \label{sec:app:dtmh}
    
    First, let us define $p_{\min} \defeq \min \{P_{XY}(x, y)\colon (x, y) \in \cX \times \cY, P_{XY}(x, y) > 0\}$. Then for all $\Ph_{XY} \in \nbhd(\eps)$, we have, using Pinsker's inequality \cite{cover2012elements}, 
    \begin{subequations}
    \begin{align}
      \sum_{x\in \cX, y\in \cY} \left[\imate(y, x)\right]^2
      &= \sum_{x\in \cX, y\in \cY\colon P_{XY}(x, y) > 0} \frac{[\Ph_{XY}(x, y) - P_{XY}(x, y)]^2}{\eps P_{XY}(x, y)}        \label{eq:imate:pinsker:1}\\
      &\leq \frac{1}{\eps p_{\min}} \sum_{x\in \cX, y\in \cY}\left[\Ph_{XY}(x, y) - P_{XY}(x, y)\right]^2\label{eq:imate:pinsker:2}\\
      &\leq \frac{1}{\eps p_{\min}} \left(\sum_{x\in \cX, y\in \cY}\left|\Ph_{XY}(x, y) - P_{XY}(x, y)\right|\right)^2\label{eq:imate:pinsker:3}\\
      &%
        \leq \frac{2D(\Ph_{XY}\| P_{XY})}{\eps p_{\min}} \\
      & \leq \frac{2}{p_{\min}\alpha_k}.
    \end{align}
        \label{eq:imate:pinsker}
  \end{subequations}
    Therefore, for all $(x, y)\in \cX \times \cY$ we obtain
    \begin{align}
      \left|\imate(y, x)\right| < \sqrt{\frac{2}{p_{\min}\alpha_k}}.
      \label{eq:imate:bound}
    \end{align}
    Thus, it follows from \eqref{eq:theta} that
    \begin{subequations}
    \begin{align}
      |\dtmdmate(y, x)|
      &\leq \frac{1}{p_{\min}} |\imate(y, x)| +  \frac{1}{p_{\min}} \Biggl|\frac{1}{P_X(x)}  \sum_{y' \in \cY} \sqrt{P_{XY}(x, y')} \imate(y', x) +\frac{1}{P_Y(y)} \sum_{x' \in \cX}\sqrt{P_{XY}(x', y)}\imate(y, x') \Biggr|\\
      &\leq \frac{1}{p_{\min}} |\imate(y, x)| +  \frac{|\cX| + |\cY|}{p^2_{\min}} \cdot \max_{(x, y)\in \cX \times \cY}|\imate(y, x)|\\
      &\leq \frac{1  +  |\cX| + |\cY|}{p^2_{\min}} \cdot  \max_{(x, y)\in \cX \times \cY}|\imate(y, x)|\label{eq:bound:Xi:3}\\
      &\leq \frac{1  +  |\cX| + |\cY|}{p^3_{\min}} \sqrt{\frac{2}{\alpha_k}},
    \end{align}
  \end{subequations}
    where %
    the last inequality follows from \eqref{eq:imate:bound} and the fact that $p_{\min} \leq 1$. Hence, we have $\bbfrob{\dtmdmat} \leq C$, where $C \defeq \frac{1  +  |\cX| + |\cY|}{p^3_{\min}} \sqrt{\frac{2|\cX||\cY|}{\alpha_k}}$ depends only on $P_{XY}$.

    To prove \eqref{eq:perturb:B}, we first define $\pv \defeq \sqrt{\eps}$ for the convenience of presentation. From \eqref{eq:imate:def}, we can represent the differences between the empirical marginal distributions and the true marginal distributions as
    \begin{align} \label{eq:marg_pert_X}
      &\Ph_{X} (x) - P_{X} (x) = \pv \sqrt{P_X(x)} \imate_{X} (x), \\ \label{eq:marg_pert_Y}
      &\Ph_{Y} (y) - P_{Y} (y) = \pv \sqrt{P_Y(y)} \imate_{Y} (y),
    \end{align}
    where
    \begin{align}\label{eq:imate:x}
      \imate_{X} (x) \triangleq \sum_{y \in \cY} \sqrt{P_{Y|X}(y|x)} \imate (y,x), \\\label{eq:imate:y}
      \imate_{Y} (y) \triangleq \sum_{x \in \cX} \sqrt{P_{X|Y}(x|y)} \imate (y,x).
    \end{align}

  In addition, it follows from~\eqref{eq:marg_pert_X} and~\eqref{eq:marg_pert_Y} that
  \begin{align}   \label{eq:pxpy}%
    \sqrt{\Ph_X(x)\Ph_Y(y)} &= \sqrt{P_X(x)P_Y(y)} %
    \Biggl[ 1 + \frac{\pv}{2} \left( \frac{\imate_X(x)}{\sqrt{P_X(x)}} +  \frac{\imate_Y(y)}{\sqrt{P_Y(y)}}  \right)\Biggr] + o(\pv)
  \end{align}
  and
  \begin{align}  \label{eq:ipxpy}
    \frac{1}{\sqrt{\Ph_X(x)\Ph_Y(y)}} &= \frac{1}{\sqrt{P_X(x)P_Y(y)}}%
    \Biggl[ 1 - \frac{\pv}{2} \left( \frac{\imate_X(x)}{\sqrt{P_X(x)}} +  \frac{\imate_Y(y)}{\sqrt{P_Y(y)}}  \right)\Biggr] + o(\pv).
  \end{align}
  Therefore, from \eqref{eq:imate:def} and \eqref{eq:marg_pert_X}--\eqref{eq:ipxpy} we have, for all $(x, y)\in \cX \times \cY$,
  \begin{align} %
    \dtmeh (y,x) - \dtmet (y,x) %
 &= \pv\Biggl(\frac{\sqrt{P_{XY}(x,y)} }{\sqrt{P_X(x)P_Y(y)}} \imate (y,x) -  \frac{P_{XY}(x, y) + P_X(x)P_Y(y)}{2\sqrt{P_X(x)P_Y(y)}}\notag\\
 &\qquad\cdot\Biggl[\frac{1}{P_X(x)}  \sum_{y' \in \cY} \sqrt{P_{XY}(x, y')} \imate(y', x) + \frac{1}{P_Y(y)} \sum_{x' \in \cX}\sqrt{P_{XY}(x', y)}\imate(y, x')\Biggr]\Biggr) + o(\pv)\notag\\
 &= \pv \dtmdmate(y, x) + o(\pv)\notag\\
 &= \sqrt{\eps}\, \dtmdmate(y, x) + o\left(\sqrt{\eps}\right),
   \label{eq:dtm:diff}
  \end{align}
  which is equivalent to \eqref{eq:perturb:B}. 

    \section{Proof of \lemref{lem:S1:nbhd}}
    
    \label{sec:app:S1:nbhd}

    For $\eps > 0$ and $t > 0$, we define the subset $\cS^{(t)}_2(\eps)$ of $\nbhd(\eps)$ as %
  \begin{align}
    \label{eq:cS2:def}
    \cS^{(t)}_2(\eps) \defeq \left\{\Ph_{XY} \colon \Ph_{XY} \in \nbhd(\eps), \Ph_{XY} \leftrightarrow \imate, ~\sum_{i = 1}^k\sum_{j = k + 1}^d \frac{\left[\phib_i^{\T}\left(\dtmt^\T \dtmdmat + \dtmdmat^\T \dtmt\right) \phib_j\right]^2}{\sigma_i^2 - \sigma_j^2} \geq t \right\}, 
  \end{align}
  where $\phib_i$ denotes the $i$-th right singular vector of $\dtmt$, and $\Xib$ is as defined in \eqref{eq:theta}. Then, it is convenient to first establish the following useful lemma.

  \begin{lemma} For all $t \in (0, 2)$, we have
    \label{lem:S2:alpha_k}
    \begin{align}
      \lim_{\eps \rightarrow 0^{+}} \eps^{-1}  D\left(\cS_2^{(t)}(\eps) \,\middle\|\, P_{XY}\right) = \frac{t}{2\alpha_k}.
      \label{eq:S2:t}
    \end{align}
  \end{lemma}

  Using \lemref{lem:S2:alpha_k}, we establish \lemref{lem:S1:nbhd} as follows. First,   
  for all empirical distributions $\Ph_{XY} \in \nbhd(\eps)$, it follows from \lemref{lem:perturb:B} that
  \begin{align} \label{eq:dtmh.T:dtmh}
    \dtmh^{\T} \dtmh = \dtmt^{\T} \dtmt + \sqrt{\eps} \left( \dtmt^{\T} \Xib + \Xib^{\T} \dtmt \right) + o\left(\sqrt{\eps}\right).
  \end{align}
  From the perturbation analysis result of \lemref{lem:eig:k}, we can represent the learning error as
  \begin{align}
    \bbfrob{ \dtmt \Phib_k }^2 - \bbfrob{ \dtmt \Phih_k }^2
    = \eps\sum_{i = 1}^k\sum_{j = k + 1}^d \frac{\left[\phib_i^{\T}\left(\dtmt^\T \dtmdmat + \dtmdmat^\T \dtmt\right) \phib_j\right]^2}{\sigma_i^2 - \sigma_j^2} + o(\eps).
    \label{eq:err:perturb}
  \end{align}
  Therefore, for any $t \in (0, 1)$, there exists an $\eps_0 > 0$ such that for all $\eps \in (0, \eps_0)$, we have
  \begin{align}
    \cS_2^{(1 + t)}(\eps)  \subseteq \cS_1(\eps) \cap \nbhd(\eps) \subseteq \cS_2^{(1 - t)}(\eps).
    \label{eq:set:inclusion}
  \end{align}

    This implies that
    \begin{align}
        D\left(\cS^{(1 - t)}_2(\eps)\,\middle\|\, P_{XY}\right) \leq D\left(\cS_1(\eps) \cap \nbhd(\eps) \,\middle\|\, P_{XY}\right)  \leq D\left(\cS^{(1 + t)}_2(\eps)\,\middle\|\, P_{XY}\right).
      \label{eq:bound:min}
    \end{align}
    From the first inequality of \eqref{eq:bound:min}, we obtain
    \begin{subequations}
      \begin{align}
        \liminf_{\eps \rightarrow 0^+} \eps^{-1} D\left(\cS_1(\eps)\cap \nbhd(\eps)\,\middle\|\, P_{XY}\right)  \geq \liminf_{\eps \rightarrow 0^+} \eps^{-1} D\left(\cS^{(1 - t)}_2(\eps)\,\middle\|\, P_{XY}\right) = \frac{1 - t}{2\alpha_k},
      \end{align}
      where the equality follows from \lemref{lem:S2:alpha_k}. Similarly, from the second inequality of \eqref{eq:bound:min} we have
      \begin{align}
        \limsup_{\eps \rightarrow 0^+} \eps^{-1} D\left(\cS_1(\eps)\cap \nbhd(\eps)\,\middle\|\, P_{XY}\right) \leq \limsup_{\eps \rightarrow 0^+} \eps^{-1} D\left(\cS^{(1 + t)}_2(\eps)\,\middle\|\, P_{XY}\right) = \frac{1 + t}{2\alpha_k}.
      \end{align}
      \label{eq:limsup:liminf}
    \end{subequations}
    As $t$ can be chosen to be arbitrarily close to $0$, we must have %
    \begin{align}
      \lim_{\eps \rightarrow 0^+} \eps^{-1} D\left(\cS_1(\eps)\cap \nbhd(\eps)\,\middle\|\, P_{XY}\right) =
      \frac{1}{2\alpha_k}.
      \label{eq:lim:cS1}
    \end{align}

    It remains only to establish \lemref{lem:S2:alpha_k}.
    
    \begin{proof}[Proof of \lemref{lem:S2:alpha_k}]
    Since the set $\cS_2^{(t)}(\eps)$ is closed, we have
    \begin{align*}
      D\left( \cS_2^{(t)}(\eps)\,\middle\|\, P_{XY}\right) = \inf_{\Ph_{XY} \in \cS_2^{(t)}(\eps)}  D(\Ph_{XY}\| P_{XY}) = \min_{\Ph_{XY} \in \cS_2^{(t)}(\eps)}  D(\Ph_{XY}\| P_{XY}).
    \end{align*}
    Then, for all $\Ph_{XY} \in \cS_2^{(t)}(\eps) \subset \nbhd(\eps)$ with $\Ph_{XY} \leftrightarrow \imate$, it follows from \eqref{eq:imate:bound} that $\imate(y, x)$ is bounded for all $(x, y) \in \cX \times \cY$. Hence, it follows from the second-order Taylor series expansion of the K-L divergence that
    \begin{align}
          D(\Ph_{XY} \| P_{XY}) \notag%
        &= \frac{1}{2} \sum_{x \in \cX, y \in \cY} \frac{\left[ \Ph_{XY}(x,y) - P_{XY}(x,y) \right]^2}{P_{XY}(x,y)} + o(\eps) \\ 
        &= \frac{\eps}{2} \sum_{x \in \cX, y \in \cY} \imate^2 (y,x) + o(\eps) = \frac{\eps}{2} \bbfrob{\imat}^2 + o(\eps).\label{eq:normg}    
    \end{align}

    Moreover, %
    since $\Ph_{XY}$ and $P_{XY}$ are probability distributions, it follows from \eqref{eq:imate:def} that
    \begin{align} \label{eq:cond_g}
      \sum_{x \in \cX, y \in \cY}\sqrt{P_{XY}(x,y)} \imate (y,x) = 0.
    \end{align}
    
    Therefore, %
    the characterization of \eqref{eq:S2:t} %
    leads to the following optimization problem:
    \begin{subequations}
      \begin{alignat}{2}
        &\min_{\imat}& \quad& \bbfrob{\imat}^2         \label{eq:opt:org:1}\\
        &~\,\text{s.t.}& & \sum_{i = 1}^k\sum_{j = k + 1}^d \frac{\left[\phib_i^{\T}\left(\dtmt^\T \dtmdmat + \dtmdmat^\T \dtmt\right) \phib_j\right]^2}{\sigma_i^2 - \sigma_j^2} \geq t,\label{eq:opt:org:2}\\
        &&& \sum_{x \in \cX, y \in \cY}\sqrt{P_{XY}(x,y)} \imate (y, x) = 0\label{eq:opt:org:3}.
      \end{alignat}
      \label{eq:opt:org}
    \end{subequations}
    As we will verify, although not imposed as a constraint, the condition $\Ph_{XY} \in \nbhd(\eps)$ can be satisfied for the optimal $\imat$. Note that since both the objective function and the inequality constraint \eqref{eq:opt:org:2} are quadratic, the optimal solution of \eqref{eq:opt:org} can be %
    obtained via solving
    \begin{alignat}{2}
      &\max_{\imat}& \quad & \sum_{i = 1}^k\sum_{j = k + 1}^d \frac{\left[\phib_i^{\T}\left(\dtmt^\T \dtmdmat + \dtmdmat^\T \dtmt\right) \phib_j\right]^2}{\sigma_i^2 - \sigma_j^2}\notag\\
      &~\,\text{s.t.}& & \bbfrob{\imat}^2 \leq 1, \ \sum_{x \in \cX, y \in \cY}\sqrt{P_{XY}(x,y)} \imate (y, x) = 0,
      \label{eq:opt}
    \end{alignat}
    where we have interchanged the objective function and the quadratic function in the inequality constraint. Furthermore, we can show that \eqref{eq:opt} is equivalent to the optimization problem without the equality constraint, i.e,
    \begin{subequations}
      \begin{align}
        &\max_{\imat} ~\sum_{i = 1}^k\sum_{j = k + 1}^d \frac{\left[\phib_i^{\T}\left(\dtmt^\T \dtmdmat + \dtmdmat^\T \dtmt\right) \phib_j\right]^2}{\sigma_i^2 - \sigma_j^2}\label{eq:opt:equiv:obj}\\
        &~\,\text{s.t.} \quad \bbfrob{\imat}^2 \leq 1.\label{eq:opt:equiv:constraint}
      \end{align}
      \label{eq:opt:equiv}
    \end{subequations}
    To see this, suppose that $\imat^*$ is the optimal solution of \eqref{eq:opt:equiv} with $c \defeq \sum_{x\in \cX, y\in \cY} \sqrt{P_{XY}(x,y)}\imate^* (y,x)$. Then, let $z(x, y) \defeq \imate^* (y, x) - c\sqrt{P_{XY}(x,y)}$, and we have
    \begin{align*}
      1 = \sum_{x \in \cX, y\in \cY}\left[\imate^*(y, x)\right]^2 = \sum_{x \in \cX, y\in \cY}z^2(x, y) + c^2,
    \end{align*}
    which implies $|c| \leq 1$. 

    If $|c| = 1$, we have $\imate^*(y,x) = \pm\sqrt{P_{XY}(x,y)}$, and it follows from \eqref{eq:theta} that $\Xi(y,x) = \mp\sqrt{P_{X}(x)P_{Y}(y)}$. Hence, we have $\Xib = \mp\psib_0\phib_0^{\T}$, where $\psib_0$ is a $|\cY|$-dimensional vector with its $y$-th element being $\sqrt{P_{Y}(y)}$, and $\phib_0 \in \mathbb{R}^{|\cX|}$ with the $x$-th element being $\sqrt{P_{X}(x)}$. Then, the objective function is zero since $\dtmt^\T \dtmdmat = \mp\dtmt^\T\psib_0\psib_0^\T = \Ob$, which contradicts the assumption that $\imat^*$ is optimal. Moreover, if $0<|c| < 1$, then we can construct the matrix $\imat'$ with elements $\imate'(y, x) = z(x, y)/\sqrt{1 - c^2}$. It can be verified that $\bbfrob{\imat'}^2 = 1$ and the objective function in \eqref{eq:opt} for $\imat'$ is $1/(1 - c^2)$ times the corresponding value for $\imat^*$. This again contradicts the optimality of $\imat^*$. Therefore, we have $c = 0$, and the optimization problem \eqref{eq:opt:equiv} has the same solution as that of \eqref{eq:opt}.

    In addition, it can be shown that %
    for $(x',y') \in \cX \times \cY$ such that $P_{XY}(x', y') = 0$, we must have $\imate^* (y', x') = 0$, since otherwise we can set $\imate^* (y', x') = 0$ and rescale $\imat^*$ to $\bbfrob{\imat^*}^2 = 1$, which increases the objective function of~\eqref{eq:opt:equiv} due to~\eqref{eq:theta}. Therefore, the optimal solution $\imat^*$ satisfies the definition \eqref{eq:imate:def}.

    To simplify the objective function \eqref{eq:opt:equiv:obj}, we employ the vectorization operation $\vec(\cdot)$ that stacks all columns of a matrix into a vector. Specifically, for $\Wb = [w_{ij}] \in \mathbb{R}^{p\times q}$, we use $\vec(\Wb)$ to denote the $pq$-dimensional column vector with the $[p(j-1) + i]$-th entry being $w_{ij}$. Then, we can rewrite \eqref{eq:opt:equiv:obj} as
    \begin{subequations}
      \begin{align}
            \sum_{i = 1}^k\sum_{j = k + 1}^d \frac{\left[\phib_i^{\T}\left(\dtmt^\T \dtmdmat + \dtmdmat^\T \dtmt\right) \phib_j\right]^2}{\sigma_i^2 - \sigma_j^2}%
          &=\sum_{i = 1}^k\sum_{j = k + 1}^d \frac{\left[\bigl(\dtmt\phib_i\bigr)^{\T} \dtmdmat\phib_j + \bigl(\dtmt\phib_j\bigr)^{\T} \dtmdmat\phib_i \right]^2}{\sigma_i^2 - \sigma_j^2}  \label{eq:tr:vec:1}\\
          &=\sum_{i = 1}^k\sum_{j = k + 1}^d \frac{\left(\tr\left\{\left(\dtmt\phib_i\phib_j^{\T} + \dtmt\phib_j\phib_i^{\T}\right)^{\T} \dtmdmat\right\}\right)^2}{\sigma_i^2 - \sigma_j^2}\label{eq:tr:vec:2}\\
          &=\sum_{i = 1}^k\sum_{j = k + 1}^d \frac{\left[\vec^{\T}\left(\dtmt\phib_i\phib_j^{\T} + \dtmt\phib_j\phib_i^{\T}\right) \vec(\dtmdmat)\right]^2}{\sigma_i^2 - \sigma_j^2}\label{eq:tr:vec:3}\\
          &=\sum_{i = 1}^k\sum_{j = k + 1}^d \frac{\left[\thetab^{\T}_{ij} \vec(\dtmdmat)\right]^2}{\sigma_i^2 - \sigma_j^2}\label{eq:tr:vec:4}\\
          &=\vec^{\T}(\dtmdmat)\left(\sum_{i = 1}^k\sum_{j = k + 1}^d \frac{\thetab_{ij}\thetab^{\T}_{ij}}{\sigma_i^2 - \sigma_j^2}\right) \vec(\dtmdmat),\label{eq:tr:vec:5}
      \end{align}
      \label{eq:tr:vec}
    \end{subequations}
    where %
    to obtain \eqref{eq:tr:vec:2}--\eqref{eq:tr:vec:3} we have used the properties of trace that
    \begin{align*}
      \ub^{\T}\Mb\vb = \tr\left\{\ub^{\T}\Mb\vb\right\} &= \tr\left\{\vb\ub^{\T}\Mb\right\} %
          = \vec^{\T}\left(\ub\vb^{\T}\right)\vec(\Mb),
    \end{align*}
    and to obtain \eqref{eq:tr:vec:4} we have used the fact that
    \begin{align*}
      \vec\left(\dtmt\phib_i\phib_j^{\T} + \dtmt\phib_j\phib_i^{\T}\right) &= \phib_j \otimes (\dtmt\phib_i) + \phib_i \otimes (\dtmt\phib_j)%
          = \thetab_{ij}.
    \end{align*}
    
    Moreover, it follows from \eqref{eq:theta} and~\eqref{eq:L:def} that $\vec(\dtmdmat) = \Lb\vec\left(\imat\right)$. Thus, \eqref{eq:tr:vec:5} can be reduced to
    \begin{align}
      &\vec^{\T}(\imat)\Lb^\T\left(\sum_{i = 1}^k\sum_{j = k + 1}^d \frac{\thetab_{ij}\thetab^{\T}_{ij}}{\sigma_i^2 - \sigma_j^2}\right) \Lb\vec(\imat)%
          = \vec^{\T}(\imat)\Gb_k \vec(\imat).\label{eq:gk}
    \end{align}
    Since $\left\|\vec\left(\imat\right)\right\| = \frob{\imat}$, the constraint of \eqref{eq:opt:equiv} is equivalent to $\left\|\vec\left(\imat\right)\right\| \leq 1$. Therefore, the maximum of \eqref{eq:gk} is the largest singular value $\alpha_k$ of $\Gb_k$, which is the optimal value of the objective functions in \eqref{eq:opt:equiv} and \eqref{eq:opt}. This implies that the optimal solution of the original optimization problem \eqref{eq:opt:org} is $\sqrt{\frac{t}{\alpha_k}} \imat^* $, with the corresponding optimal value being %
 $t\alpha_k^{-1}$. Let $\Ph^*_{XY} \leftrightarrow \sqrt{\frac{t}{\alpha_k}}\imate^*$ denote the corresponding optimal empirical distribution, then we have, for $\eps$ sufficiently small,
    \begin{align*}
      D(\Ph^*_{XY} \| P_{XY}) = \frac{\eps t}{2\alpha_k} + o(\eps) < \frac{\eps }{\alpha_k},
    \end{align*}
    where we have used the fact that $t \in (0, 2)$.

    Hence, we obtain $\Ph^*_{XY} \in \nbhd(\eps)$ and thus
    \begin{align}
      \min_{\Ph_{XY} \in \cS_2^{(t)}(\eps)}  D(\Ph_{XY}\| P_{XY}) = D(\Ph^*_{XY} \| P_{XY}) + o(\eps) = \frac{\eps t}{2\alpha_k} + o(\eps),
      \label{eq:res:final}
    \end{align}
    which implies \eqref{eq:S2:t}.

  \end{proof}

\section{Proof of  \thmref{thm:samples}}
\label{sec:app:col:non-asym}
\reply{
  First, it follows from \eqref{eq:lim:S1} that there exists an $\eps_0 > 0$ such that for all $\eps \in (0, \eps_0)$ we have
  \begin{align}
    D\left(S_1(\eps) \,\middle\|\, P_{XY}\right) = \frac{\eps}{2\alpha_k} + o(\eps) > \frac{\eps}{3\alpha_k} = \kappa\eps,
    \label{eq:lb:D}
  \end{align}
  where we have defined $\kappa = (3\alpha_k)^{-1}$.

  Then, using Sanov's theorem, we have for all $\eps \in (0, \eps_0)$, 
  \begin{align}
    \mathbb{P}_n \left\{ \bbfrob{ \dtmt \Phib_k }^2 - \bbfrob{ \dtmt \Phih_k }^2 > \eps  \right\}
    &\leq (n + 1)^{|\cX||\cY|} \exp\left(-n\cdot  D\left(S_1(\eps) \,\middle\|\, P_{XY}\right) \right)\\
    &< (n + 1)^{|\cX||\cY|} \exp\left(-\kappa n\eps\right)\label{eq:lb:D:res}\\
    &\leq (2n)^{|\cX||\cY|} \exp\left(-\kappa n\eps\right)\label{eq:n:geq1}\\
    &= \left(\frac{8|\cX||\cY|}{\kappa \eps}\right)^{|\cX||\cY|}
     \left(\frac{\kappa  n\eps}{4|\cX||\cY|}\right)^{|\cX||\cY|} \exp\left(-\kappa  n\eps\right)\\
    &<  \left(\frac{8|\cX||\cY|}{\kappa \eps}\right)^{|\cX||\cY|}
      \exp\left(\frac{\kappa  n\eps}{4|\cX||\cY|} \cdot |\cX||\cY|\right) \exp\left(-\kappa  n\eps\right)
    \label{eq:ieq:exp}\\
    &=  \left(\frac{8|\cX||\cY|}{\kappa \eps}\right)^{|\cX||\cY|} \exp\left(-\frac{3\kappa  n\eps}{4}\right),
      \label{eq:sanov:ieq}
  \end{align}
  where to obtain \eqref{eq:lb:D:res} we have used \eqref{eq:lb:D}, to obtain \eqref{eq:n:geq1} we have used the fact that $ n \geq 1$, and to obtain \eqref{eq:ieq:exp} we have used the fact that $x \leq e^x - 1 < e^x$.
  
  Therefore, it suffices to choose $n$ such that
  \begin{align*}
    \left(\frac{8|\cX||\cY|}{\kappa \eps}\right)^{|\cX||\cY|} \exp\left(-\frac{3\kappa  n\eps}{4}\right)
    < \delta,
  \end{align*}
  which is equivalent to
  \begin{align*}
    n > \frac{4|\cX||\cY|}{3\kappa \eps} \log \frac{8|\cX||\cY|}{\kappa \eps} + \frac{4}{3\kappa \eps}\log\frac1\delta
    = \frac{4\alpha_k|\cX||\cY|}{\eps} \log \frac{24\alpha_k|\cX||\cY|}{\eps} + \frac{4\alpha_k}{\eps}\log\frac1\delta = N^{(4\alpha_k)}(\eps, \delta),
  \end{align*}
  where we have used the fact that $\kappa = (3\alpha_k)^{-1}$.
}

\section{Proof of \propref{eg:1}} \label{sec:appc}
  First, note that \eqref{eq:gb:k} can be reduced to
  \begin{align}
    \Gb_k = \sum_{i = 1}^{d - 1} \frac{1}{\sigma_i^2}\Lb^{\T}\thetab_{id}\thetab^{\T}_{id}\Lb = \sum_{i = 1}^{d - 1} \frac{1}{\sigma_i^2}\left(\Lb^{\T}\thetab_{id}\right)\left(\Lb^{\T}\thetab_{id}\right)^\T,
    \label{eq:Gbk:d-1}
  \end{align}
  where
  \begin{align*}
   \thetab_{id} = \phib_d \otimes \bigl(\dtmt\phib_i\bigr) = \sigma_i(\phib_d \otimes \psib_i),\quad i = 1, \dots, d - 1,
  \end{align*}
  where $\psib_i \in \mathbb{R}^{|\cY|}$ %
  is the $i$-th left singular vector of $\dtmt$. Then, from the facts that $\sigma_{d - 1} > 0 = \sigma_d$ and $d = |\cX| \leq |\cY|$, we know that $\phib_d$ is the only right singular vector associated with the singular value $0$, and thus we have
  \begin{align}
    \label{eq:phib:d}    
    \phib_d = \Bigl[\sqrt{P_X(1)}, \dots, \sqrt{P_X(d)}\Bigr]^{\T}.
  \end{align}
  Then it follows from \eqref{eq:L:def} that the $[(x - 1)|\cY| + y]$-th entry of $\left(\Lb^{\T}\thetab_{id}\right)$ is
  \begin{subequations}
    \begin{align}
      &\sum_{x' \in \cX, y' \in \cY} \sqrt{\frac{P_{XY}(x, y)}{P_X(x')P_Y(y')}}  \Biggl(\delta_{xx'}\delta_{yy'} - \frac{1}{2}\biggl[\frac{\delta_{xx'}}{P_X(x')} %
        + \frac{\delta_{yy'}}{P_Y(y')}\biggr]\cdot\left[P_{XY}(x',y')+P_{X}(x')P_Y(y')\right]\Biggr)%
        \cdot\bigl[\sigma_i \phi_d(x')\psi_i(y')\bigr]\notag\\
      &= \sigma_i\phi_d(x)\psi_i(y)\sqrt{\frac{P_{XY}(x, y)}{P_X(x)P_Y(y)}}-\frac{\sigma_i}{2}\sqrt{P_{XY}(x, y)}\notag\\
      &\qquad
        \cdot  \sum_{x' \in \cX, y' \in \cY} \Biggl(\biggl[\frac{\delta_{xx'}}{P_X(x')} + \frac{\delta_{yy'}}{P_Y(y')}\biggr] %
        \cdot \left[\dtmet(y', x') + 2\phi_d(x')\sqrt{P_Y(y')}\right]\psi_i(y')\phi_d(x')\Biggr)\label{eq:comp:alpha:1}\\
      &= \sigma_i\phi_d(x)\psi_i(y)\sqrt{\frac{P_{XY}(x, y)}{P_X(x)P_Y(y)}}-\frac{\sigma_i}{2}\sqrt{P_{XY}(x, y)}\notag\\
      &\qquad\cdot \Biggl( \frac{\phi_d(x)}{{P_X(x)}} \sum_{y' \in \cY}\left[\dtmet(y', x) + 2\phi_d(x)\sqrt{P_Y(y')}\right]\psi_i(y')%
        + \frac{\psi_i(y)}{{P_Y(y)}} \sum_{x' \in \cX}\left[\dtmet(y, x') + 2\phi_d(x')\sqrt{P_Y(y)}\right]\phi_d(x')\Biggr)\label{eq:comp:alpha:2}\\
      &= \sigma_i\phi_d(x)\psi_i(y)\sqrt{\frac{P_{XY}(x, y)}{P_X(x)P_Y(y)}}-\frac{\sigma_i}{2}\sqrt{P_{XY}(x, y)}%
        \cdot \left[ \frac{\sigma_i\phi_i(x)\phi_d(x)}{{P_X(x)}} + \frac{2 \psi_i(y)}{\sqrt{P_Y(y)}}\right]\label{eq:comp:alpha:3}\\
      &= -\frac{\sigma_i^2}{2}\sqrt{P_{Y|X}(y|x)}\phi_i(x),\label{eq:comp:alpha:4}
    \end{align}
    \label{eq:comp:alpha}
  \end{subequations}
  where $\phi_i(x)$ and $\psi_j(y)$ denote the $x$-th entry of $\phib_i$ and the $y$-th entry of $\psib_j$, respectively, and where to obtain \eqref{eq:comp:alpha:1} we have exploited the fact that
  \begin{align*}
    &\frac{P_{XY}(x', y') + P_{X}(x')P_Y(y')}{\sqrt{P_X(x')P_Y(y')}}%
      =\dtmet(y', x') + 2\phi_d(x')\sqrt{P_Y(y')}.
  \end{align*}  
  In addition, to obtain \eqref{eq:comp:alpha:3}, we have used the facts that 
  \begin{subequations}
    \begin{gather}
      \sum_{x' \in \cX} \dtmet(y, x')\phi_d(x') = \sigma_d\psi_d(y) = 0,\\
      \sum_{x'\in \cX}[\phi_d(x')]^2 = \|\phib_d\|^2 = 1,
    \end{gather}
  \end{subequations}
  and for $1 \leq i \leq d-1$, 
  \begin{subequations}
    \begin{gather}
      \sum_{y'\in \cY} \dtmet(y', x)\psi_i(y') = \sigma_i\phi_i(x),\\
      \sum_{y'\in \cY}\sqrt{P_Y(y')}\psi_i(y') = 0,\label{eq:ortho:psi}
    \end{gather}
  \end{subequations}
  where \eqref{eq:ortho:psi} follows from the fact that the vector $\Bigl[\sqrt{P_Y(1)}, \dots, \sqrt{P_Y(|\cY|)}\Bigr]^{\T}\in \mathbb{R}^{|\cY|}$ is a left singular vector of the matrix $\dtmt$ associated with the singular value $0$.

  Hence, from \eqref{eq:comp:alpha} we have
  \begin{align*}
    \Lb^\T\thetab_{id} =
    -\frac{\sigma_i^2}{2} \Mb\phib_i,
  \end{align*}
  where $\Mb \in \mathbb{R}^{(|\cX|\cdot|\cY|) \times |\cX|}$ is a block diagonal matrix defined as  
  \begin{subequations}
  \begin{align}
    \Mb \defeq
      \begin{bmatrix}
        \nub_1 &\zerob_{|\cY|} &\cdots &\zerob_{|\cY|} \\
        \zerob_{|\cY|}& \nub_2& \cdots&\zerob_{|\cY|} \\      
        \vdots&  \vdots&  \ddots& \vdots\\
        \zerob_{|\cY|}& \zerob_{|\cY|} &\cdots  & \nub_{|\cX|}
    \end{bmatrix},
    \label{eq:Mb}
  \end{align}
  where $\zerob_{|\cY|}$ is the zero vector in $\mathbb{R}^{|\cY|}$, and for each $x \in \cX$, $\nub_x$ is a $|\cY|$-dimensional vector defined as
  \begin{align}\label{eq:nub}    
    \nub_x = \biggl[\sqrt{P_{Y|X}(1|x)}, \dots, \sqrt{P_{Y|X}\bigl(|\cY|\big|x\big)}\biggr]^\T.
  \end{align}
  \label{eq:Mb:nub}    
\end{subequations}
  As a result, it follows from \eqref{eq:Gbk:d-1} that
  \begin{align}\label{eq:gbk:eig}
    \Gb_k = \frac14\sum_{i = 1}^{d - 1}\sigma_i^2 (\Mb\phib_i)(\Mb\phib_i)^{\T},
  \end{align}
  from which we can obtain the eigen-decomposition of $\Gb_k$. Indeed, since $\Mb^{\T}\Mb = \Ib_d$, we have $\langle\Mb\phib_i, \Mb\phib_j\rangle = \langle\phib_i, \phib_j\rangle = \delta_{ij}$. Therefore, from \eqref{eq:gbk:eig}, the non-zero eigenvalues of $\Gb_k$ are $\sigma_i^2/4\,(i = 1, \dots, d - 1)$, with the corresponding eigenvectors $\Mb\phib_i\,(i = 1, \dots, d - 1)$. As a result, the largest eigenvalue (i.e., the largest singular value) of $\Gb_k$ is 
  \begin{align*}
    \alpha_k = \spectral{\Gb_k} = \frac{\sigma_1^2}{4},
  \end{align*}
  where $\spectral{\cdot}$ denotes the spectral norm of its argument. 
  
  \section{Proof of \thmref{thm:exponent:k:k+1}} \label{sec:appd}
  The proof is similar to that of \thmref{thm:exponent}, except that we need to replace the perturbation analysis result of \lemref{lem:eig:k} with the corresponding result of \lemref{lem:2}. In particular, we extend the definition $\nbhd(\eps)$ to the case $\sigma_k = \sigma_{k + 1}$ via letting\footnote{It can be verified that, we have $\beta_k = \alpha_k$ if $\sigma_k > \sigma_{k + 1}$. Therefore, the definition \eqref{eq:nbhd:=} is a generalization of \eqref{eq:nbhd}.} %
  \begin{align}
    \nbhd(\eps) \defeq \left\{\Ph_{XY}\colon D(\Ph_{XY}\|P_{XY}) \leq \frac{\eps}{\beta_k}\right\}.
    \label{eq:nbhd:=}
  \end{align}

  Then, we define
  $\cS^{(t)}_3(\eps)$ as the set of $\Ph_{XY}$ such that %
  the corresponding $\imate$ from~\eqref{eq:imate:def} satisfies
\begin{align} 
  &\sum_{i = 1}^{l - 1}\sum_{j = l}^d \frac{\left[\phib_i^{\T}\left(\dtmt^\T \dtmdmat + \dtmdmat^\T \dtmt\right) \phib_j\right]^2}{\sigma_i^2 - \sigma_j^2}%
    + \sum_{i = l}^{k}\sum_{j \in \cIb_k} \frac{\left[\varphib_i^{\T}\left(\dtmt^\T \dtmdmat + \dtmdmat^\T \dtmt\right) \phib_j\right]^2}{\sigma_i^2 - \sigma_j^2}\label{eq:condition_k=k+1}
  \geq t,
\end{align} 
where $\varphib_i$ are as defined in \eqref{eq:varphi:def}. Then, analogous to \lemref{lem:S2:alpha_k}, it is convenient to first establish the following result.
  \begin{lemma} When $\sigma_k = \sigma_{k + 1}$, for all $t \in (0, 2)$, we have
    \label{lem:S3:beta_k}
    \begin{align}
      \lim_{\eps \rightarrow 0^{+}} \eps^{-1}  D\left(\cS_3^{(t)}(\eps) \,\middle\|\, P_{XY}\right) = \frac{t}{2\beta_k}.
      \label{eq:S2:t:=}
    \end{align}
  \end{lemma}

  \begin{proof}
    The proof is similar to that of \lemref{lem:S2:alpha_k}. Using the second-order Taylor series expansion of the K-L divergence \eqref{eq:normg}, the limit \eqref{eq:S2:t:=} can be characterized by the following optimization problem:
    \begin{subequations}
      \begin{alignat}{2} \notag
        &\min_{\imat} &\quad&\bbfrob{\imat}^2 \\ 
        &~\,\text{s.t.}&& \sum_{i = 1}^{l - 1}\sum_{j = l}^d \frac{\left[\phib_i^{\T}\left(\dtmt^\T \dtmdmat + \dtmdmat^\T \dtmt\right) \phib_j\right]^2}{\sigma_i^2 - \sigma_j^2}%
        + \sum_{i = l}^{k}\sum_{j \in \cIb_k} \frac{\left[\varphib_i^{\T}\left(\dtmt^\T \dtmdmat + \dtmdmat^\T \dtmt\right) \phib_j\right]^2}{\sigma_i^2 - \sigma_j^2} \geq t,\\
        &&&\ \sum_{x \in \cX, y \in \cY}\sqrt{P_{XY}(x,y)} \imate (y,x) = 0.
      \end{alignat}
      \label{eq:opt:k:k+1}
    \end{subequations}
    Following the same argument as that for \lemref{lem:S2:alpha_k}, the optimal solution of \eqref{eq:opt:k:k+1} can be obtained by solving
    \begin{subequations}
    \begin{alignat}{2} 
      &\max_{\imat} &\quad&\sum_{i = 1}^{l - 1}\sum_{j = l}^d \frac{\left[\phib_i^{\T}\left(\dtmt^\T \dtmdmat + \dtmdmat^\T \dtmt\right) \phib_j\right]^2}{\sigma_i^2 - \sigma_j^2}%
      + \sum_{i = l}^{k}\sum_{j \in \cIb_k} \frac{\left[\varphib_i^{\T}\left(\dtmt^\T \dtmdmat + \dtmdmat^\T \dtmt\right) \phib_j\right]^2}{\sigma_i^2 - \sigma_j^2} \label{eq:opt:k:k+1:eq:obj}\\ 
      &~\,\text{s.t.} && \bbfrob{\imat}^2 \leq 1,
    \end{alignat}
    \label{eq:opt:k:k+1:eq}
    \end{subequations}
    where we have interchanged the objective function and the quadratic function in the inequality constraint, and removed the equality constraint.

    Then, similar to \eqref{eq:tr:vec}, we can rewrite the objective function \eqref{eq:opt:k:k+1:eq:obj} as
    \begin{align*} \notag
      &\sum_{i = 1}^{l - 1}\sum_{j = l}^d \frac{\left[\phib_i^{\T}\left(\dtmt^\T \dtmdmat + \dtmdmat^\T \dtmt\right) \phib_j\right]^2}{\sigma_i^2 - \sigma_j^2}%
          + \sum_{i = l}^{k}\sum_{j \in \cIb_k} \frac{\left[\varphib_i^{\T}\left(\dtmt^\T \dtmdmat + \dtmdmat^\T \dtmt\right) \phib_j\right]^2}{\sigma_i^2 - \sigma_j^2} \\
      &= \vec^{\T}(\dtmdmat)\left(\sum_{i = 1}^{l - 1}\sum_{j = l}^d \frac{\thetab_{ij}\thetab^{\T}_{ij}}{\sigma_i^2 - \sigma_j^2} + \sum_{i = l}^k\sum_{j \in \cIb_k} \frac{\varthetab_{ij}\varthetab^{\T}_{ij}}{\sigma_i^2 - \sigma_j^2}\right)\vec(\dtmdmat)\\
      &= \vec^{\T}(\imat)\Lb^{\T}\left(\sum_{i = 1}^{l - 1}\sum_{j = l}^d \frac{\thetab_{ij}\thetab^{\T}_{ij}}{\sigma_i^2 - \sigma_j^2} + \sum_{i = l}^k\sum_{j \in \cIb_k} \frac{\varthetab_{ij}\varthetab^{\T}_{ij}}{\sigma_i^2 - \sigma_j^2}\right)%
          \Lb\vec(\imat)\\
      &= \vec^{\T}(\imat)\Jb_k(\imat)\vec(\imat),
    \end{align*}
    where the second equality follows from the fact that $\vec(\dtmdmat) = \Lb\vec(\imat)$. As a result, the optimization problem \eqref{eq:opt:k:k+1:eq} can be rewritten as \eqref{eq:opt:k:k+1:Jbk}, and thus the optimal value is $\beta_k$. Note that if $\sigma_k > \sigma_{k + 1}$, we may let $\varphib_i = \phib_i$ for $i = l, \dots, k$ since it does not change the value of \eqref{eq:opt:k:k+1:eq:obj}. Then, it can be verified that the optimal value of \eqref{eq:opt:k:k+1:eq} is $\alpha_k$, i.e., we have $\beta_k = \alpha_k$ if $\sigma_k > \sigma_{k + 1}$. 

    Finally, using the same argument as that for \lemref{lem:S2:alpha_k}, we conclude that the optimal value of \eqref{eq:opt:k:k+1} is $t/\beta_k$ and we have
    \begin{align*}
      D\left(\cS_3^{(t)}(\eps) \,\middle\|\, P_{XY}\right) = \frac{\eps t}{2\beta_k} + o(\eps),
    \end{align*}
    which implies \eqref{eq:S2:t:=}.
  \end{proof}

  In addition, it follows from \lemref{lem:2} and \lemref{lem:perturb:B} that the corresponding learning error for the empirical distribution $\Ph_{XY} \in \nbhd(\eps)$ is
  \begin{align}
    \bbfrob{ \dtmt \Phib_k }^2 - \bbfrob{ \dtmt \Phih_k }^2
    = \eps\sum_{i = 1}^{l - 1}\sum_{j = l}^d \frac{\left[\phib_i^{\T}\left(\dtmt^\T \dtmdmat + \dtmdmat^\T \dtmt\right) \phib_j\right]^2}{\sigma_i^2 - \sigma_j^2}%
    + \eps\sum_{i = l}^{k}\sum_{j \in \cIb_k} \frac{\left[\varphib_i^{\T}\left(\dtmt^\T \dtmdmat + \dtmdmat^\T \dtmt\right) \phib_j\right]^2}{\sigma_i^2 - \sigma_j^2} + o(\eps).
    \label{eq:err:perturb:=}
  \end{align}

  Therefore, for any $t \in (0, 1)$, there exists an $\eps_0 > 0$ such that for all $\eps \in (0, \eps_0)$, we have
  \begin{align*}
    \cS_3^{(1 + t)}(\eps)  \subseteq \cS_1(\eps) \cap \nbhd(\eps) \subseteq \cS_3^{(1 - t)}(\eps).
  \end{align*}
  
  Then, using arguments similar to \eqref{eq:bound:min}--\eqref{eq:lim:cS1}, we conclude
  \begin{align*}
    \lim_{\eps \rightarrow 0^{+}} \eps^{-1}  D\left(\cS_1^{(t)}(\eps) \,\middle\|\, P_{XY}\right) = \frac{2}{\beta_k}.    
  \end{align*}
  Finally, following the same proof as that for \thmref{thm:exponent}, we obtain \eqref{eq:exponent:k:k+1}.

\section{Proof of Corollary \ref{eg:2}} \label{sec:appe}
We first introduce two useful lemmas.
\begin{lemma}
  \label{lem:psi=phi}
  Suppose $P_{XY}$ is as defined in Corollary \ref{eg:2} with the corresponding matrix $\dtmt$ as given by \eqref{eq:dtmt}. Then, the matrix $\dtmt$ have singular values
  \begin{align}\label{eq:sigma:eg2:new}
    \sigma_1 = \dots = \sigma_{d - 1} = \frac{|p_1 - p_2|\sqrt{d}}{\sqrt{p_1 + (d - 1)p_2}}, \quad\sigma_d = 0.
  \end{align}  
  In addition, for all $\phib  = [\phi(1), \dots, \phi(d)]^{\T} \in \mathbb{R}^d$ with
  \begin{align}\label{eq:phi:constraint}
  \langle\phib, \bone_d\rangle = 0\quad\text{and}\quad\|\phib\| = 1,
  \end{align}
  the corresponding $\psib \defeq \sigma_1^{-1}\dtmt^{\T}\phib = \left[\psi(1), \dots, \psi(|\cY|)\right]^{\T} \in \mathbb{R}^{|\cY|}$ satisfies
\begin{align}\label{eq:phi=psi}
  \psi(y) =
  \begin{cases}
    \sgn(p_1 - p_2)\phi(y),&\text{if}~ y \leq d,\\
    0,&\text{otherwise},\\
  \end{cases}
\end{align}
where $\bone_d$ denotes the vector in $\mathbb{R}^d$ with all entries being $1$.
\end{lemma}

\begin{proof}
  From the definition of $P_{XY}$, we have
  \begin{align*}
    P_X(x) = \frac{1}{d}, \quad \forall\, x \in \cX,
  \end{align*}
  and
  \begin{align*}
    P_Y(y) =
    \begin{cases}
      p_1 + (d - 1)p_2,& \text{if}~ y \leq d,\\
      dp_2,& \text{otherwise.}
    \end{cases}
  \end{align*}
  
  Therefore, from \eqref{eq:dtmt} we have
  \begin{align}\label{eq:dtmt:eg2}
    \dtmt = \frac{(p_1 - p_2)\sqrt{d}}{\sqrt{p_1 + (d - 1)p_2}}
    \begin{bmatrix}
      \Ib_d - d^{-1}\bone_d\bone_d^{\T}\\
      \Ob_{|\cY| - d,d}
    \end{bmatrix},
  \end{align}
  where $\Ob_{|\cY| - d,d}$ is the zero matrix in $\mathbb{R}^{(|\cY| - d) \times d}$. As a result, we have
  \begin{align*}
    \dtmt^{\T}\dtmt = \frac{(p_1 - p_2)^2d}{p_1 + (d - 1)p_2} \left(\Ib_d - \frac{1}{d}\bone_d\bone_d^{\T}\right).
  \end{align*}
  Since the matrix
  \begin{align*}
    \Ib_d - \frac{1}{d}\bone_d\bone_d^{\T}
  \end{align*}
  has eigenvalues $\lambda_1 = \dots = \lambda_{d - 1} = 1$ and $\lambda_d = 0$, we obtain the singular values $\sigma_1, \dots, \sigma_d$ of $\dtmt$ as given by \eqref{eq:sigma:eg2:new}, and thus we can rewrite \eqref{eq:dtmt:eg2} as
  \begin{align}\label{eq:dtmt:eg2:new}
    \dtmt = \sgn(p_1 - p_2)\cdot\sigma_1
    \begin{bmatrix}
      \Ib_d - d^{-1}\bone_d\bone_d^{\T}\\
      \Ob_{|\cY| - d,d}
    \end{bmatrix}.
  \end{align}    

  Hence, for all $\phib \in \mathbb{R}^d$ with \eqref{eq:phi:constraint}, we have
  \begin{align*}    
    \psib = \sigma_1^{-1}\dtmt\phib
    &= \sgn(p_1 - p_2)
    \begin{bmatrix}
      \Ib_d - d^{-1}\bone_d\bone_d^{\T}\\
      \Ob_{|\cY| - d,d}
    \end{bmatrix}
    \phib%
      =
    \sgn(p_1 - p_2)
    \begin{bmatrix}
      \phib\\
      \zerob_{|\cY| - d}
    \end{bmatrix},
  \end{align*}
  where $\zerob_{|\cY| - d}$ is the zero vector in $\mathbb{R}^{|\cY| - d}$. 
\end{proof}

\begin{lemma}\label{lem:phi_1:max}
For all $\phib = [\phi(1), \dots, \phi(d)]^{\T} \in \mathbb{R}^d$ satisfying \eqref{eq:phi:constraint}, we have
  \begin{align*}
    \phi^2(1) \leq \frac{d - 1}{d},
  \end{align*}
  where the inequality holds with equality if and only if $\phib = \pm \phib'$, where
  \begin{align}\label{eq:phi'}
    \phib' \defeq \frac{1}{\sqrt{(d-1)d}}\left[d-1, -1, \dots, -1
    \right]^{\T} \in \mathbb{R}^d.
  \end{align}
\end{lemma}
\begin{proof}
From $\langle\phib, \bone_d\rangle = 0$ we have
\begin{align*}
  \sum_{i = 1}^d \phi(i) = 0.
\end{align*}
Therefore, we obtain
\begin{align*}
  \phi^2(1) = \left[\sum_{i= 2}^d \phi(i)\right]^2
  &\leq (d - 1)\sum_{i = 2}^d\phi^2(i)%
    = (d - 1)\left[\|\phib\|^2 - \phi^2(1)\right],
\end{align*}
where the inequality follows from the fact that the arithmetic mean is no greater than the root mean square. As a result, we have
\begin{align*}
  \phi^2(1) \leq \frac{d-1}{d} \|\phib\|^2 = \frac{d-1}{d},
\end{align*}
where the inequality holds with equality if and only if
\begin{align}
  \phi(2) = \phi(3) = \dots = \phi(d).
  \label{eq:phi:2:d}
\end{align}
Hence, it follows from \eqref{eq:phi:constraint} and \eqref{eq:phi:2:d} that $\phib = \pm \phib'$.
\end{proof}

Now, Corollary \ref{eg:2} can be proved as follows.

\begin{proof}[Proof of Corollary \ref{eg:2}]
  From \lemref{lem:psi=phi}, we have $\sigma_1 = \dots = \sigma_{d - 1} > \sigma_d = 0$. Therefore, for all $1 \leq k \leq d - 1$ we have $\cI_k = [d - 1]$, which further implies that $l = \min \cI_k = 1$ and $\cIb_k = \{d\}$.
  Hence, from \eqref{eq:Jbk} we have
  \begin{align}
    \Jb_k(\imat) = \Lb^{\T}\left(\sum_{i = 1}^k \frac{\varthetab_{id}\varthetab^{\T}_{id}}{\sigma_1^2}\right)\Lb = \frac{1}{\sigma_1^2}\sum_{i = 1}^k \left(\Lb^{\T} \varthetab_{id}\right)\left(\Lb^{\T} \varthetab_{id}\right)^{\T}.
  \end{align}
  In addition, following the same derivation as that in \eqref{eq:comp:alpha}, we have
  \begin{align}
    \Lb^{\T} \varthetab_{id} = -\frac{\sigma_1^2}{2}\Mb \varphib_i,\quad  1 \leq i \leq k,
  \end{align}
  and thus
  \begin{align}\label{eq:gbk:eig:semi}
    \Jb_k(\imat) = \frac{\sigma_1^2}{4} \sum_{i = 1}^{k}(\Mb \varphib_i)(\Mb \varphib_i)^{\T}.
  \end{align}
  Note that since $\left\langle\Mb \varphib_i, \Mb \varphib_j\right\rangle = \delta_{ij}$, \eqref{eq:gbk:eig:semi} demonstrates the eigen-decomposition of $\Gb_k$. Therefore, from \thmref{thm:exponent:k:k+1}, we have
  \begin{align*}
    \beta_k \leq \bbspectral{\Jb_k(\imat)} = \frac{\sigma_1^4}{4}.
  \end{align*}
  To prove the inequality holds with equality, it suffices to show that there exists a $\imat$ with $\|\imat\|_{\F} \leq 1$ such that
  \begin{align}\label{eq:alpha:eq}
    \vec^{\T}\left(\imat\right)\Jb_k(\imat)\vec\left(\imat\right) = \bbspectral{\Jb_k(\imat)}.
  \end{align}
  Indeed, as we now illustrate, if $\imat$ is chosen as
  \begin{align}\label{eq:imat:opt}
    \imate(y, x) = \sqrt{P_{Y|X}(y|x)}\phi'(x)
  \end{align}
  with $\phi'$ as defined in \eqref{eq:phi'}, then we have $\varphib_1 = \pm \phib'$ and $\vec\left(\imat\right) = \Mb\phib'= \pm\Mb\varphib_1$, and thus \eqref{eq:alpha:eq} holds.

  To see this, first note that from \eqref{eq:imat:opt} we have $\|\imat\|_{\F} = 1$, 
  \begin{align*}
    \sum_{y' \in \cY}\sqrt{P_{X Y}(x, y')}\imate(y', x) = \sqrt{P_X(x)}\phi'(x)
  \end{align*}
  and
  \begin{align*}
      \sum_{x' \in \cX}\sqrt{P_{X Y}(x', y)}\imate(y, x')%
    &= \sqrt{P_Y(y)}\sum_{x' \in \cX}\frac{P_{XY}(x', y)}{\sqrt{P_X(x')P_Y(y)}}\phi'(x')%
      = \sigma_1\sqrt{P_Y(y)}\psi'(y),
  \end{align*}
  where $\psib' = \left[\psi'(1), \dots, \psi'(|\cY|)\right]^{\T} \defeq \sigma_1^{-1}\dtmt\phib'$.

  Therefore, from \eqref{eq:theta} we obtain
  \begin{align}
      \dtmdmate(y, x)%
    &= \frac{\sqrt{P_{XY}(x,y)}}{\sqrt{P_X(x)P_Y(y)}} \imate(y, x)%
      -  \frac{P_{XY}(x, y) + P_X(x)P_Y(y)}{2\sqrt{P_X(x)P_Y(y)}}\notag\\
    &\qquad \cdot\Biggl[\frac{1}{P_X(x)}  \sum_{y' \in \cY} \sqrt{P_{XY}(x, y')} \imate(y', x)%
      +\frac{1}{P_Y(y)} \sum_{x' \in \cX}\sqrt{P_{XY}(x', y)}\imate(y, x') \Biggr] \notag\\
    &= \frac{P_{XY}(x,y)}{\sqrt{P_X(x)P_Y(y)}} \cdot \frac{\phi'(x)}{\sqrt{P_X(x)}}%
      -  \frac{P_{XY}(x, y) + P_X(x)P_Y(y)}{2\sqrt{P_X(x)P_Y(y)}}
      \Biggl[\frac{\phi'(x)}{\sqrt{P_X(x)}} +\frac{\sigma_1\psi'(y)}{\sqrt{P_Y(y)}}  \Biggr] \notag\\
    &= \frac{1}{2}\dtmet(y, x)\cdot \left[\frac{\phi'(x)}{\sqrt{P_X(x)}} - \frac{\sigma_1\psi'(y)}{\sqrt{P_Y(y)}}\right] - \sigma_1\sqrt{P_X(x)}\psi'(y).\label{eq:theta:eg}
  \end{align}

  In addition, since $\cI_k = [d - 1]$, from \eqref{eq:varphi:def}, $\varphib_1$ is the solution of the optimization problem
  \begin{alignat}{2}
    &\max_{\phib} &\quad & \phib^{\T}\left(\dtmt^{\T}\dtmdmat\right)\phib   \notag\\
    &~\,\mathrm{s.t.}& & \langle\phib, \phib_d\rangle = 0, \quad\|\phib\| = 1,\label{eq:opt:phi}
  \end{alignat}  
  where $\phib_d$ is the $d$-th right singular vector of $\dtmt$. Since $\sigma_{d - 1} > 0 = \sigma_d$ and $d = |\cX| \leq |\cY|$, we know that 
\begin{align*}
  \phib_d = \left[\sqrt{P_X(1)}, \dots, \sqrt{P_X(d)}\right]^{\T} = \frac{1}{\sqrt{d}}\bone_d,
\end{align*}
and thus $\langle\phib, \phib_d\rangle = 0$ is equivalent to $\langle\phib, \bone_d\rangle = 0$.

Now, for all $\phib$ satisfying the constraints of \eqref{eq:opt:phi}, the objective function of \eqref{eq:opt:phi} is
  \begin{subequations}
    \begin{align}
        \phib^{\T}\left(\dtmt^{\T}\dtmdmat\right)\phib&=\sigma_1\psib^{\T}\dtmdmat\phib\\
      &= \frac{\sigma_1}{2}\sum_{x \in \cX, y\in \cY} \dtmet(y, x)\cdot \left[\frac{\phi'(x)}{\sqrt{P_X(x)}}- \frac{\sigma_1\psi'(y)}{\sqrt{P_Y(y)}}\right] \phi(x) \psi(y)%
        - \sigma_1\sum_{x \in \cX, y\in \cY}\sqrt{P_X(x)}\psi'(y)\phi(x) \psi(y) \label{eq:opt:phi:obj:2}\\
      &= \frac{\sigma_1}{2}\sum_{x \in \cX, y\in \cY} \dtmet(y, x)\cdot \left[\frac{\phi'(x)}{\sqrt{P_X(x)}}- \frac{\sigma_1\psi'(y)}{\sqrt{P_Y(y)}}\right] \phi(x) \psi(y)\label{eq:opt:phi:obj:3}\\
      &= \frac{\sigma^2_1}{2}\left[\sum_{x \in \cX} \frac{\phi'(x)}{\sqrt{P_X(x)}}\cdot \phi^2(x) - \sigma_1\sum_{y\in \cY}\frac{\psi'(y)}{\sqrt{P_Y(y)}} \cdot \psi^2(y)\right]\label{eq:opt:phi:obj:4}\\
      &= \frac{\sigma^2_1}{2}\Biggl[\frac{1}{\sqrt{P_X(1)}}\sum_{x \in \cX} \phi'(x) \phi^2(x)%
        - \frac{\sgn(p_1 - p_2)\cdot\sigma_1}{\sqrt{P_Y(1)}}\sum_{y\in [d]}\phi'(y) \phi^2(y)\Biggr]\label{eq:opt:phi:obj:5}\\
      &= \frac{\sigma^2_1}{2}\left[\frac{1}{\sqrt{P_X(1)}} - \frac{\sgn(p_1 - p_2)\cdot\sigma_1}{\sqrt{P_Y(1)}}\right]\sum_{x \in \cX} \phi'(x) \phi^2(x).\label{eq:opt:phi:obj:6}
    \end{align}
    \label{eq:opt:phi:obj}
  \end{subequations}
  where $\psib \defeq \sigma_1^{-1}\dtmt \phib$, and where to obtain \eqref{eq:opt:phi:obj:3} we have used the fact that $\langle\phib, \phib_d\rangle = 0$, to obtain \eqref{eq:opt:phi:obj:4} we have used the facts that $\dtmt \phib = \sigma_1 \psib$ and $\dtmt^{\T} \psib = \sigma_1 \phib$, and to obtain \eqref{eq:opt:phi:obj:5} we have used \lemref{lem:psi=phi} and the facts that
  \begin{gather*}
    P_X(1) = P_X(2) = \dots = P_X(d),\\
    P_Y(1) = P_Y(2) = \dots = P_Y(d).
  \end{gather*}
  Furthermore, to maximize \eqref{eq:opt:phi:obj:6}, note that
  \begin{align*}
    \sum_{x \in \cX} \phi'(x) \phi^2(x)
      = \sqrt{\frac{d - 1}{d}}\left[\phi^2(1) - \frac{1}{d - 1}\sum_{i = 2}^{d}\phi^2(i)\right]%
    &= \sqrt{\frac{d - 1}{d}}\left[\phi^2(1) - \frac{\|\phib\|^2 - \phi^2(1)}{d - 1}\right]\\
    &= \sqrt{\frac{d}{d - 1}}\phi^2(1) - \frac{1}{\sqrt{(d - 1)d}}.
  \end{align*}
  As a result, if follows from \lemref{lem:phi_1:max} that \eqref{eq:opt:phi:obj:6} is maximized when $\phib = \pm \phib'$, i.e., we have $\varphib_1 = \pm \phib'$, which finishes the proof.
\end{proof}

  \section{The Generalized ACE Algorithm \eqref{eq:gACE}} \label{sec:appf}
First, we define the $| \cX |$ and $| \cY |$ dimensional vectors $\bar{\phib}_i$ and  $\bar{\psib}_i$, respectively, for $i = 1, \ldots , k$ as
\begin{align*}
\phibar_i (x) = \sqrt{\Pt_X(x)} \fb_i(x), \quad \psibar_i (y) = \sqrt{\Pt_Y(y)} \gb_i (y),
\end{align*}
where $\Pt_X$ and $\Pt_Y$ are the marginal distributions of $\Pt_{XY}$, and $\fb_i$ and $\gb_i$ are the $i$-th dimension of $\fb$ and $\gb$, i.e.,
\begin{align*}
\fb(x) = \left[ \fb_1(x) \ \cdots \ \fb_k(x) \right]^{\T}, \quad \gb(y) = \left[ \gb_1(y) \ \cdots \ \gb_k(y) \right]^{\T},
\end{align*}
for all $x$ and $y$. Then, the iterative steps of the generalized ACE algorithm~\eqref{eq:gACE} can be equivalently expressed as
\begin{equation}
\begin{aligned} 
\text{i)}\quad \Phibb_k &\leftarrow \dtmtt^{\T} \Psibb_k \left( \Psibb_k^{\T} \Psibb_k \right)^{-1}  \\
\text{ii)}\quad \Psibb_k &\leftarrow \dtmtt \Phibb_k \left( \Phibb_k^{\T} \Phibb_k \right)^{-1}
\end{aligned}
\label{eq:ace_matrix}
\end{equation}
where 
\begin{align*} %
\Phibb_k = \left[ \phibb_1, \dots, \phibb_k \right], \quad  \Psibb_k = \left[\psibb_1, \dots, \psibb_k \right].
\end{align*}
Note that \eqref{eq:ace_matrix} coincides with %
 the alternating least squares algorithm \cite{koren2009matrix} for solving the low-rank approximation problem
\begin{align*}
  \min_{\Psibb_k, \Phibb_k}~\bbfrob{\dtmtt - \Psibb_k\Phibb_k^\T}^2.
\end{align*}
Then, using the same argument as that in \appref{sec:app:ace}, we know that the generalized ACE algorithm~\eqref{eq:gACE} essentially computes the singular vectors of $\dtmtt$ with respect to the top $k$ singular values. 

  \section{Proof of \lemref{lem:perturb:B:semi}}
  \label{sec:app:dtmtt}
  For any $\Pt_{XY} \in \nbhdbar(\eps)$ with the corresponding empirical distributions $\Ph_{XY}$ and $Q_{X}$, it follows from \eqref{eq:nbhd:bar} that
  \begin{align*}
    D\bigl(\Ph_{XY}\big\|P_{XY}\bigr) \leq \frac{\eps}{\alphah_k(r)}\quad\text{and}\quad D(Q_{X}\|P_{X}) \leq \frac{\eps}{r\alphah_k(r)}.
  \end{align*}
  Then, following the same argument as that for \eqref{eq:imate:pinsker}, we obtain
  \begin{align*}
    \max_{x\in \cX, y \in \cY} |\imatej(y, x)| \leq \sqrt{\frac{2}{p_{\min}\alphah_k(r)}}\quad\text{and}\quad \max_{x\in \cX} |\imatem(x)| \leq \sqrt{\frac{2}{rp_{\min}\alphah_k(r)}}.
  \end{align*}
  In addition, from \eqref{eq:imates} we have
  \begin{align}
    |\imates(y, x)|
    &\leq |\imatej(y, x)| + %
        \frac{r}{1 + r}|\imatem(x)| + |\cY| \cdot \max_{y' \in \cY}{|\imatej(y', x)|}\notag\\ 
    &\leq (|\cY| + 1)\cdot \max_{x \in \cX, y\in \cY}|\imatej(y, x)| + %
        \sqrt{r} \cdot \max_{x \in \cX}|\imatem(x)|\notag\\ %
    &\leq (|\cY| + 2)\sqrt{\frac{2}{p_{\min}\alphah_k(r)}},
  \end{align}
  where to obtain the second inequality we have used the fact that $\frac{r}{1 + r} \leq \frac{\sqrt{r}}{2} \leq \sqrt{r}$.

  Then, it follows from \eqref{eq:bound:Xi:3} that
  \begin{align}
    |\dtmdmatet(y, x)|
    &\leq \frac{1  +  |\cX| + |\cY|}{p^2_{\min}} \cdot  \max_{(x, y)\in \cX \times \cY}|\imates(y, x)|\\
    &\leq \frac{1  +  |\cX| + |\cY|}{p^2_{\min}} \cdot  (|\cY| + 2) \cdot \sqrt{\frac{2}{p_{\min}\alphah_k(r)}} \\
    &\leq \frac{(1  +  |\cX| + |\cY|)^2}{p^3_{\min}} \sqrt{\frac{2}{\alphah_k(r)}}.
  \end{align}
  Hence, we have $\frob{\dtmdmatt} \leq  \bar{C}$ with $\bar{C} \defeq \frac{(1  +  |\cX| + |\cY|)^2}{p^3_{\min}} \sqrt{\frac{2|\cX||\cY|}{\alphah_k(r)}}$.

  Turning now to the second part of the lemma, for the convenience of representation, in the following we use $\pv$ to replace $\sqrt{\eps}$.

  From \eqref{eq:Px:semi}, \eqref{eq:Q} and \eqref{eq:marg_pert_X}, we conclude
  \begin{align*}
    \Pt_X(x)
        = \frac{1}{r + 1}\Ph_X(x) + \frac{r}{r + 1}Q_X(x)
        = P_X(x) + \pv \sqrt{P_X(x)} \cdot \frac{\imate_X(x) + r \imatem(x)}{1 + r},
  \end{align*}
  where $\imatej_{X}$ and $\imatem$ are as defined in \eqref{eq:marg_pert_X} and \eqref{eq:Q}. In addition, it follows from \eqref{eq:imate:def} and \eqref{eq:marg_pert_X} that
  \begin{align*}
    \Ph_{Y|X}(y|x)
    &= \frac{\Ph_{XY}(x, y)}{\Ph_{X}(x)}%
        = P_{Y|X}(y|x) + \pv\sqrt{P_{Y|X}(y|x)}%
            \cdot\left[ \frac{\imatej(y, x)}{\sqrt{P_{X}(x)}} - \sqrt{P_{XY}(x, y)}\cdot\frac{\imatej_X(x)}{P_X(x)}\right] + o(\pv),
  \end{align*}
  where $\imatej$ is as defined in \eqref{eq:imate:def}.

  Therefore, we have%
  \begin{align*}
    \Pt_{XY}(x, y)
    &= \Pt_X(x) \Ph_{Y|X}(y|x)\\
    &= P_{XY}(x,y) + \pv\sqrt{P_{XY}(x, y)}\cdot\biggl( \imatej(y, x)%
        + \frac{r}{1 + r}\sqrt{P_{Y|X}(y|x)}\cdot\left[\imatem(x)-\imatej_X(x)\right]\biggr) + o(\pv)\\
    &= P_{XY}(x,y) + \pv\sqrt{P_{XY}(x, y)} \imates(y, x) + o(\pv).%
  \end{align*}
  Finally, it follows from \eqref{eq:marg_pert_X}--\eqref{eq:dtm:diff} that%
  \begin{align*}
    \dtmtt =  \dtmt + \tau \dtmdmatt + o(\tau) = \dtmt + \sqrt{\eps}\, \dtmdmatt + o\left(\sqrt{\eps}\right).
  \end{align*}

    \section{Proof of Lemma \ref{lem:exponent:kl:semi}}
  \label{sec:app:sanov:semi}
First, note that
  \begin{align}
      \mathbb{P}_{n, m} \left\{ \bbfrob{ \dtmt \Phib_k }^2 - \bbfrob{ \dtmt \Phibt_k }^2 > \eps  \right\} 
    &= \sum_{T(\Ph_{XY}), \ T(Q_{X})\colon \Pt_{XY}\in \cSt_1(\eps)} \P_{n, m}\bigl\{T(\Ph_{XY}), T(Q_{X})\bigr\}\notag\\
      &= \sum_{T(\Ph_{XY}), \ T(Q_{X})\colon \Pt_{XY}\in \cSt_1(\eps)} \P_n\bigl\{T(\Ph_{XY})\bigr\}\P_{m}\left\{T(Q_{X})\right\},
        \label{eq:P:semi}
  \end{align}
  where %
  $T(\Ph_{XY})$ and $T(Q_X)$ denote the type class of $\Ph_{XY}$ and the type class of $Q_X$, respectively, and the last equality follows from the fact that $Q_{X}$ is independent of $\Ph_{XY}$. Then, the probabilities of the two type classes are \cite{cover2012elements}
  \begin{gather*}
    \P_n\bigl\{T(\Ph_{XY})\bigr\} \doteq \exp\left\{-n D(\Ph_{XY}\|P_{XY})\right\}
  \end{gather*}
  and
  \begin{align*}
    \P_{m}\{T(Q_{X})\} &\doteq \exp\left\{-m D(Q_X\|P_{X})\right\}%
      = \exp\left\{-nr D(Q_X\|P_{X})\right\}.
  \end{align*}
  Moreover, for both type classes, the numbers of types are at most polynomial in $n$. Therefore, via the Laplace principle \cite{dembo2011large} it follows that
  \begin{align*}
    &\mathbb{P}_{n, m} \left\{ \bbfrob{ \dtmt \Phib_k }^2 - \bbfrob{ \dtmt \Phibt_k }^2 > \eps  \right\}%
      \doteq \exp \biggl\{ -n\cdot \inf_{\Pt_{XY} \in \cSt_1(\eps)} \Bigl[D(\Ph_{XY}\| P_{XY})+ r D(Q_{X}\| P_{X}) \Bigr] \biggr\}.
  \end{align*}

  \section{Proof of \lemref{lem:S1t:nbhd}}  
  \label{sec:app:S1t:nbhd}
  For $\eps > 0$ and $t > 0$, we define the subset $\cSt^{(t)}_2(\eps)$ of $\nbhdbar(\eps)$ as %
  \begin{align}
    \label{eq:cS2:def:semi}
    \cSt^{(t)}_2(\eps) \defeq \left\{\Pt_{XY} \colon \Pt_{XY} \in \nbhdbar(\eps), ~\sum_{i = 1}^k\sum_{j = k + 1}^d \frac{\left[\phib_i^{\T}\left(\dtmt^\T \dtmdmatt + \dtmdmatt^\T \dtmt\right) \phib_j\right]^2}{\sigma_i^2 - \sigma_j^2} \geq t \right\}, 
  \end{align}
  where $\dtmdmatt$ is as defined in \eqref{eq:theta:semi}. Then, it is convenient to first establish the following useful lemma.

\begin{lemma}
  \label{lem:S2t:alphah_k}
  For all $t \in (0, 2)$, we have
  \begin{align}
    -\lim_{\eps \rightarrow 0^+} \frac{1}{\eps} \inf_{\Pt_{XY} \in \cSt_2^{(t)}(\eps)} \left[D\bigl(\Ph_{XY}\big\| P_{XY}\bigr) + r D(Q_X \| P_X)\right] = \frac{t}{2\alphah_k(r)}.
    \label{eq:S2t:t}
  \end{align}
\end{lemma}
Using \lemref{lem:S2t:alphah_k}, we establish \lemref{lem:S1t:nbhd} as follows. First, for all $\Pt_{XY} \in \nbhdbar(\eps)$, it follows from \lemref{lem:perturb:B:semi} that
\begin{align*}
{\dtmtt}{}^{\T} \dtmtt = {\dtmt}^{\T} \dtmt + \eps \left( \dtmt^{\T} \dtmdmatt + \dtmdmatt^{\T} \dtmt \right) + o\left(\sqrt{\eps}\right).
\end{align*}
From the perturbation analysis result of \lemref{lem:eig:k}, we can represent the learning error as
  \begin{align}
    \bbfrob{ \dtmt \Phib_k }^2 - \bbfrob{ \dtmt \Phibt_k }^2
    = \eps\sum_{i = 1}^k\sum_{j = k + 1}^d \frac{\left[\phib_i^{\T}\left(\dtmt^\T \dtmdmatt + \dtmdmatt^\T \dtmt\right) \phib_j\right]^2}{\sigma_i^2 - \sigma_j^2} + o(\eps).
    \label{eq:err:perturb:semi}
  \end{align}
Therefore, for any $t \in (0, 1)$, there exists an $\eps_0 > 0$ such that for all $\eps \in (0, \eps_0)$, we have
  \begin{align*}
    \cSt_2^{(1 + t)}(\eps)  \subseteq \cSt_1(\eps) \cap \nbhdbar(\eps) \subseteq \cSt^{(1 - t)}_2(\eps).
  \end{align*}

  Then, using arguments similar to \eqref{eq:bound:min}--\eqref{eq:lim:cS1}, from \lemref{lem:S2t:alphah_k} we obtain
  \begin{align*}
    &-\lim_{\eps \rightarrow 0^+} \frac{1}{\eps} \inf_{\Pt_{XY} \in \cSt_1(\eps) \cap \nbhdbar(\eps)} \left[D\bigl(\Ph_{XY}\big\| P_{XY}\bigr) + r D(Q_X \| P_X)\right]\\
    &= -\lim_{\eps \rightarrow 0^+} \frac{1}{\eps} \inf_{\Pt_{XY} \in \cSt_2^{(1)}(\eps)} \left[D\bigl(\Ph_{XY}\big\| P_{XY}\bigr) + r D(Q_X \| P_X)\right] = \frac{1}{2\alphah_k(r)}.
  \end{align*}
  
  It remains only to establish \lemref{lem:S2t:alphah_k}.
    
  \begin{proof}[Proof of \lemref{lem:S2t:alphah_k}]
    Since the set $\cSt_2^{(t)}(\eps)$ is closed, we have
    \begin{align*}
      \inf_{\Pt_{XY} \in \cSt_2^{(t)}(\eps)} \left[D\bigl(\Ph_{XY}\big\| P_{XY}\bigr) + r D(Q_X \| P_X)\right] = \min_{\Pt_{XY} \in \cSt_2^{(t)}(\eps)} \left[D\bigl(\Ph_{XY}\big\| P_{XY}\bigr) + r D(Q_X \| P_X)\right].
    \end{align*}
    Then, for all $\Pt_{XY} \in \cSt_2^{(t)}$ with the corresponding empirical distributions $\Ph_{XY} \leftrightarrow \imatej$ and $Q_X \leftrightarrow \imatem$ for labeled and unlabeled data, it follows from the second-order Taylor series expansion of the K-L divergence that
    \begin{align}
      &D(\Ph_{XY}\| P_{XY}) + r D(Q_{X}\| P_{X})%
                                = \frac{\eps}{2} \left[ \|\imatj\|_{\F}^2 + r\|\imatm\|^2\right] + o(\eps),\label{eq:local:kl:semi}
    \end{align}
Therefore, %
the characterization of the error exponent~\eqref{eq:asym_sample_complexity:semi} can be reduced to the following optimization problem:
\begin{subequations}
  \begin{alignat}{2}
    &\min_{\imatj, \imatm} &\quad &\|\imatj\|_{\F}^2 + r \|\imatm\|^2\\
    &~\,\text{s.t.}& &\sum_{i = 1}^k\sum_{j = k + 1}^d \frac{\left[\phib_i^{\T}\left(\dtmt^\T \dtmdmatt + \dtmdmatt^\T \dtmt\right) \phib_j\right]^2}{\sigma_i^2 - \sigma_j^2}  \geq t,\\
    &              & &\sum_{x \in \cX}\sqrt{P_{X}(x)} \imatem (x) = 0, \\
    &              & &\sum_{x \in \cX, y \in \cY}\sqrt{P_{XY}(x,y)} \imatej (y, x) = 0,
  \end{alignat}
  \label{eq:opt:semi:org}
\end{subequations}
where the equality constraints follow from the definitions of $\imatej$ and $\imatem$. As we will verify, although not imposed as a constraint, the condition $\Pt_{XY} \in \nbhdbar(\eps)$ can be satisfied for the optimal $(\imatj, \imatm)$. Since both the objective function and the inequality constraint of \eqref{eq:opt:semi:org} are quadratic, the optimal solution can be obtained via solving
\begin{subequations}
  \begin{alignat}{2}
    &\max_{\imatj, \imatm} &\quad &\sum_{i = 1}^k\sum_{j = k + 1}^d \frac{\left[\phib_i^{\T}\left(\dtmt^\T \dtmdmatt + \dtmdmatt^\T \dtmt\right) \phib_j\right]^2}{\sigma_i^2 - \sigma_j^2}\\
    &~\,\text{s.t.}& & \|\imatj\|_{\F}^2 + r \|\imatm\|^2 \leq 1,\\
    &              & &\sum_{x \in \cX}\sqrt{P_{X}(x)} \imatem (x) = 0, \\
    &              & &\sum_{x \in \cX, y \in \cY}\sqrt{P_{XY}(x,y)} \imatej (y, x) = 0,
  \end{alignat}
  \label{eq:opt:semi}
\end{subequations}
where we have again interchanged the objective function and the quadratic function in the inequality constraint. Then, with arguments similar to those of the supervised case, we can verify the optimal solution of~\eqref{eq:opt:semi} also satisfies~\eqref{eq:imate:def} and~\eqref{eq:imates}. Furthermore, it can be verified that \eqref{eq:opt:semi} is equivalent to the optimization problem without the equality constraints, i.e.,
  \begin{subequations}
    \begin{alignat}{2}
      &\max_{\imatj, \imatm} &\quad &\sum_{i = 1}^k\sum_{j = k + 1}^d \frac{\left[\phib_i^{\T}\left(\dtmt^\T \dtmdmatt + \dtmdmatt^\T \dtmt\right) \phib_j\right]^2}{\sigma_i^2 - \sigma_j^2}\\
      &~\,\text{s.t.}& & \|\imatj\|_{\F}^2 + r \|\imatm\|^2 \leq 1.
    \end{alignat}
    \label{eq:opt:semi:equiv}
  \end{subequations}
  To see this, suppose $(\imatj^*, \imatm^*)$ is the optimal solution of \eqref{eq:opt:semi:equiv}, and define $c_1 \defeq \sum_{x \in \cX, y\in \cY}\imatej^*(y, x)\sqrt{P_{XY}(x, y)}$ and $c_2 \defeq \sum_{x \in \cX}\imatem^*(x)$. With $z_1(x, y) \defeq \imatej^*(y, x) - c_1\sqrt{P_{XY}(x, y)}$ and $z_2(x, y) \defeq \imatem^*(x) - c_2\sqrt{P_X(x)}$, we have
  \begin{align}
    1 &= \|\imatj^*\|_{\F}^2 + r \|\imatm^*\|^2%
      = \sum_{x \in \cX, y\in \cY}z_1^2(x, y) + r \sum_{x \in \cX}z_2^2(x) + (c_1^2 + rc_2^2),
  \end{align}
  which implies $c_1^2 + rc_2^2 \leq 1$.
  
  If $c_1^2 + r c_2^2 = 1$, then we have $z_1(x, y) \equiv 0$ and $z_2(x) \equiv 0$, which implies $\imatej^*(y, x) = c_1\sqrt{P_{XY}(x, y)}$ and $\imatem^*(x) = c_2\sqrt{P_X(x)}$. Therefore, it follows from %
  \eqref{eq:theta:semi} that
  \begin{align*}
    \dtmdmatet(y, x) = - \frac{c_1 + c_2 r}{1 + r} \sqrt{P_X(x)P_Y(y)}.
  \end{align*}
  which implies that $\dtmt^{\T}\dtmdmatt$ is a zero matrix. As a result, 
  the objective function of \eqref{eq:opt:semi} is zero, which contradicts the optimality of $(\imatj^*, \imatm^*)$. Moreover, if $c_1^2 + r c_2^2 < 1$, we can construct a feasible solution $(\imatj', \imatm')$ with 
  \begin{align*}
    \imatej'(y, x) = \frac{z_1(x, y)}{\sqrt{1 - c_1^2 - rc_2^2}}
  \quad\text{and}\quad
  \imatem'(x) = \frac{z_2(x)}{\sqrt{1 - c_1^2 - rc_2^2}},
  \end{align*}
  and it is straightforward to verify that the objective function for $(\imatj', \imatm')$ is $\left(1 - c_1^2 - rc_2^2\right)^{-1}$ times the value for $(\imatj^*, \imatm^*)$. This again contradicts the optimality of $(\imatj^*, \imatm^*)$. Therefore, we have $c_1 = c_2 = 0$, and the optimization problem \eqref{eq:opt:semi:equiv} has the same solution as \eqref{eq:opt:semi}.

To simplify the optimization problem \eqref{eq:opt:semi:equiv}, we define the vector $\convec \in \mathbb{R}^{|\cX|(|\cY| + 1)}$ as
\begin{align}\label{eq:convec:def}
  \convec \defeq
  \begin{bmatrix}
    \vec(\imatj)\\
    \sqrt{r}\imatm
  \end{bmatrix},
\end{align}
and let $\imats$ be the $|\cY| \times |\cX|$ matrix with the entries $\imates(y, x)$. Then, it follows from \eqref{eq:Lbh:def} and \eqref{eq:imates} that $\vec(\imats) = \Lbh(r)\convec$.

Therefore, the objective function of \eqref{eq:opt:semi:equiv} can be rewritten as
\begin{subequations}
  \begin{align}
      \sum_{i = 1}^k\sum_{j = k + 1}^d \frac{\left[\phib_i^{\T}\left(\dtmt^\T \dtmdmatt + \dtmdmatt^\T \dtmt\right) \phib_j\right]^2}{\sigma_i^2 - \sigma_j^2}%
    &=\vec^{\T}(\dtmdmatt)\left(\sum_{i = 1}^k\sum_{j = k + 1}^d \frac{\thetab_{ij}\thetab^{\T}_{ij}}{\sigma_i^2 - \sigma_j^2}\right) \vec(\dtmdmatt),\label{eq:gbh:1}\\
    &= \vec^{\T}(\imats)\Lb^{\T}\left(\sum_{i = 1}^k\sum_{j = k + 1}^d \frac{\thetab_{ij}\thetab^{\T}_{ij}}{\sigma_i^2 - \sigma_j^2}\right)\Lb \vec(\imats)\label{eq:gbh:2}\\
    &= \vec^{\T}(\imats)\Gb_k \vec(\imats)\label{eq:gbh:3}\\
    &= \convec^{\T}\Lbh^{\T}(r)\Gb_k\Lbh(r)\convec\label{eq:gbh:4}\\
    &= \convec^{\T}\Gbh_k(r)\convec,\label{eq:gbh:5}
  \end{align}
  \label{eq:gbh:convec}
\end{subequations}
where to obtain \eqref{eq:gbh:1} we have used \eqref{eq:tr:vec}, %
and to obtain \eqref{eq:gbh:3} we have used \eqref{eq:gb:k}. In addition, since
  $ \|\convec\|^2 = \|\imatj\|_{\F}^2 + r\|\imatm\|^2$,
  the constraint of \eqref{eq:opt:semi:equiv} can be rewritten as $\|\convec\| \leq 1$.

  As a result, the maximum of \eqref{eq:gbh:5} is the spectrum norm of $\Gbh_k(r)$, i.e., $\alphah_k(r)$, which is also the optimal value of the objective functions in \eqref{eq:opt:semi:equiv} and \eqref{eq:opt:semi}. This implies that the optimal solution of the original optimization problem \eqref{eq:opt:semi:org} is $\left(\sqrt{\frac{t}{\alphah_k(r)}} \imatj^*, \sqrt{\frac{t}{\alphah_k(r)}}\imatm^*\right) $, with the corresponding optimal value being $t/\alphah_k(r)$. Let $\Ph^*_{XY} \leftrightarrow \sqrt{\frac{t}{\alpha_k}}\imate^*$ and $Q^*_{X} \leftrightarrow \sqrt{\frac{t}{\alpha_k}}\imatem^*$ denote the corresponding empirical distributions, then we have, for $\eps$ sufficiently small,
  \begin{align*}
      D\bigl(\Ph^*_{XY} \big\| P_{XY}\bigr) + r D(Q^*_X \| P_X) = \frac{\eps t}{2\alphah_k(r)} + o(\eps) < \frac{\eps }{\alphah_k(r)},
  \end{align*}
  where to obtain the inequality we have used the fact that $t \in (0, 2)$.

  Hence, the corresponding optimal distribution $\Pt^*_{XY}$ as defined in \eqref{eq:PtXY} satisfies $\Pt^*_{XY} \in \nbhdbar(\eps)$. Therefore, we conclude
  \begin{align*}
    \min_{\Pt_{XY} \in \cSt_2^{(t)}(\eps)} \left[D\bigl(\Ph_{XY}\big\| P_{XY}\bigr) + r D(Q_X \| P_X)\right] = \frac{\eps t}{2\alphah_k(r)} + o(\eps), 
  \end{align*}
  which implies \eqref{eq:S2t:t}.

\end{proof}

  \section{Proof of \thmref{thm:sample:semi}}
  \label{sec:app:thm:non-asym:semi}
\reply{  
  First, it follows from \eqref{eq:lim:S1t} that there exists an $\bar{\eps}_0 > 0$ that depends only on $P_{XY}$ and $r$ such that for all $\eps \in (0, \bar{\eps}_0)$ we have
  \begin{align}
    \inf_{\Pt_{XY} \in \cSt_1(\eps)} \Bigl[D(\Ph_{XY}\| P_{XY})+ r D(Q_{X}\| P_{X}) \Bigr] = \frac{\eps}{2\alphah_k(r)} + o(\eps) > \frac{\eps}{3\alphah_k(r)} = \bar{\kappa}\eps,
    \label{eq:lb:D:semi}
  \end{align}
  where we have defined $\bar{\kappa} \defeq [{3}{\alphah_k(r)}]^{-1}$.

  Then,  for all $\eps \in (0, \bar{\eps}_0)$, it follows from \eqref{eq:P:semi} that
  \begin{align}
    &
      \mathbb{P}_{n, m} \left\{ \bbfrob{ \dtmt \Phib_k }^2 - \bbfrob{ \dtmt \Phibt_k }^2 > \eps  \right\} \notag \\
    &= \sum_{T(\Ph_{XY}), \ T(Q_{X})\colon \Pt_{XY}\in \cSt_1(\eps)} \P_n\bigl\{T(\Ph_{XY})\bigr\}\P_{m}\left\{T(Q_{X})\right\}\\
    &\leq \sum_{T(\Ph_{XY}), \ T(Q_{X})\colon \Pt_{XY}\in \cSt_1(\eps)} \exp\left\{-n D(\Ph_{XY}\|P_{XY}) - nr D(Q_X\|P_{X})\right\}\label{eq:sanov:ieq:semi:1}\\
    &\leq \sum_{T(\Ph_{XY}), \ T(Q_{X})\colon \Pt_{XY}\in \cSt_1(\eps)} \exp \biggl\{ -n\cdot \inf_{\Pt_{XY} \in \cSt_1(\eps)} \Bigl[D(\Ph_{XY}\| P_{XY})+ r D(Q_{X}\| P_{X}) \Bigr] \biggr\}\\
    &\leq (n + 1)^{|\cX||\cY|} (nr + 1)^{|\cX|} \exp \biggl\{ -n\cdot \inf_{\Pt_{XY} \in \cSt_1(\eps)} \Bigl[D(\Ph_{XY}\| P_{XY})+ r D(Q_{X}\| P_{X}) \Bigr] \biggr\}\label{eq:sanov:ieq:semi:2}\\
    &< (n + 1)^{|\cX||\cY|} (nr + 1)^{|\cX|} \exp\left(-\bar{\kappa} n\eps \right)\label{eq:sanov:ieq:semi:3}\\
    &\leq (2n)^{|\cX||\cY|} (n(r + 1))^{|\cX|} \exp\left(-\bar{\kappa} n\eps \right)\label{eq:sanov:ieq:semi:4}\\
    &= 2^{|\cX||\cY|} \cdot (r + 1)^{|\cX|} \cdot n^{|\cX|(1 + |\cY|)} \exp\left(-\bar{\kappa} n\eps \right)\\
    &= 2^{|\cX||\cY|} \cdot (r + 1)^{|\cX|} \cdot
      \left(\frac{4|\cX|(1 + |\cY|)}{\bar{\kappa}\eps}\right)^{|\cX|(1 + |\cY|)} \cdot
      \left(\frac{\bar{\kappa}n\eps}{4|\cX|(1 + |\cY|)}\right)^{|\cX|(1 + |\cY|)} \cdot \exp\left(-\bar{\kappa} n\eps \right)\\
    &< 2^{|\cX||\cY|} \cdot (r + 1)^{|\cX|} \cdot
      \left(\frac{4|\cX|(1 + |\cY|)}{\bar{\kappa}\eps}\right)^{|\cX|(1 + |\cY|)} \cdot\exp\left(-\frac{3\bar{\kappa}n\eps}{4} \right)\label{eq:sanov:ieq:semi:7}\\
    &< \left(\frac{8(1+r)|\cX|(1 + |\cY|)}{\bar{\kappa}\eps}\right)^{|\cX|(1 + |\cY|)} \cdot \exp\left(-\frac{3\bar{\kappa}n\eps}{4} \right)\\
    &\leq \left(\frac{12(1+r)|\cX| |\cY|}{\bar{\kappa}\eps}\right)^{|\cX|(1 + |\cY|)} \cdot \exp\left(-\frac{3\bar{\kappa}n\eps}{4} \right),\label{eq:sanov:ieq:semi:9}
  \end{align}
  where $T(\Ph_{XY})$ and $T(Q_X)$ denote the type class of $\Ph_{XY}$ and the type class of $Q_X$, respectively, and where \eqref{eq:sanov:ieq:semi:1} follows from the upper bound of probability of type classes \cite[Theorem 11.1.4]{cover2012elements}, where \eqref{eq:sanov:ieq:semi:2} follows from the upper bound of the number of types \cite[Theorem 11.1.1]{cover2012elements}. In addition, \eqref{eq:sanov:ieq:semi:3} follows from \eqref{eq:lb:D:semi}, \eqref{eq:sanov:ieq:semi:4} follows from $n \geq 1$,  \eqref{eq:sanov:ieq:semi:7} follows from the fact that $x \leq e^x - 1 < e^x$, and \eqref{eq:sanov:ieq:semi:9} follows from the fact that $1 + |\cY| \leq 3|\cY|/2$ since $|\cY| \geq 2$.

  Therefore, it suffices to choose $n$ such that
  \begin{align*}
    \left(\frac{12(1+r)|\cX| |\cY|}{\bar{\kappa}\eps}\right)^{|\cX|(1 + |\cY|)} \cdot \exp\left(-\frac{3\bar{\kappa}n\eps}{4} \right) < \delta,
  \end{align*}
  which is equivalent to 
  \begin{align*}
    n > \frac{4|\cX|(1+|\cY|)}{3\bar{\kappa}\eps} \log \frac{12(1+r)|\cX||\cY|}{\bar{\kappa}\eps} + \frac{4}{3\bar{\kappa}\eps}\log\frac{1}{\delta} &= \frac{4\alphah_k(r)|\cX|(1+|\cY|)}{\eps} \log \frac{36\alphah_k(r)(1+r)|\cX||\cY|}{\eps} + \frac{4\alphah_k(r)}{\eps}\log\frac{1}{\delta}\\
   &=  \bar{N}^{(4\alphah_k(r))}(\eps, \delta, r),
  \end{align*}
  where we have used the fact that $\bar{\kappa} = [3\alphah_k(r)]^{-1}$.
}  
\section{Proof of  \propref{prop:alphah}} \label{sec:apph}

First, we write the matrix $\Lbh(r)$ as defined in \eqref{eq:Lbh:def} as
    $\Lbh(r) = \left[\Lbh_1(r), \Lbh_2(r)\right]$,
  where $\Lbh_1(r)$ is composed of the first $(|\cX|\cdot|\cY|)$ columns of $\Lbh(r)$, and $\Lbh_2(r)$ is composed of the rest $|\cX|$ columns of $\Lbh$. Then it follows from the definition of $\Lbh(r)$ that
  \begin{subequations}
  \begin{align}
    \Lbh_1(r) = \Ib_{|\cX|\cdot|\cY|} - \frac{r}{1 + r} \Mb\Mb^{\T}
    \label{eq:Lbh:1}
  \end{align}
  and
  \begin{align}
    \Lbh_2(r) =
    \frac{\sqrt{r}}{1 + r} \Mb,
    \label{eq:Lbh:2}
  \end{align}
  \label{eq:Lbh:12}
\end{subequations}
  where $\Ib_{|\cX|\cdot|\cY|}$ is the identity matrix in $\mathbb{R}^{(|\cX|\cdot|\cY|)\times (|\cX|\cdot|\cY|)}$, and $\Mb$ is as defined in \eqref{eq:Mb:nub}.

  Therefore, we have
  \begin{subequations}
    \begin{align}
      \Lbh(r)\Lbh^{\T}(r)
      &= \Lbh_1(r)\Lbh_1^{\T}(r) + \Lbh_2(r)\Lbh_2^{\T}(r)\\
      &= \Ib_{|\cX|\cdot|\cY|} - \frac{2r}{1 + r} \Mb\Mb^{\T} + \frac{r^2}{(1+r)^2} \Mb\Mb^{\T}\Mb\Mb^{\T}%
        + \frac{r}{(1 + r)^2}\Mb\Mb^{\T}\\
      &= \Ib_{|\cX|\cdot|\cY|} - \frac{r}{1 + r} \Mb\Mb^{\T}\label{eq:Lbh:3}\\
      &= \Lbh_1(r),  
    \end{align}
  \end{subequations}
  where to obtain \eqref{eq:Lbh:3} we have exploited the fact that $\Mb^{\T}\Mb$ is the identity matrix in $\mathbb{R}^{|\cX|}$.

  Then, with $\spectral{\cdot}$ denoting the spectral norm, we have
  \begin{subequations}
  \begin{align}
    \alphah_k(r) &= \bbspectral{\Lbh^\T(r) \Gb_k \Lbh(r)} = \bbspectral{\left[\Gb_k^{\frac12}\Lbh(r)\right]^{\T} \Gb_k^{\frac12}\Lbh(r)}\label{eq:alphah:1}\\
    &= \bbspectral{ \Gb_k^{\frac12} \Lbh(r)\Lbh^\T(r)\Gb_k^{\frac12}}\label{eq:alphah:2}\\
    &= \bbspectral{ \Gb_k^{\frac12} \left(\Ib_{|\cX|\cdot|\cY|} - \frac{r}{1 + r}\Mb\Mb^\T\right)\Gb_k^{\frac12}},\label{eq:alphah:3}%
  \end{align}
\end{subequations}
where $\Gb_k^{\frac{1}{2}}$ is defined as the positive semidefinite matrix $\Cb$ such that $\Cb^2 = \Gb_k$, and where \eqref{eq:alphah:2} follows from the fact that for all matrices $\Ab$, we have
\begin{align}
  \bbspectral{\Ab\Ab^\T} = \bbspectral{\Ab^\T\Ab}.
  \label{eq:AAt:AtA}
\end{align}

Moreover, from $\Mb^{\T}\Mb = \Ib_{|\cX|\cdot|\cY|}$ we have
\begin{align}
  \Ib_{|\cX|\cdot|\cY|} - \frac{r}{1 + r}\Mb\Mb^\T = \left[\Pb(r)\right]^2, %
  \label{eq:t(r)}
\end{align}
where we have defined %
\begin{align*}
  \Pb(r) \defeq \Ib_{|\cX|\cdot|\cY|} - \left(1 - \frac{1}{\sqrt{1 + r}}\right)\Mb\Mb^\T.
\end{align*}
Then, it follows from \eqref{eq:alphah:3} and \eqref{eq:AAt:AtA}--\eqref{eq:t(r)} that
\begin{align*}
  \alphah_k(r) = \bbspectral{ \Pb(r)\Gb_k\Pb(r)}.
\end{align*}
Furthermore, for all $r_2 > r_1 \geq 0$, we define $\hat{\Pb}$ as
\begin{align*}
  \hat{\Pb} \defeq \Ib_{|\cX|\cdot|\cY|} - \left(1 - \sqrt{\frac{1+r_1}{1+r_2}}\right)\Mb\Mb^\T,
\end{align*}
then it can be verified that $\hat{\Pb}$ satisfies $\bspectral{\hat{\Pb}} = 1$ and
$\Pb(r_2) = \Pb(r_1) \hat{\Pb} = \hat{\Pb} \Pb(r_1)$. Hence, we have
\begin{align*}
  \alphah_k(r_2)
  &= \bbspectral{ \Pb(r_2)\Gb_k \Pb(r_2)}\\
  &= \left\| \hat{\Pb}\Pb(r_1)\Gb_k \Pb(r_1) \hat{\Pb} \right\|_{\mathrm{s}}\\
  &\leq  \bspectral{\hat{\Pb}}^2 \bbspectral{\Pb(r_1)\Gb_k \Pb(r_1)}\\
  &= \bbspectral{\Pb(r_1)\Gb_k \Pb(r_1)}%
      = \alphah_k(r_1),
\end{align*}
where the inequality follows from the submultiplicativity of the spectral norm \cite{horn2012matrix}. %

To prove the convexity of $\alphah_k(r)$, we first define the function $w(r) = \frac{r}{1 + r}$ for $r \geq 0$. Since $w(r)$ is an increasing and concave function of $r$, we have, for all $r_1, r_2 > 0$ and $\theta \in (0, 1)$, 
\begin{align*}
  w(\theta r_1 + (1 - \theta)r_2) \geq \theta w(r_1) + (1 - \theta)w(r_2),
\end{align*}
which implies that
\begin{align*}
  \theta r_1 + (1 - \theta) \geq w^{-1}\left(\theta w(r_1) + (1 - \theta)w(r_2)\right).
\end{align*}
Therefore, we have
\begin{align*}
    \alphah_k\left(\theta r_1 + (1 - \theta)\right)%
  &\leq \alphah_k\left(w^{-1}\left(\theta w(r_1) + (1 - \theta)w(r_2)\right)\right)\\
  &= \bbspectral{ \Gb_k^{\frac12} \left[\Ib_{|\cX|\cdot|\cY|} - \left(\theta w(r_1) + (1 - \theta)w(r_2)\right)\Mb\Mb^\T\right]\Gb_k^{\frac12}}\\
  &= \left\| \Gb_k^{\frac12} \left[\theta\left(\Ib_{|\cX|\cdot|\cY|} - w(r_1)\Mb\Mb^\T\right)%
    + (1 - \theta)\left(\Ib_{|\cX|\cdot|\cY|} - w(r_2)\Mb\Mb^\T\right)\right]\Gb_k^{\frac12}\right\|_{\mathrm{s}}\\
  &\leq \theta\bbspectral{ \Gb_k^{\frac12} \left(\Ib_{|\cX|\cdot|\cY|} - w(r_1)\Mb\Mb^\T\right)\Gb_k^{\frac12}}%
    + (1-\theta)\bbspectral{ \Gb_k^{\frac12} \left(\Ib_{|\cX|\cdot|\cY|} - w(r_2)\Mb\Mb^\T\right)\Gb_k^{\frac12}}\\
  &=  \theta \alphah_k(r_1) + (1-\theta)\alphah_k(r_2),
\end{align*}
where the first equality follows from the fact that $\alphah_k(r)$ is non-increasing, and the second equality follows from the triangle inequality for the spectral norm.

Finally, to obtain the lower bound of \eqref{eq:alphah:bound}, note that
\begin{subequations}
  \begin{align}
    \alphah_k(r)
    &= \bbspectral{ \Gb_k^{\frac12} \left(\Ib_{|\cX|\cdot|\cY|} - \frac{r}{1 + r}\Mb\Mb^\T\right)\Gb_k^{\frac12}}\label{eq:alphah:lb:1}\\
    &\geq \bbspectral{ \Gb_k} - \frac{r}{1 + r}\bbspectral{\Gb_k^{\frac12}\Mb\Mb^\T\Gb_k^{\frac12}}\label{eq:alphah:lb:2}\\
    &= \bbspectral{ \Gb_k} - \frac{r}{1 + r}\bbspectral{\Mb^\T\Gb_k\Mb}\label{eq:alphah:lb:3}\\
    &\geq \bbspectral{ \Gb_k} - \frac{r}{1 + r}\bbspectral{\Gb_k}\bbspectral{\Mb}^2\label{eq:alphah:lb:4}\\
    &= \frac{1}{1 + r} \bbspectral{\Gb_k}= \frac{1}{1 + r} \alphah_k(0),\label{eq:alphah:lb:5}
  \end{align}
\end{subequations}
where \eqref{eq:alphah:lb:2} follows from the triangle inequality, \eqref{eq:alphah:lb:3} follows from \eqref{eq:AAt:AtA}, \eqref{eq:alphah:lb:4} follows from the submultiplicativity of the spectral norm, and the penultimate equality follows from the fact that $\spectral{\Mb} = \sqrt{\bbspectral{\Mb^{\T}\Mb}} = 1$, since $\Mb^{\T}\Mb$ is an identity matrix.

To obtain the upper bound of \eqref{eq:alphah:bound}, note that
\begin{align*}
    \alphah_k(r)%
  &= \bbspectral{ \Gb_k^{\frac12} \left(\Ib_{|\cX|\cdot|\cY|} - \frac{r}{1 + r}\Mb\Mb^\T\right)\Gb_k^{\frac12}}\\
  &= \left\|\frac{1}{1 + r}\Gb_k + \frac{r}{1 + r}\Gb_k^{\frac12}\left(\Ib_{|\cX|\cdot|\cY|} - \Mb\Mb^\T\right)\Gb_k^{\frac12}\right\|_{\mathrm{s}}\\
  &\leq \frac{1}{1 + r} \bbspectral{\Gb_k} + \frac{r}{1 + r}\bbspectral{\Gb_k^{\frac12}\left(\Ib_{|\cX|\cdot|\cY|} - \Mb\Mb^\T\right)\Gb_k^{\frac12}}\\
  &= \frac{1}{1 + r}\alphah_k(0) + \frac{1}{1 + r}\alphah_k(\infty),
\end{align*}
where we have again used the triangle inequality.

\section{Proof of \propref{eg:1:semi}} \label{sec:app:eg1:semi}
  First, note that from \eqref{eq:gbh:def} we have
\begin{align}
  \Gbh_k(r) &= \Lbh^{\T}(r)\Gb_k\Lbh(r)%
    = \frac14\sum_{i = 1}^{d - 1}\sigma_i^2 \left(\Lbh^{\T}(r)\Mb\phib_i\right)\left(\Lbh^{\T}(r)\Mb\phib_i\right)^{\T}%
    = \frac{1}{4(1+r)}\sum_{i = 1}^{d - 1}\sigma_i^2 \Bigl(\Mbh(r)\phib_i\Bigr)\Bigl(\Mbh(r)\phib_i\Bigr)^{\T},\label{eq:gbh:eig}
\end{align}
where the second equality follows from \eqref{eq:gbk:eig}, and in the last equality we have defined
\begin{align}\label{eq:Mbh:def}
  \Mbh(r) \defeq \sqrt{1 + r}\cdot\Lbh^{\T}(r)\Mb.
\end{align}
 In addition, note that $\Mbh(r)$ satisfies
\begin{align}
  \label{eq:Mbh}
  \Mbh^{\T}(r)\Mbh(r)
    = (1 + r)\Mb^{\T}\Lbh(r)\Lbh^{\T}(r)\Mb%
    = (1 + r)\Mb^{\T}\left(\Ib_{|\cX|\cdot|\cY|} - \frac{r}{1 + r}\Mb\Mb^{\T}\right)\Mb%
  &= \Mb^{\T}\Mb= \Ib_d,
\end{align}
where to obtain the second equality we have used \eqref{eq:Lbh:3}. Therefore, we have $\left\langle\Mbh(r)\phib_i, \Mbh(r)\phib_j\right\rangle = \langle\phib_i, \phib_j\rangle = \delta_{ij}$, and it follows from \eqref{eq:gbh:eig} that the non-zero eigenvalues of $\Gbh_k(r)$ are 
\begin{align*}
  \frac{\sigma_i^2}{4(1+r)},\quad i = 1, \dots, d - 1.
\end{align*}
Hence, the largest eigenvalue (i.e., the largest singular value) of $\Gbh_k(r)$ is 
\begin{align*}
   \alphah_k(r) = \bspectral{\Gbh_k(r)} = \frac{\sigma_1^2}{4(1 + r)}.
\end{align*}
\section{Proof of \thmref{thm:exponent:semi:k:k+1}}
\label{sec:app:thm:semi:k:k+1}

Similar to the proof of \thmref{thm:exponent:k:k+1}, we first extend the definition of $\nbhdbar(\eps)$ to the case $\sigma_k = \sigma_{k + 1}$ via letting
  \begin{align}
    \nbhdbar(\eps) \defeq \left\{\Pt_{XY}\colon
    D(\Ph_{XY}\| P_{XY})+ r D(Q_{X}\| P_{X}) \leq \frac{\eps}{\betah_k(r)}\right\},
    \label{eq:nbhd:bar:=}
  \end{align}
  and define $\cSt^{(t)}_3(\eps)$ as the set of $\Pt_{XY}$ such that the corresponding $\imates$ from~\eqref{eq:imates} satisfies
  \begin{align} 
    &\sum_{i = 1}^{l - 1}\sum_{j = l}^d \frac{\left[\phib_i^{\T}\left(\dtmt^\T \dtmdmatt + \dtmdmatt^\T \dtmt\right) \phib_j\right]^2}{\sigma_i^2 - \sigma_j^2}%
                + \sum_{i = l}^{k}\sum_{j \in \cIb_k} \frac{\left[\varphibb_i^{\T}\left(\dtmt^\T \dtmdmatt + \dtmdmatt^\T \dtmt\right) \phib_j\right]^2}{\sigma_i^2 - \sigma_j^2}
                \geq 1,\label{eq:condition_k=k+1:semi}
  \end{align} 
  where $\varphibb_i$ are as defined in \eqref{eq:varphi:def:semi}. Then the following result, analogous to \lemref{lem:S2t:alphah_k} for the case $\sigma_k > \sigma_{k + 1}$, will be useful in our analysis.
  \begin{lemma}
    \label{lem:S3t:betah_k}
    For all $t \in (0, 2)$, we have
    \begin{align}
      -\lim_{\eps \rightarrow 0^+} \frac{1}{\eps} \inf_{\Pt_{XY} \in \cSt_3^{(t)}(\eps)} \left[D\bigl(\Ph_{XY}\big\| P_{XY}\bigr) + r D(Q_X \| P_X)\right] = \frac{t}{2\betah_k(r)}.
      \label{eq:S3t:=}
    \end{align}
  \end{lemma}
  \begin{proof}
    The proof is similar to that of \lemref{lem:S3:beta_k}. Using the second-order Taylor series expansion of the K-L divergence \eqref{eq:local:kl:semi}, the limit \eqref{eq:S3t:=} can be characterized by the following optimization problem:
    \begin{subequations}
      \begin{alignat}{2} %
        &\min_{\imatj, \imatm} & \quad & \|\imatj\|_{\F}^2 + r \|\imatm\|^2
         \\
        &~\,\text{s.t.} && \sum_{i = 1}^{l - 1}\sum_{j = l}^d \frac{\left[\phib_i^{\T}\left(\dtmt^\T \dtmdmatt + \dtmdmatt^\T \dtmt\right) \phib_j\right]^2}{\sigma_i^2 - \sigma_j^2}%
        + \sum_{i = l}^{k}\sum_{j \in \cIb_k} \frac{\left[\varphibb_i^{\T}\left(\dtmt^\T \dtmdmatt + \dtmdmatt^\T \dtmt\right) \phib_j\right]^2}{\sigma_i^2 - \sigma_j^2} \geq t,\\
        &&& \sum_{x \in \cX}\sqrt{P_{X}(x)} \imatem (x) = 0,\quad \sum_{x \in \cX, y \in \cY}\sqrt{P_{XY}(x,y)} \imatej (y, x) = 0.
      \end{alignat}
      \label{eq:opt:semi:k:k+1}
    \end{subequations}

    Following the same argument as that for \lemref{lem:S2t:alphah_k}, the optimal solution of \eqref{eq:opt:semi:k:k+1} can be obtained by solving
    \begin{subequations}
      \begin{alignat}{2} %
        &\max_{\imatj, \imatm} & \quad &\sum_{i = 1}^{l - 1}\sum_{j = l}^d \frac{\left[\phib_i^{\T}\left(\dtmt^\T \dtmdmatt + \dtmdmatt^\T \dtmt\right) \phib_j\right]^2}{\sigma_i^2 - \sigma_j^2}%
        + \sum_{i = l}^{k}\sum_{j \in \cIb_k} \frac{\left[\varphibb_i^{\T}\left(\dtmt^\T \dtmdmatt + \dtmdmatt^\T \dtmt\right) \phib_j\right]^2}{\sigma_i^2 - \sigma_j^2} \\
        &~\,\text{s.t.} && \|\imatj\|_{\F}^2 + r \|\imatm\|^2 \leq 1,
      \end{alignat}
      \label{eq:opt:semi:k:k+1:eq}
    \end{subequations}
    where we have interchanged the objective function and the quadratic function in the inequality constraint, and removed the equality constraints.

    In addition, similar to \eqref{eq:gbh:convec}, we can rewrite the objective function of \eqref{eq:opt:semi:k:k+1:eq} as
    \begin{align*} \notag
      &\sum_{i = 1}^{l - 1}\sum_{j = l}^d \frac{\left[\phib_i^{\T}\left(\dtmt^\T \dtmdmatt + \dtmdmatt^\T \dtmt\right) \phib_j\right]^2}{\sigma_i^2 - \sigma_j^2}%
          + \sum_{i = l}^{k}\sum_{j \in \cIb_k} \frac{\left[\varphibb_i^{\T}\left(\dtmt^\T \dtmdmatt + \dtmdmatt^\T \dtmt\right) \phib_j\right]^2}{\sigma_i^2 - \sigma_j^2} 
          = \convec^{\T}\Jbh_k(r, \imat, \imatm)\convec,
    \end{align*}
    where $\Jbh_k(r, \imat, \imatm)$ is as defined in \eqref{eq:Jbh}. As a result, the optimization problem \eqref{eq:opt:semi:k:k+1:eq} can be rewritten as \eqref{eq:opt:semi:k:k+1:simple}, and thus the optimal value is $\betah_k(r)$. Finally, using the same argument as that for \lemref{lem:S2t:alphah_k}, we conclude that the optimal value of \eqref{eq:opt:semi:k:k+1:simple} is $t/\betah_k(r)$ and thus
    \begin{align*}
      \inf_{\Pt_{XY} \in \cSt_3^{(t)}(\eps)} \left[D\bigl(\Ph_{XY}\big\| P_{XY}\bigr) + r D(Q_X \| P_X)\right] = \frac{\eps t}{2\betah_k(r)} + o(\eps),
    \end{align*}
    which implies \eqref{eq:S3t:=}.    
  \end{proof}
  
  In addition, it follows from \lemref{lem:2} and \lemref{lem:perturb:B:semi} that the corresponding learning error for the distribution $\Pt_{XY} \in \nbhd(\eps)$ is
  \begin{align}
    \bbfrob{ \dtmt \Phib_k }^2 - \bbfrob{ \dtmt \Phibt_k }^2
    = \eps\sum_{i = 1}^{l - 1}\sum_{j = l}^d \frac{\left[\phib_i^{\T}\left(\dtmt^\T \dtmdmatt + \dtmdmatt^\T \dtmt\right) \phib_j\right]^2}{\sigma_i^2 - \sigma_j^2}%
    + \eps\sum_{i = l}^{k}\sum_{j \in \cIb_k} \frac{\left[\varphibb_i^{\T}\left(\dtmt^\T \dtmdmatt + \dtmdmatt^\T \dtmt\right) \phib_j\right]^2}{\sigma_i^2 - \sigma_j^2} + o(\eps).
    \label{eq:err:perturb:=:semi}
  \end{align}

  Therefore, for any $t \in (0, 1)$, there exists an $\eps_0 > 0$ such that for all $\eps \in (0, \eps_0)$, we have
  \begin{align*}
    \cSt_3^{(1 + t)}(\eps)  \subseteq \cSt_1(\eps) \cap \nbhdbar(\eps) \subseteq \cSt_3^{(1 - t)}(\eps).
  \end{align*}
  
  Then, using arguments similar to \eqref{eq:bound:min}--\eqref{eq:lim:cS1}, from \lemref{lem:S3t:betah_k} we have
    \begin{align*}
    &-\lim_{\eps \rightarrow 0^+} \frac{1}{\eps} \inf_{\Pt_{XY} \in \cSt_1(\eps) \cap \nbhdbar(\eps)} \left[D\bigl(\Ph_{XY}\big\| P_{XY}\bigr) + r D(Q_X \| P_X)\right]\\
    &= -\lim_{\eps \rightarrow 0^+} \frac{1}{\eps} \inf_{\Pt_{XY} \in \cSt_3^{(1)}(\eps)} \left[D\bigl(\Ph_{XY}\big\| P_{XY}\bigr) + r D(Q_X \| P_X)\right] = \frac{1}{2\betah_k(r)}.
  \end{align*}
  Finally, following the same proof as that for \thmref{thm:exponent:semi}, we obtain \eqref{eq:exponent:semi:k:k+1}.

  \section{Proof of Corollary \ref{eg:2:semi}}  
  \label{sec:app:eg2:semi}

  From \lemref{lem:psi=phi}, we have $\sigma_1 = \dots = \sigma_{d - 1} > \sigma_d = 0$. Therefore, for all $1 \leq k \leq d - 1$ we have $\cI_k = [d - 1]$, which further implies that
  \begin{align*}
    l = \min \cI_k = 1\quad\text{and}\quad\cIb_k = \{d\}.
  \end{align*}
  Hence, from \eqref{eq:Jbh} we have
  \begin{align}
    \Jbh_k(r, \imatj, \imatm)& = \Lbh^{\T}(r)\Lb^{\T}\left(\sum_{i = 1}^k \frac{\varthetabb_{id}\varthetabb^{\T}_{id}}{\sigma_1^2}\right)\Lb\Lbh(r)%
      = \frac{1}{\sigma_1^2}\sum_{i = 1}^k \left(\Lbh^{\T}(r)\Lb^{\T} \varthetabb_{id}\right)\left(\Lbh^{\T}(r)\Lb^{\T} \varthetabb_{id}\right)^{\T}.
  \end{align}  
  In addition, similar to \eqref{eq:comp:alpha}, we have
  \begin{align}
    \Lb^{\T} \varthetabb_{id} = -\frac{\sigma_1^2}{2}\Mb \varphibb_i,
  \end{align}
  and thus
  \begin{align}\label{eq:gbh:eig:semi}
    \Gbh_k
    &= \frac{\sigma_1^2}{4} \sum_{i = 1}^{k}\left(\Lbh^{\T}(r)\Mb \varphibb_i\right)\left(\Lbh^{\T}(r)\Mb \varphibb_i\right)^{\T}%
      = \frac{1}{4(1+r)}\sum_{i = 1}^{d - 1}\sigma_i^2 \Bigl(\Mbh(r)\varphibb_i\Bigr)\Bigl(\Mbh(r)\varphibb_i\Bigr)^{\T},
  \end{align}
  where $\Mbh(r)$ is as defined in \eqref{eq:Mbh:def}. Note that since $\left\langle\Mbh(r) \varphibb_i, \Mbh(r) \varphibb_j\right\rangle = \delta_{ij}$, \eqref{eq:gbh:eig:semi} demonstrates the eigen-decomposition of $\Gbh_k$. Therefore, from \thmref{thm:exponent:semi:k:k+1} and the definition of $\betah(r)$, we have
  \begin{align*}
    \betah_k(r) \leq \bbspectral{\Jbh_k(r, \imatj, \imatm)} = \frac{\sigma_1^4}{4(1 + r)}.
  \end{align*}
  To prove that the inequality holds with equality, it suffices to construct $\imat$ and $\imatm$ such that the corresponding $\convec$ as defined in \eqref{eq:def:convec} satisfies $\|\convec\|^2 \leq 1$ and
  \begin{align}\label{eq:alphah:eq}
    \convec^{\T}\Gbh_k\convec = \bbspectral{\Jbh_k(r, \imatj, \imatm)}.
  \end{align}
  Indeed, as we now illustrate, if $\imat$ and $\imatm$ are chosen as
  \begin{subequations}
    \begin{align}\label{eq:imat:opt:semi}
      \imate(y, x) = \frac{1}{\sqrt{1 + r}}\sqrt{P_{Y|X}(y|x)}\phi'(x)
    \end{align}
    and
    \begin{align}\label{eq:imatm:opt:semi}
      \imatem(x) = \frac{1}{\sqrt{1 + r}}\phi'(x)
    \end{align}
    \label{eq:imat:m:opt:semi}
  \end{subequations}
  with $\phib'$ as defined in \eqref{eq:phi'}, then we have $\varphibb_1 = \pm\phib'$ and $\convec = \Mbh(r)\phib'= \pm\Mbh(r)\varphibb_1$, and thus \eqref{eq:alphah:eq} holds.

  To see this, first note that from \eqref{eq:Lbh:12} we have
  \begin{align*}
    \Mbh(r) = \sqrt{1 + r}\Lbh^{\T}(r)\Mb =
    \sqrt{1 + r}
    \begin{bmatrix}
      \Lbh_1^{\T}(r)\Mb\\
      \Lbh_2^{\T}(r)\Mb
    \end{bmatrix}
    =\frac{1}{\sqrt{1+r}}
    \begin{bmatrix}
      \Mb\\
      \sqrt{r}\Ib_{d}
    \end{bmatrix},  
  \end{align*}
  and it follows from \eqref{eq:def:convec} and \eqref{eq:imat:m:opt:semi} that $\convec = \Mbh(r)\phib'$. Therefore, we have $\|\convec\|^2 = \|\phib'\|^2 = 1$. 

In addition, from \eqref{eq:imates} we have
\begin{align*}
  \imates(y, x) 
  &= \imatej(y, x) + \frac{r}{1 + r}\sqrt{P_{Y|X}(y|x)}%
    \cdot\biggl[\imatem(x) - \sum_{y'}\sqrt{P_{Y|X}(y'|x)}\imatej(y', x)\biggr]\\
  &= \imatej(y, x)%
    = \frac{1}{\sqrt{1 + r}}\sqrt{P_{Y|X}(y|x)}\phi'(x),
\end{align*}
i.e.,
\begin{align*}
  \vec(\imats) = \frac{1}{\sqrt{1 + r}} \Mb\phib'.
\end{align*}
Then, similar to \eqref{eq:theta:eg}, from \eqref{eq:theta:semi} we obtain
\begin{align*}
  \dtmdmatet(y, x) = \frac{1}{\sqrt{1 + r}} \dtmdmate(y, x)
\end{align*}
with $\dtmdmate(y, x)$ as given by \eqref{eq:theta:eg}.
Furthermore, following the same proof as that for Corollary \ref{eg:2}, $\varphibb_1$ is the solution of the optimization problem
  \begin{alignat}{2}
    &\max_{\phib} &\quad & \phib^{\T}\left(\dtmt^{\T}\dtmdmatt\right)\phib   \notag\\
    &~\,\mathrm{s.t.}& & \langle\phib, \phib_d\rangle = 0, \quad\|\phib\| = 1,\label{eq:opt:phi:semi}
  \end{alignat}
  which has the same solution as the optimization problem \eqref{eq:opt:phi} since $\dtmdmatt = \dtmdmat / \sqrt{1 + r}$. Hence, we obtain $\varphibb_1 = \pm\phib’$, which finishes the proof.
  
\bibliographystyle{IEEEtran}
\bibliography{ref}

\end{document}